\documentclass[12pt, draftclsnofoot, onecolumn]{IEEEtran}
\IEEEoverridecommandlockouts
\usepackage{epsfig,graphicx,psfrag,amsmath,cases,bm}
\usepackage{latexsym,amssymb,algorithm,mathtools}
\usepackage{algorithmic}
\usepackage{color}
\usepackage{url}
\usepackage{scrtime}
\usepackage{stfloats}
\usepackage{tablefootnote}
\usepackage{cite}
\usepackage{amsfonts,mathabx}
\usepackage{textcomp}
\usepackage{xcolor,capt-of}
\usepackage{multirow}
\usepackage{amsthm}
\usepackage{geometry}
\ifCLASSOPTIONcompsoc
    \usepackage[caption=false, font=normalsize, labelfont=sf, textfont=sf]{subfig}
\else
\usepackage[caption=false, font=footnotesize]{subfig}
\fi
\allowdisplaybreaks
 \def\BibTeX{{\rm B\kern-.05em{\sc i\kern-.025em b}\kern-.08em
    T\kern-.1667em\lower.7ex\hbox{E}\kern-.125emX}}
\makeatletter
\newcommand*{\rom}[1]{\expandafter\@slowromancap\romannumeral #1@}
\makeatother

\DeclareMathOperator{\mino}{minimize}
\makeatother
\newtheorem{remark}{Remark}
\newtheorem{lemma}{Lemma}
\newtheorem{theorem}{Theorem}
\newtheorem{proposition}{Proposition}

\IEEEaftertitletext{\vspace{-12mm}}
\title{Globally Optimal Resource Allocation Design for Discrete Phase Shift IRS-Assisted Multiuser Networks with Perfect and Imperfect CSI\thanks{This work has been presented in part at the IEEE International Communication Conference Workshops, Rome, Italy, June 2023 \cite{wu2023globally}.}}
\author{\IEEEauthorblockN {Yifei Wu\IEEEauthorrefmark{1}, Dongfang Xu\IEEEauthorrefmark{2}, Derrick Wing Kwan Ng\IEEEauthorrefmark{3}, Robert Schober\IEEEauthorrefmark{1}, and Wolfgang Gerstacker\IEEEauthorrefmark{1}}

\IEEEauthorrefmark {1}Friedrich-Alexander-Universit\"at
Erlangen-N\"urnberg, Germany;
\IEEEauthorrefmark {3}The University
of New South Wales, Australia\\
\IEEEauthorrefmark {2}The Hong Kong University of Science and Technology, Hong Kong

}

\begin{document}
\maketitle
\begin{abstract}
    Intelligent reflecting surfaces (IRSs) are a promising low-cost solution for achieving high spectral and energy efficiency in future communication systems by enabling the customization of wireless propagation environments. Despite the plethora of research on resource allocation design for IRS-assisted multiuser wireless communication systems, the optimal design and the corresponding performance upper bound are still not fully understood. To bridge this gap in knowledge, in this paper, we investigate the optimal resource allocation design for IRS-assisted multiuser multiple-input single-output (MISO) systems employing practical discrete IRS phase shifters. In particular, we jointly optimize the beamforming vector at the base station (BS) and the discrete IRS phase shifts to minimize the total transmit power for the cases of perfect and imperfect channel state information (CSI) knowledge. To this end, two novel algorithms based on the generalized Benders decomposition (GBD) method are developed to obtain the globally optimal solution for perfect and imperfect CSI, respectively. Moreover, to facilitate practical implementation, we propose two corresponding low-complexity suboptimal algorithms with guaranteed convergence by capitalizing on successive convex approximation (SCA). In particular, for imperfect CSI, we adopt a bounded error model to characterize the CSI uncertainty and propose a new transformation to convexify the robust quality-of-service (QoS) constraints. Our numerical results confirm the optimality of the proposed GBD-based algorithms for the considered system for both perfect and imperfect CSI. 
    {\color{black}Furthermore, we unveil that both proposed SCA-based algorithms can attain a locally optimal solution within a few iterations. Moreover, compared with the state-of-the-art solution based on alternating optimization (AO), the proposed low-complexity SCA-based schemes achieve a significant performance gain, especially for moderate-to-large numbers of IRS elements.}

\end{abstract}
\vspace*{-4mm}
\section{Introduction}

Recently, intelligent reflecting surface (IRS)-assisted communication has emerged as a promising approach to satisfy the escalating demand for high spectral and energy efficiency in future wireless systems. Specifically, comprising cost-effective passive and programmable elements, IRSs possess the capability to intelligently establish favorable wireless propagation environments such that the signal paths between various communication nodes can be tailored \cite{wu2023globally,wu2019intelligent, huang2019reconfigurable, wu2019beamforming, xu2020resource, wu2021intelligent}. By leveraging this appealing property, one can concentrate the energy of the transmitted signal in desired directions via beamforming, which facilitates power savings in wireless systems. Moreover, typically fabricated as thin rectangular planes, IRSs can be conveniently attached to building facades, indoor ceilings, and vehicles, thereby enabling seamless integration of IRSs into existing wireless systems\cite{wu2021intelligent,yu2021robust}. Driven by these advanced features, numerous works have explored the application of IRSs in wireless systems, aiming to enhance, e.g., multiple-input multiple-output (MIMO) transmission \cite{xu2021resource}, physical layer security \cite{yu2020robust}, and simultaneous wireless information and power transfer (SWIPT) \cite{xu2022optimal}.

To fully exploit the vast potential of IRSs, both the phase shift configuration of the IRS and the transmit beamforming at the base station (BS) have to be delicately designed \cite{wu2019beamforming,wu2021intelligent,9014322,xu2022optimal}. In light of this, several works have developed joint BS beamforming and IRS reflection coefficient design policies for the minimization of the total transmit power while guaranteeing the quality-of-service (QoS) of the communication users. Assuming perfect channel state information (CSI) knowledge, transmit power minimization problems were investigated for different IRS-assisted wireless systems, including multiuser MIMO systems \cite{wu2019intelligent}, secure wireless systems \cite{9014322}, and SWIPT systems \cite{xu2022optimal}. However, the significant power savings shown in \cite{wu2019intelligent,xu2022optimal}, and \cite{9014322} rely on the assumption of continuous IRS phase shifts, which may not be applicable in practical IRS systems.
In practice, due to various challenges, such as high energy consumption, high integration complexity, and intractable component coupling, the phase shifts generated by each IRS element are generally confined to discrete values and low resolution. In fact, practical large-scale IRSs typically utilize 1-bit on-off phase shifters or 2-bit quadruple-level phase shifters \cite{wu2021intelligent,wu2019beamforming,shi2022multiuser}. To account for this limitation, several early studies have investigated the resource allocation design for discrete IRS phase shifts. For example, in \cite{shi2022multiuser}, a suboptimal alternating optimization (AO)-based algorithm was developed to promote user fairness in IRS-assisted multiuser systems with discrete IRS phase shifts. However, as shown in \cite{9014322, xu2021resource}, for the considered non-convex optimization problems, iterative AO-based algorithms unavoidably compromise the optimality of the solution. In particular, it is well-known that these algorithms may get trapped in a locally optimal solution \cite{9014322,xu2021resource,bezdek2002some}, and their performance depends heavily on the choice of the initial points, which may lead to an unsatisfactory system performance. To investigate the optimal performance of IRS-assisted systems, in \cite{wu2019beamforming}, the authors considered a single-user scenario and determined the optimal BS beamformer and the optimal discrete IRS phase shifts using an enumeration-based algorithm. However, this scheme cannot be directly extended to multiuser systems, as finding closed-form optimal beamforming vectors seems intractable for multiuser scenarios. Moreover, assuming perfect CSI, in the conference version \cite{wu2023globally} of this paper, a novel algorithm based on the generalized Benders decomposition (GBD) framework was developed for attaining the globally jointly optimum BS beamformer and discrete IRS phase shift matrix for a simplified multiuser system model without direct link between the BS and the users. However, the optimal design for IRS-assisted multiuser multiple-input single-output (MISO) systems with direct link and imperfect CSI is still an open problem.

%
The authors of \cite{wu2019beamforming, wu2023globally} assume the perfect CSI of the IRS-assisted wireless system can be acquired. Unfortunately, due to the passive nature of IRSs, it is challenging to estimate the IRS-assisted links with conventional channel estimation schemes. As such, channel estimation errors are generally inevitable, leading to availability of imperfect CSI only \cite{yu2020robust,yu2021robust}. Furthermore, when considering the IRS-assisted links, i.e., the BS-IRS and IRS-user channels, they are typically cascaded into one effective channel for end-to-end channel estimation. Therefore, the CSI of the BS-IRS and IRS-user channels, along with the corresponding estimation errors should be considered jointly \cite{zhou2020framework,yu2021robust}. As a result, several early works have studied the robust design of IRS-assisted wireless systems assuming imperfect CSI knowledge of the IRS-induced cascaded channels and the direct channels between the BS and the users. For instance, in \cite{yu2021robust} and \cite{zhou2020framework}, worst-case robust optimization methods for practical system settings were developed based on algorithms exploiting linear approximation and successive convex approximation (SCA) methods, respectively. However, the authors in \cite{zhou2020framework,yu2021robust} assume IRSs with continuous phase shifts. Nevertheless, even under this optimistic assumption, the methods proposed in \cite{zhou2020framework,yu2021robust} cannot guarantee the global optimality of the resulting robust design. For practical IRS systems with imperfect CSI and discrete IRS phase shifters, a globally optimal robust resource allocation framework is not available in the literature, yet.

Motivated by the above discussion, in this paper, we investigate for the first time the globally optimal joint BS beamforming and IRS phase shift design for multiuser IRS-assisted systems employing discrete phase shifters. The main contributions of this paper can be summarized as follows:
\begin{itemize}
    \item In this work, we study resource allocation in IRS-assisted communication systems with discrete IRS phase shifters with the objective to minimize the BS transmit power, while guaranteeing the minimum signal-to-interference-plus-noise ratio (SINR) requirements of each user. Optimal and suboptimal resource allocation strategies are developed for perfect and imperfect CSI. 
    \item {\color{black}Our exploration begins with the resource allocation design for the case of perfect CSI. Assuming that perfect CSI of the entire system is available, we formulate the joint BS beamforming and IRS phase shift design as a mixed integer nonlinear programming (MINLP) problem. A series of transformations introduced in \cite{6698281, zhang2008joint} are leveraged and reformulated to recast the considered MINLP problem into a more tractable form, which facilitates the design of an iterative algorithm based on GBD. Specifically, the proposed GBD-based method is guaranteed to converge to the globally optimal solution. Moreover, the proposed GBD-based algorithm reveals an attainable performance upper bound for IRS-assisted wireless systems with discrete IRS phase shifts and can serve as a performance benchmark for any corresponding suboptimal design, e.g., \cite{wu2019beamforming, shi2022multiuser}. In addition, the proposed optimization framework further facilitates the design of an efficient suboptimal algorithm. In particular, leveraging the SCA approach, we exploit the proposed optimization framework to develop a novel and computationally efficient suboptimal scheme striking a favorable balance between performance and complexity.}

    \item {\color{black}We then, for the first time, investigate more practical IRS-assisted systems with discrete IRS phase shifts and imperfect CSI. Here, we introduce robust SINR constraints and, by suitably reformulating transformations reported in \cite{yu2021robust} and \cite{6698281} and adding a series of auxiliary optimization variables, recast them into a set of linear matrix inequality (LMI) constraints to obtain an MINLP problem for robust resource allocation design. Then, based on the proposed robust optimization framework, we extend the proposed GBD-based and SCA-based methods, originally conceived for the scenario with perfect CSI, to obtain the globally optimal and a locally optimal solution to the formulated robust resource allocation problem, respectively. Thus, the proposed robust optimization framework facilitates the design of practical IRS-assisted systems with discrete IRS phase shifts, where the proposed GBD-based method can serve as a performance benchmark for any corresponding suboptimal design. }
    \item By analyzing the convergence behavior of the proposed GBD-based algorithms, we confirm that they are capable of obtaining the globally optimal solution of the considered resource allocation problems for both perfect and imperfect CSI. Besides, both proposed SCA-based algorithms attain local optimality. Notably, the proposed SCA-based algorithms exhibit a faster convergence rate at the cost of a marginally higher total transmit power compared with the proposed GBD-based algorithms, which is advantageous for real-time implementation. Furthermore, our simulation results reveal that the proposed schemes outperform the commonly adopted suboptimal solution based on AO employing semidefinite relaxation. In particular, the performance gap between the proposed schemes and the AO-based design increases with the number of IRS elements, as the latter tends to get trapped in ineffective spurious local minima.
\end{itemize}
The remainder of this paper is organized as follows. We introduce the system model and formulate the considered resource allocation problem in Section \rom{2}. In Section \rom{3}, the proposed optimal and suboptimal algorithms for the case of perfect CSI are developed. These algorithms are extended to the case of imperfect CSI in Section \rom{4}. Simulation results are presented in Section \rom{5}, and Section \rom{6} concludes this paper.
\par
\textit{Notation:} 
Vectors and matrices are denoted by boldface lower-case and boldface capital letters, respectively. $\mathbb{R}^{N\times M}$ and $\mathbb{C}^{N\times M}$ denote the space of $N\times M$ real-valued and complex-valued matrices, respectively. $\mathbb{H}^N$ denotes the space of $N\times N$ Hermitian matrices. $||\cdot||_2$ and $||\cdot||_F$ denote the $l_2$-norm and Frobenius norm of the argument, respectively. $(\cdot)^T$, $(\cdot)^*$, and $(\cdot)^H$ stand for the transpose, conjugate, and conjugate transpose of their input arguments, respectively. $\mathbf{I}_{N}$ refers to the identity matrix of dimension $N$. $\mathrm{Tr}(\cdot)$ and $\mathrm{Rank}(\cdot)$ denote the trace and rank of the input matrix. $\mathbf{0}_{1\times L}$ represents the $1\times L$ all-zeros row vector. $\mathbf{A}\succeq\mathbf{0}$ indicates that $\mathbf{A}$ is a positive semidefinite matrix. $\mathrm{diag}(\mathbf{a})$ denotes a diagonal matrix whose main diagonal elements are given by vector $\mathbf{a}$.  $\mathrm{Re}\{\cdot\}$ and $\mathrm{Im}\{\cdot\}$ represent the real and imaginary parts of a complex number, respectively. $\mathbb{E}[\cdot]$ refers to statistical expectation. $\triangledown_{\mathbf{A}}f$ denotes the gradient of function $f$ with respect to matrix $\mathbf{A}$. $\mathrm{vec}(\cdot)$ refers to the vectorization of the input matrix. $\bigotimes$ stands for the Kronecker product between two matrices.
\vspace*{-3mm}
\section{System Model and Problem Formulation}
In this section, we first present the model for the considered IRS-assisted multiuser MISO system. Then, we formulate the resource allocation problem studied in this paper.
\vspace*{-4mm}
\subsection{IRS-assisted Multiuser System}
\begin{figure}[t]
	\centering
	\includegraphics[width=2.6in]{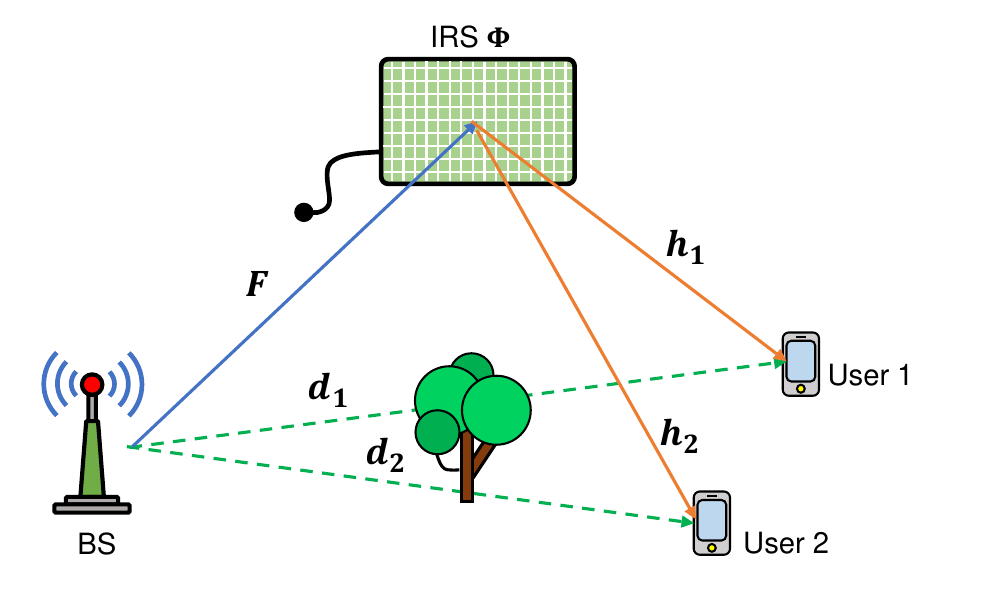}
	\caption{An IRS-assisted multiuser MISO system comprising $K=2$ users.}
	\label{fig::Model}
\end{figure}
{\color{black}As shown in Fig. \ref{fig::Model}, we consider a multiuser wireless communication system comprising an $M$-antenna BS and $K$ single-antenna users similar to \cite{wu2019beamforming}. We investigate a general setting, where the $K$ users share the same time-frequency resources for downlink communication. To provide high-data-rate communication services to the $K$ users, an IRS comprising $N$ phase shift elements is deployed to establish adaptive reflecting paths between the BS and the $K$ users\footnote{\color{black}The design proposed in this paper can be extended to the more general case with multiple IRSs by following the same approach as in \cite{yu2021robust}.}. The IRS is equipped with a smart controller that communicates with the BS to coordinate data transmission and to control the phase shifts of the reflecting elements.}
For notational simplicity, we define sets $\mathcal{K}=\{1,\cdots,K\}$ and $\mathcal{N}=\{1,\cdots,N\}$ to collect the indices of the users and the IRS elements, respectively. Then, the received discrete-time equivalent complex baseband signal $y_k$ at user $k$ is given by
\begin{equation}\label{oriChannel}
    {y}_k=(\mathbf{h}_k^H\bm{\Phi}\mathbf{F}+\mathbf{d}^H_k)\mathbf{W}\mathbf{s}+{n}_k,
\end{equation}
where $\mathbf{h}_k\in\mathbb{C}^{N\times 1}$ and $\mathbf{d}_k\in \mathbb{C}^{M\times 1}$ denote the conjugate channel between the IRS and user $k$ and the conjugate channel between the BS and user $k$, respectively. Matrix $\boldsymbol{\Phi}=\operatorname{diag}\left(e^{j\theta_1}.\cdots,e^{j\theta_N}\right)$ represents the phase shift matrix of the IRS. 
{\color{black}In this paper, we consider the practical case, where each element of the IRS can assume only $L=2^B$ different discrete phase shift values. Here, $B$ is the bit-resolution of each IRS element. For simplicity, we assume that the $L$ discrete phase shift values are uniformly quantized in the interval $[0,2\pi)$, i.e., $\theta_n\in\bm{\Psi}=\{0,\Delta\theta,\cdots,(L-1)\Delta\theta\}$, $\forall n\in\mathcal{N}$, where $\Delta\theta=\frac{2\pi}{L}$.\footnote{\color{black}The proposed resource allocation design can be adapted to the more general scenario considered in \cite{abeywickrama2020intelligent} where the feasible set of the phases of the IRS reflection coefficients is not uniformly quantized, and the amplitudes depend on the phases.  }}
The channel matrix between the BS and the IRS is denoted by $\mathbf{F}\in\mathbb{C}^{N\times M}$, and $\mathbf{W}=[\mathbf{w}_1,\cdots,\mathbf{w}_K]$ denotes the linear beamforming matrix at the BS, where $\mathbf{w}_k\in \mathbb{C}^{M\times 1}$ represents the linear beamforming vector for user $k$. Vector $\mathbf{s}=[s_1,\cdots,s_K]^T$ is the information-carrying symbol vector transmitted to the users, where $s_j\in\mathbb{C}$ denotes the symbol transmitted to user $j$ and $\mathbb{E}[|s_j|^2]=1,\mathbb{E}[s_j^*s_i]=0,j\neq i, \forall j,i\in\mathcal{K}$. Furthermore, $n_k\in\mathbb{C}$ represents the additive white Gaussian noise at user $k$ with zero mean and variance $\sigma_k^2$. 
\vspace*{-5mm}
\subsection{Resource Allocation Problem Formulation}
Without loss of generality, the received signal at user $k$ in \eqref{oriChannel} can be rewritten as:
\begin{equation}\label{orichannel1}
\begin{aligned}
    {y}_k&= [{\mathbf{h}}_k^H,1] \begin{bmatrix}
    \bm{\Phi} & \mathbf{0}_{N\times 1}\\
    \mathbf{0}_{1\times N} & 1
  \end{bmatrix}\begin{bmatrix}
    \mathbf{F} \\
    \mathbf{d}_k^H
  \end{bmatrix}\mathbf{W}\mathbf{s}+{n}_k\\
  &=\mathbf{v}^T{\mathbf{H}}_k\hat{\mathbf{F}}_k\mathbf{W}\mathbf{s}+{n}_k,
\end{aligned}
\end{equation}
where $\hat{\mathbf{F}}_k=[\mathbf{F}^H,\mathbf{d}_k]^H$, ${\mathbf{H}}_k=\mathrm{diag}([{\mathbf{h}}_k^H,1]),$ and $\mathbf{v}=[e^{j\theta_1},\cdots,e^{j\theta_N},1]^T$.
Next, we define the phase shift vector $\bm{\theta}=[1,e^{j\Delta\theta },\cdots, e^{j(L-1)\Delta\theta}]^T$. The binary selection vector of the $n$-th phase shifter is denoted by $\mathbf{b}_n=\big[b_n[1],\cdots,b_n[L]\big]^T,\hspace*{1mm}\forall n$, where $b_n[l]\in\left\{0,\hspace*{1mm}1\right\},\hspace*{1mm}\forall l,n,\hspace*{2mm} \sum_{l=1}^{L}b_n[l]=1,\hspace*{1mm}\forall n$. Then, the coefficient of the $n$-th IRS phase shifter can be expressed as $e^{j\theta_n}=\mathbf{b}_n^T\bm{\theta}$. Here, $b_n[l]=1$ if and only if the $l$-th phase shift value is selected, i.e., $e^{j\theta_n}=e^{j(l-1)\Delta\theta }.$ Hence, $\mathbf{v}$ can be expressed as $\mathbf{v}=\mathbf{B}^T\bm{\theta}$, where matrix $\mathbf{B}\in\mathbb{C}^{L\times(N+1)}$ is defined as $\mathbf{B}=[\mathbf{b}_1,\cdots,\mathbf{b}_{N+1}]$.
Note that $\mathbf{b}_{N+1}^T=[1,0,\cdots,0]$ is a constant binary vector. Therefore, the received signal of user $k$ in \eqref{oriChannel} can be rewritten as follows
\begin{equation}\label{orichannel1}
    y_k=\bm{\theta}^T\mathbf{B}{\mathbf{H}}_k\hat{\mathbf{F}}_k\mathbf{W}\mathbf{s}+{n}_k.
\end{equation}
Thus, the received SINR of user $k$ is given by
\begin{equation}\label{SINR}
    \mathrm{SINR}_k=\frac{|\bm{\theta}^T\mathbf{B}{\mathbf{H}}_k\hat{\mathbf{F}}_k\mathbf{w}_k|^2}{\sum_{k'\in\mathcal{K}\setminus\{k\}}|\bm{\theta}^T\mathbf{B}{\mathbf{H}}_k\hat{\mathbf{F}}_k\mathbf{w}_{k'}|^2+\sigma_k^2}.
\end{equation}
In this paper, we aim to minimize the total BS transmit power while guaranteeing
a minimum required SINR for each user. The optimal BS beamforming and IRS phase
shift selection policy, i.e., $\mathbf{W}$ and $\mathbf{B}$, is obtained by solving the following optimization problem
\begin{eqnarray}
\label{Ori_Problem}
    &&\hspace*{-6mm}\underset{\mathbf{B},\mathbf{W}}{\mino}\hspace*{2mm}\sum_{k\in\mathcal{K}}\left\|\mathbf{w}_k\right\|_2^2\notag\\[-10pt]
    &&\hspace*{-1mm}\mbox{s.t.}\hspace*{5mm} \mbox{C1:}\hspace*{1mm} \mathrm{SINR}_k\geq \gamma_k,\hspace*{1mm}\forall k\in \mathcal{K},\hspace*{1mm}\mbox{C2a:}\hspace*{1mm}\sum_{l=1}^{L}b_n[l]=1,\hspace*{1mm}\forall n,\notag\\
    &&\hspace*{8mm}\mbox{C2b:}\hspace*{1mm}b_n[l]\in\left\{0,\hspace*{1mm}1\right\},\hspace*{1mm}\forall l,\hspace*{1mm}\forall n,\vspace*{-3mm}
\end{eqnarray}
where $\gamma_k\geq 0$ denotes the pre-defined minimum required SINR of user $k$. There are two challenges to obtain the global optimum of \eqref{Ori_Problem}. First, the optimization variables $\mathbf{W}$ and $\mathbf{B}$ are coupled in constraint C1. Second, the feasible set defined by constraint C2b is highly non-convex. Overall, problem \eqref{Ori_Problem} belongs to the class of combinatorial optimization problems which are generally NP-hard \cite{xu2022optimal, wu2019beamforming, shi2022multiuser, hu2021robust}. 

We note that the problem formulation in \eqref{Ori_Problem} assumes availability of perfect CSI. The modifications needed to account for imperfect CSI will be presented in Section \rom{4}.
\vspace*{-1mm}
\section{Algorithm Design for Perfect CSI}
\vspace*{-1mm}
\begin{figure*}[!htbp]
    \centering
    \includegraphics[width=6.8in]{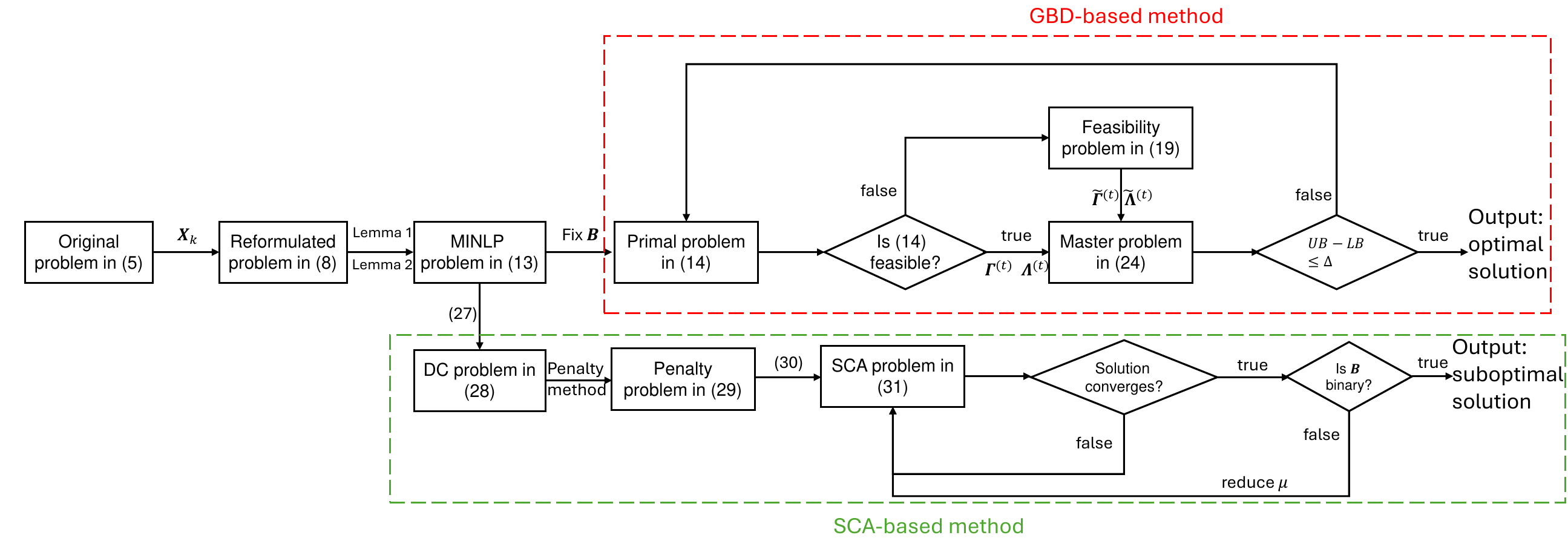}
    \caption{\color{black}Flowchart of the algorithms proposed for the case of perfect CSI.}
    \label{fig:FC_perfect_CSI}
\end{figure*}
{\color{black}In this section, assuming perfect CSI is available, we first reformulate the problem in \eqref{Ori_Problem} into an equivalent yet more tractable MINLP problem. Then, we develop two novel algorithms by exploiting the GBD and SCA methods, respectively, which attain the global optimum and a local optimum of the reformulated problem, respectively. An overview of the individual steps of the derivation and flow of the proposed algorithms is given in Fig. \ref{fig:FC_perfect_CSI}.}
\vspace*{-1mm}
\subsection{Problem Transformation}
\vspace*{-1mm}
{\color{black}To facilitate the optimal resource allocation design, we first define $\mathbf{X}_k=\mathbf{B}{\mathbf{H}}_k\hat{\mathbf{F}}_k\mathbf{W}\in\mathbb{C}^{L\times K}$ and rewrite \eqref{orichannel1} as follows
\begin{equation}
    y_k={\boldsymbol{\theta}^T\mathbf{x}_{k,k}s_k}+\sum_{k'\in\mathcal{K}\setminus k}{\boldsymbol{\theta}^T\mathbf{x}_{k,k'}s_{k'}}+n_k,
    \end{equation}
    where $\mathbf{x}_{k,k'}\in\mathbb{C}^{L}$ denotes the $k'$-th column of $\mathbf{X}_k$.}
Then, we can rewrite the SINR of user $k$ in \eqref{SINR} as 
\begin{equation}
     \mathrm{SINR}_k=\frac{|\bm{\theta}^T\mathbf{x}_{k,k}|^2}{\sum_{k'\in\mathcal{K}\setminus\{k\}}|\bm{\theta}^T\mathbf{x}_{k,k'}|^2+\sigma_k^2}.
\end{equation}
Next, we define $\mathbf{X}=[\mathbf{X}_1,\cdots,\mathbf{X}_K]\in\mathbb{C}^{L\times K^2}$ to collect all $\mathbf{X}_k, \forall k\in\mathcal{K}$, and recast the problem in \eqref{Ori_Problem} equivalently as follows
\begin{eqnarray}
\label{Ori_Problem1}
    &&\hspace*{-4mm}\underset{\mathbf{B},\,\mathbf{W},\,\mathbf{X}}{\mino}\hspace*{2mm}\sum_{k\in\mathcal{K}}\left\|\mathbf{w}_k\right\|_2^2\notag\\
    &&\hspace*{1mm}\mbox{s.t.}\hspace*{7mm} \mbox{C1},\mbox{C2a}, \mbox{C2b},
    \hspace*{1mm}\mbox{C3:}\,\mathbf{X}_k=\mathbf{B}{\mathbf{H}}_k\hat{\mathbf{F}}_k\mathbf{W},
\end{eqnarray}
where equality constraint C3 is non-convex due to the coupling between $\mathbf{B}$ and $\mathbf{W}$.
To deal with this difficulty, we exploit the following lemma to transform equality constraint ${\mbox{C3}}$ into two equivalent inequality constraints.
\vspace*{-2mm}
\begin{lemma}
Based on Schur's complement, equality constraint ${\mbox{C3}}$ is equivalent to the following inequality constraints:
\begin{eqnarray}
\mathrm{{C3a:}}&\hspace*{1mm}\label{sdp}
   \begin{bmatrix}
        \mathbf{S}_k & \mathbf{X}_k & \mathbf{B}{\mathbf{H}}_k\\
        \mathbf{X}_k^H & \mathbf{T}_k & \mathbf{W}^H\hat{\mathbf{F}}_k^H\\
        {\mathbf{H}}_k^H\mathbf{B}^H & \hat{\mathbf{F}}_k\mathbf{W} & \mathbf{I}_{N+1}
    \end{bmatrix}&\succeq \mathbf{0},\forall k\in\mathcal{K},\notag\\
\mathrm{{C3b:}}&\hspace*{1mm}\label{DC}
    \mathrm{Tr}\left(\mathbf{S}_k-\mathbf{B}{\mathbf{H}}_k{\mathbf{H}}_k^H\mathbf{B}^H\right)&\leq0,\forall k\in\mathcal{K},\vspace*{-2mm}
\end{eqnarray}
where $\mathbf{S}_k\in\mathbb{C}^{L\times L}\succeq \mathbf{0}$ and $\mathbf{T}_k\in\mathbb{C}^{K\times K}\succeq \mathbf{0},\forall k\in\mathcal{K}$, are auxiliary optimization variables.\vspace*{-2mm}
\end{lemma}
\vspace*{-2mm}\begin{proof}
    Please refer to \cite[Appendix A]{6698281}.
\end{proof}

We note that constraint $\mbox{C3a}$ is a convex LMI constraint. 
Yet, constraint $\mbox{C3b}$ is in the form of a difference of two convex functions, which is a non-convex constraint. 
{\color{black}To circumvent this obstacle, we rewrite constraint C3b as a linear inequality constraint for $\mathbf{S}_k$, $\forall k \in\mathcal{K}$, as follows: 
\begin{equation}\label{LPDC}
    \begin{aligned}
         \mathrm{Tr}\left(\mathbf{S}_k-\mathbf{B}{\mathbf{H}}_k{\mathbf{H}}_k^H\mathbf{B}^H\right)&\overset{(a)}{=}\mathrm{Tr}(\mathbf{S}_k)-\sum_{n\in\mathcal{N}}|h_{k,n}|^2\mathbf{b}_n^T\mathbf{b}_n-1\\
         &\overset{(b)}{=}\mathrm{Tr}\big(\mathbf{S}_k\big)-\sum_{n\in\mathcal{N}}|h_{k,n}|^2-1\leq 0,
    \end{aligned}
    \end{equation}
     where equalities (a) and (b) hold due to the diagonality of matrix $\mathbf{H}_k=\mathrm{diag}([\mathbf{h}_k^H,1])$ and the nature of the binary selection vector, i.e., $\|\mathbf{b}_n\|_2^2=1$, $\forall n\in\mathcal{N}$, respectively.}
{\color{black} On the other hand, the SINR constraint C1 in (8) is also non-convex. Nevertheless, we observe that if $\mathbf{W}$ satisfies QoS constraint C1 in (8),  so does $\mathbf{W}\mathrm{diag}(e^{j\omega_1},\cdots,e^{j\omega_K})$ for any arbitrary phase shift $\omega_k$, $\forall k \in\mathcal{K}$, with unchanged value of the objective function \cite{zhang2008joint}. Thus, we can always find an optimal $\mathbf{w}_k$, for which 
$\bm{\theta}^T\mathbf{B}{\mathbf{H}}_k\hat{\mathbf{F}}_k\mathbf{w}_k\in\mathbb{R}$ holds. Therefore, we can assume that $\bm{\theta}^T\mathbf{B}{\mathbf{H}}_k\hat{\mathbf{F}}_k\mathbf{w}_k\in\mathbb{R}$ without loss of generality and optimality. Moreover, since $\mathbf{x}_{k,k}$ is defined as $\mathbf{x}_{k,k}=\mathbf{B}{\mathbf{H}}_k\hat{\mathbf{F}}_k\mathbf{w}_k$, $\bm{\theta}^T\mathbf{x}_{k,k}\in\mathbb{R}$ holds if $\bm{\theta}^T\mathbf{B}{\mathbf{H}}_k\hat{\mathbf{F}}_k\mathbf{w}_k\in\mathbb{R}$ holds.
Thus, we can assume ${\bm{\theta}}^T\mathbf{x}_{k,k}\in\mathbb{R}$ without loss of generality and optimality and transform C1 to a convex constraint by exploiting the following lemma.}
\vspace*{-2mm}
\begin{lemma}
Without loss of generality and optimality, we assume that ${\bm{\theta}}^T\mathbf{x}_{k,k}\in\mathbb{R}$ \cite{zhang2008joint}. Then, constraint C1 can be equivalently rewritten as
\begin{eqnarray}
&& \hspace*{-12mm}{{\mathrm{C1a}}}\mbox{:}\hspace*{1mm}\sqrt{\sum_{k'\in\mathcal{K}\setminus\{k\}}|{\bm{\theta}}^T\mathbf{x}_{k,k'}|^2+\sigma_k^2}\leq\frac{\operatorname{Re}\{{\bm{\theta}}^T\mathbf{x}_{k,k}\}}{\sqrt{\gamma_k}}, \hspace*{1mm}\forall k,\\
&&\hspace*{-12mm}\mathrm{C1b:}\hspace*{1mm}\operatorname{Im}\{{\bm{\theta}}^T\mathbf{x}_{k,k}\}=0,\hspace*{1mm}\forall k,\vspace*{-2mm}
\end{eqnarray}
where constraints C1a and C1b are both convex constraints.\vspace*{-2mm}
\end{lemma}
\vspace*{-2mm}
\begin{proof}
Please refer to \cite[Appendix \rom{2}]{zhang2008joint}.\vspace*{-2mm}
\end{proof}
For notational simplicity, we define $\mathbf{S}=[\mathbf{S}_1,\cdots,\mathbf{S}_K]\in\mathbb{C}^{L\times LK}$ and $\mathbf{T}=[\mathbf{T}_1,\cdots,\mathbf{T}_K]\in\mathbb{C}^{K\times K^2}$ as the collections of all $\mathbf{S}_k$ and $\mathbf{T}_k,\forall k \in\mathcal{K}$, respectively. Thus, the original optimization problem can be equivalently transformed into
\color{black}
\begin{eqnarray}
\label{Reformulated_problem}
    &&\hspace*{-6mm}\underset{\substack{\mathbf{X},\mathbf{W},\mathbf{B},\mathbf{S},\mathbf{T}}}{\mino}\hspace*{2mm}\sum_{k\in\mathcal{K}}\left\|\mathbf{w}_k\right\|_2^2\notag\\
    &&\hspace*{0mm}\mbox{s.t.}\hspace*{8mm}\mbox{C1a},\mbox{C1b},\mbox{C2a},{\mbox{C2b}},{\mbox{C3a}}\\
    &&\hspace*{12mm}\overline{\mbox{C3b}}:\mathrm{Tr}\big(\mathbf{S}_k\big)-\sum_{n\in\mathcal{N}}|h_{k,n}|^2-1\leq 0,\ \forall k \in\mathcal{K}.\notag
\end{eqnarray}
\color{black}
\vspace*{-2mm}\begin{remark}
For a given fixed binary matrix $\mathbf{B}$, the MINLP problem in \eqref{Reformulated_problem} degenerates into a convex optimization problem with respect to (w.r.t.) continuous variables $\mathbf{X},\mathbf{W},\mathbf{S}$, and $\mathbf{T}$. On the other hand, the problem in \eqref{Reformulated_problem} is a linear programming problem w.r.t. the discrete matrix $\mathbf{B}$ if continuous variables $\mathbf{X},\mathbf{W},\mathbf{S}$, and $\mathbf{T}$ are fixed. Therefore, according to \cite{geoffrion1972generalized,sahinidis1991convergence}, GBD is guaranteed to yield the globally optimal solution of \eqref{Reformulated_problem}.
\end{remark}
\vspace*{-2mm}\subsection{Proposed GBD-Based Optimal Algorithm}\label{Opt_perfect_CSI}
We develop a GBD-based iterative algorithm to optimally solve the combinatorial optimization problem \eqref{Reformulated_problem}. The crux of GBD is the partitioning of an MINLP problem into a primal problem and a master problem. These two problems are solved iteratively until convergence \cite{geoffrion1972generalized}.
In particular, we obtain the primal problem by fixing the binary matrix $\mathbf{B}$ in \eqref{Reformulated_problem}, and then solve the primal problem w.r.t. $\mathbf{X}$, $\mathbf{W}$, $\mathbf{S}$, and $\mathbf{T}$, which produces an upper bound (UB) for the original problem in \eqref{Reformulated_problem}. Furthermore, by fixing optimization variables $\mathbf{X}$, $\mathbf{W}$, $\mathbf{S}$, and $\mathbf{T}$, the master problem is in the form of a typical mixed integer linear programming (MILP) problem w.r.t. $\mathbf{B}$. The solution of the master problem yields a lower bound (LB) for the optimization problem in \eqref{Reformulated_problem}. By solving the primal problem and the master problem iteratively, the proposed GBD-based algorithm is guaranteed to converge to the globally optimal solution of \eqref{Reformulated_problem} \cite{geoffrion1972generalized}. In the following, we formulate and solve the primal and master problems for the $i$-th iteration of the GBD algorithm.
\subsubsection{Primal Problem}

Given binary matrix $\mathbf{B}^{(i-1)}$ obtained by solving the master problem in the $(i-1)$-th iteration, the primal problem in the $i$-th iteration is given by

\vspace*{-2mm}
\begin{eqnarray}
\label{Primal_problem}
    \underset{\bm{\Gamma}}{\mino}\hspace*{2mm}\sum_{k\in\mathcal{K}}\left\|\mathbf{w}_k\right\|_2^2
    \hspace*{4mm}\mbox{s.t.}\hspace*{2mm} \mbox{C1a}, \mbox{C1b},\mbox{C3a}, \overline{\mbox{C3b}},
\end{eqnarray}
where $\bm{\Gamma}=\{\mathbf{X},\mathbf{W},\mathbf{S},\mathbf{T}\}$ denotes the collection of the optimization variables in \eqref{Primal_problem}.
Note that the problem in \eqref{Primal_problem} is convex w.r.t. $\bm{\Gamma}$ and can be solved by using standard off-the-shelf convex program solvers such as CVX \footnote{\color{black}A closed-form solution of \eqref{Primal_problem} cannot be derived analytically due to the challenges in handling LMI and second order cone constraints. Thus, the performance of the proposed algorithm can only be verified by numerical evaluation, see Section \rom{5}.} \cite{grant2008cvx}. If the problem in \eqref{Primal_problem} is feasible in the $i$-th iteration, we denote the optimal solution by $\bm{\Gamma}^{(i)}=\{\mathbf{X}^{(i)},\mathbf{W}^{(i)},\mathbf{S}^{(i)}, \mathbf{T}^{(i)}\}$. Then, we obtain the Lagrangian of \eqref{Primal_problem} as follows
\begin{equation}\label{Largrangian_conf}
\hspace*{-1mm}\mathcal{L}(\bm{\Gamma},\mathbf{B}^{(i-1)},\bm{\Lambda})\hspace*{-0.5mm}=\hspace*{-1mm}\sum_{k\in\mathcal{K}}\left\|\mathbf{w}_k\right\|_2^2\hspace*{-0.25mm}+f_1(\bm{\Gamma},\bm{\Lambda})\hspace*{-0.25mm}+f_2(\mathbf{B}^{(i-1)},\bm{\Lambda}),
\end{equation}
where $f_1(\bm{\Gamma},\bm{\Lambda})$ is given by \eqref{f1} which is shown at the top of the next page and 
\vspace*{-2mm}
\begin{figure*}[!t]
\normalsize

\setcounter{equation}{15}

\begin{equation}\label{f1}
\begin{aligned}
f_1(\bm{\Gamma},\bm{\Lambda})&=\sum_{k\in\mathcal{K}}\Biggr[\alpha_k\left(\sqrt{\sum_{k'\neq k}\hspace*{-2mm}|\bm{\theta}^T\mathbf{x}_{k,k'}|^2+\sigma_k^2}-\frac{\operatorname{Re}\{\bm{\theta}^T\mathbf{x}_{k,k}\}}{\sqrt{\gamma_k}}\right)+\beta_k\big(\operatorname{Im}\{\bm{\theta}^T\mathbf{x}_{k,k}\}\big)+\zeta_k(\mathrm{Tr}(\mathbf{S}_k)-(N+1))\\
&+\mathrm{Tr}\big(\mathbf{S}_k\mathbf{Q}_{k,11}\big)+\mathrm{Tr}\big(\mathbf{T}_k\mathbf{Q}_{k,22}\big)+2\mathrm{Re}\left\{\mathrm{Tr}\big({\mathbf{Q}^H_{k,32}}\widehat{\mathbf{F}}_k\mathbf{W}\big)+\mathrm{Tr}\big(\mathbf{X}_k\mathbf{Q}_{k,21}\big)\right\}\Biggr],
\end{aligned}\vspace*{0mm}
\end{equation}
\setcounter{equation}{16}
\hrulefill
\vspace*{0pt}
\end{figure*}

\vspace*{-1mm}
\begin{eqnarray}
     f_2(\mathbf{B}^{(i-1)},\bm{\Lambda})=\sum_{k\in\mathcal{K}}2\mathrm{Re}\left\{\mathrm{Tr}\big(\mathbf{B}^{(i-1)}{\mathbf{H}}_k\mathbf{Q}_{k,31}\big)\right\}.\vspace*{-1mm}
\end{eqnarray}
Here, $\bm{\Lambda}=\left\{\alpha_k,\beta_k,\mathbf{Q}_k,\zeta_k \right\}$ denotes the collection of dual variables, where $\alpha_k,\beta_k$, $\mathbf{Q}_k$, and $\zeta_k$ represent the dual variables for constraints C1a, C1b, C3a, and $\overline{\mbox{C3b}}$, respectively, and the dual variable matrix $\mathbf{Q}_k\in\mathbb{C}^{(N+K+L+1)\times (N+K+L+1)}$ for constraint C3a is decomposed into nine sub-matrices as follows
\begin{equation}\label{Qi}
  \mathbf{Q}_k=\left[ \begin{array}{ccc}
        \mathbf{Q}_{k,11} & \mathbf{Q}_{k,21}^H & \mathbf{Q}_{k,31}^H\\
        \mathbf{Q}_{k,21} & \mathbf{Q}_{k,22} & \mathbf{Q}_{k,32}^H\\
        \mathbf{Q}_{k,31} & \mathbf{Q}_{k,32} & \mathbf{Q}_{k,33}
    \end{array}\right],\forall k \in\mathcal{K},\vspace*{0mm}
\end{equation}
with $\mathbf{Q}_{k,11}\in\mathbb{C}^{L\times L}$, $\mathbf{Q}_{k,21}\in\mathbb{C}^{K\times L}$, $\mathbf{Q}_{k,22}\in\mathbb{C}^{K\times K}$, $\mathbf{Q}_{k,31}\in\mathbb{C}^{(N+1)\times L}$, $\mathbf{Q}_{k,32}\in\mathbb{C}^{(N+1)\times K}$, and $\mathbf{Q}_{k,33}\in\mathbb{C}^{(N+1)\times (N+1)}$.
On the other hand, if problem \eqref{Primal_problem} is infeasible for a given $\mathbf{B}^{(i-1)}$, then we formulate an $l_1$-minimization feasibility-check problem as follows
\begin{eqnarray}
\label{Feasible_problem}
    &&\hspace*{-6mm}\underset{\bm{\Gamma},\,\bm{\lambda}}{\mino}\hspace*{1mm}\sum_{k\in\mathcal{K}}\lambda_k\notag\\
    &&\hspace*{-1mm}\mbox{s.t.}\hspace*{3mm}\overline{\mbox{C1a}}\mbox{:}\hspace*{1mm}\sqrt{\sum_{k'\neq k}\hspace*{-1mm}|\bm{\theta}^T\mathbf{x}_{k,k'}|^2\hspace*{-0.5mm}+\hspace*{-0.5mm}\sigma_k^2}\hspace*{-0.5mm}-\hspace*{-0.5mm}\frac{\operatorname{Re}\{\bm{\theta}^T\mathbf{x}_{k,k}\}}{\sqrt{\gamma_k}}\leq \lambda_k,\forall k,\notag\\
    &&\hspace*{6mm}\mbox{C1b},\mbox{C3a},\overline{\mbox{C3b}},\hspace*{1mm}\mbox{C4}\mbox{:}\hspace*{1mm}\lambda_k\geq 0, \forall k,\vspace*{-2mm}
\end{eqnarray}
where $\bm{\lambda}=[\lambda_1,\cdots,\lambda_K]$ is an auxiliary optimization vector. Problem \eqref{Feasible_problem} is convex w.r.t. $\bm{\Gamma}$ and $\bm{\lambda}$ and always feasible \cite{geoffrion1972generalized,sahinidis1991convergence}. Thus, problem \eqref{Feasible_problem} can be solved with CVX \cite{grant2008cvx}.
{\color{black}\begin{remark}
    If $\mathbf{B}^{(i-1)}$ is infeasible for the primal problem in \eqref{Primal_problem}, the objective function value of \eqref{Feasible_problem} will be positive, i.e.,  $\sum_{k\in\mathcal{K}}\lambda_k>0$ holds. On the other hand, if the objective function value $\sum_{k\in\mathcal{K}}\lambda_k$ is equal to zero, i.e., $\lambda_1=\lambda_2\cdots=\lambda_K=\sum_{k\in\mathcal{K}}\lambda_k=0$, the constraints in \eqref{Feasible_problem} are identical to the constraints in \eqref{Primal_problem}, and thus $\mathbf{B}^{(i-1)}$ is feasible for the primal problem in \eqref{Primal_problem} \cite[Theorem 2.2]{geoffrion1972generalized}.
\end{remark}}
Furthermore, the Lagrangian of \eqref{Feasible_problem} is given by
\begin{equation}\label{Largrangian_feasible}
\begin{aligned}
    \widetilde{\mathcal{L}}(\bm{\Gamma},\mathbf{B}^{(i-1)},\widetilde{\bm{\Lambda}})&={f}_1(\bm{\Gamma},\widetilde{\bm{\Lambda}})+{f}_2(\mathbf{B}^{(i-1)},\widetilde{\bm{\Lambda}}),
\end{aligned}
\end{equation}
where $\widetilde{\bm{\Lambda}}=[\widetilde{\alpha}_k,\widetilde{\beta_k},\widetilde{\mathbf{Q}}_k,\widetilde{\zeta}_k]$ denotes the collection of dual variables $\widetilde{\alpha}_k\in\mathbb{R},\widetilde{\beta}_k\in\mathbb{R}$, \\
$\widetilde{\mathbf{Q}}_k\in\mathbb{C}^{(N+K+L+1)\times (N+K+L+1)}$, and $\widetilde{\zeta}_k\in\mathbb{R}$ for constraints $\overline{\mbox{C1a}}$, C1b, C3a, and $\overline{\mbox{C3b}}$ in \eqref{Feasible_problem}, respectively. Similar to the notation in \eqref{Primal_problem}, the optimal solution of \eqref{Feasible_problem} in the $i$-th iteration is denoted by $\widetilde{\bm{\Gamma}}^{(i)}=\{\widetilde{\mathbf{X}}^{(i)},\widetilde{\mathbf{W}}^{(i)},\widetilde{\mathbf{S}}^{(i)}, \widetilde{\mathbf{T}}^{(i)}\}$, while $\widetilde{\bm{\Lambda}}^{(i)}$ denotes the optimal dual solution. The solutions of the feasibility-check problem in \eqref{Feasible_problem} are used to generate the feasibility cut removing infeasible solutions $\mathbf{B}^{(i-1)}$ from the feasible set of the master problem in the subsequent iterations.
\subsubsection{Master Problem}
{\color{black}The master problem can be derived based on nonlinear convex duality theory \cite{geoffrion1972generalized}. In particular, since the MINLP problem in \eqref{Reformulated_problem} is convex w.r.t. continuous variables $\mathbf{X},\mathbf{W},\mathbf{S}$, and $\mathbf{T}$ for a given binary variable $\mathbf{B}$, the MINLP problem in \eqref{Ori_Problem1} is equivalent to the following Lagrangian dual problem,
\begin{eqnarray}
    &&\hspace*{-6mm}\underset{\mathbf{B}}{\mino}\hspace*{2mm}\max_{\bm{\Lambda}\in\mathcal{F}_{\Lambda}}\min_{\bm{\Gamma}}\mathcal{L}(\bm{\Gamma},\mathbf{B},\bm{\Lambda})\notag\\
    &&\hspace*{0mm}\mbox{s.t.}\hspace*{6mm}\mbox{C2a}, \mbox{C2b},\\
    &&\hspace*{12mm}\mathbf{B}\in\mathcal{F}_B,\notag
\end{eqnarray}
where $\mathcal{F}_{\Lambda}$ and $\mathcal{F}_B$ denote the feasible sets of $\bm{\Lambda}$ and $\mathbf{B}$, respectively. Next, based on \cite[Theorem 2.2]{geoffrion1972generalized}, the feasible set of $\mathbf{B}$ can be represented by the following inequality constraints
\begin{equation}
    0\geq\min_{\bm{\Gamma}}\widetilde{\mathcal{L}}(\bm{\Gamma},\mathbf{B},\widetilde{\bm{\Lambda}}), \, \forall \widetilde{\bm{\Lambda}}\in\mathcal{F}_{\widetilde{{\Lambda}}},
\end{equation}
where $\mathcal{F}_{\widetilde{{\Lambda}}}$ denotes the feasible region of $\widetilde{\boldsymbol{\Lambda}}$. Moreover, we introduce an auxiliary variable \\$\eta=\max_{\bm{\Lambda}\in\mathcal{F}_{\Lambda}}\min_{\bm{\Gamma}}\mathcal{L}(\bm{\Gamma},\mathbf{B},\bm{\Lambda})$ and reformulate the Lagrangian dual problem as follows
\begin{eqnarray}\label{Dual_problem_epi}
    &&\hspace*{-6mm}\underset{\mathbf{B}}{\mino}\hspace*{2mm}\eta\notag\\
    &&\hspace*{0mm}\mbox{s.t.}\hspace*{6mm}\mbox{C2a}, \mbox{C2b},\\
    &&\hspace*{12mm}\hspace*{1mm}\eta\geq\min_{\bm{\Gamma}}\mathcal{L}(\bm{\Gamma},\mathbf{B},\bm{\Lambda}),\hspace*{0mm}\forall \bm{\Lambda}\in\mathcal{F}_{{{\Lambda}}},\notag\\
    &&\hspace*{12mm}\hspace*{1mm}0\geq\min_{\bm{\Gamma}}\widetilde{\mathcal{L}}(\bm{\Gamma},\mathbf{B},\widetilde{\bm{\Lambda}}),\hspace*{0mm}\forall \widetilde{\bm{\Lambda}}\in\mathcal{F}_{\widetilde{{\Lambda}}}.\notag\vspace*{-2mm}
\end{eqnarray}
Note that the reformulated dual problem cannot be solved due to the infinite number of inequality constraints. Therefore, we relax constraints $\eta\geq\min_{\bm{\Gamma}}\mathcal{L}(\bm{\Gamma},\mathbf{B},\bm{\Lambda}),\hspace*{0mm}\forall \bm{\Lambda}\in\mathcal{F}_{{{\Lambda}}}$ and $0\geq\min_{\bm{\Gamma}}\widetilde{\mathcal{L}}(\bm{\Gamma},\mathbf{B},\widetilde{\bm{\Lambda}}),\hspace*{0mm}\forall \widetilde{\bm{\Lambda}}\in\mathcal{F}_{\widetilde{{\Lambda}}}$ to formulate the master problem. In particular, we define sets $\mathcal{F}$ and $\mathcal{I}$ as the collections of iteration indices for which \eqref{Primal_problem} is feasible and infeasible, respectively, and recast the master problem in the $i$-th iteration as follows
\vspace*{-2mm}
\begin{eqnarray}
\label{Master_problem}
    &&\hspace*{-6mm}\underset{\mathbf{B},\eta}{\mino}\hspace*{2mm}\eta\notag\\
    &&\hspace*{0mm}\mbox{s.t.}\hspace*{6mm}\mbox{C2a}, \mbox{C2b},\\
    &&\hspace*{10mm}\mbox{C5a}\mbox{:}\hspace*{1mm}\eta\geq\min_{\bm{\Gamma}}\mathcal{L}(\bm{\Gamma},\mathbf{B},\bm{\Lambda}^{(t)}),\hspace*{0mm}\forall t\in\mathcal{F}\cap\{1,\cdots,i\},\notag\\
    &&\hspace*{10mm}\mbox{C5b}\mbox{:}\hspace*{1mm}0\geq\min_{\bm{\Gamma}}\widetilde{\mathcal{L}}(\bm{\Gamma},\mathbf{B},\widetilde{\bm{\Lambda}}^{(t)}),\hspace*{0mm}\forall t\in\mathcal{I}\cap\{1,\cdots,i\},\notag\vspace*{-2mm}
\end{eqnarray}
where constraints C5a and C5b represent the optimality and feasibility cuts \cite{geoffrion1972generalized}, respectively. In particular, introducing the optimality cut reduces the search region of the globally optimal solution. On the other hand, the feasibility cut removes the corresponding infeasible solutions from the search space. If the given $\mathbf{B}^{(i-1)}$ is infeasible for the primal problem in (14), $0<\widetilde{\mathcal{L}}(\bm{\Gamma}^{(i)},\mathbf{B}^{(i-1)},\widetilde{\bm{\Lambda}}^{(i)})$ must hold based on Remark 2. Thus, the infeasible $\mathbf{B}^{(i-1)}$ can be removed from the feasible set by introducing a constraint $0\geq\min_{\bm{\Gamma}}\widetilde{\mathcal{L}}(\bm{\Gamma},\widetilde{\mathbf{B}},\widetilde{\bm{\Lambda}}^{(i)})$. Note that the master problem is a relaxed version of the Lagrangian dual problem in \eqref{Dual_problem_epi}. Thus, the optimal solution of the master problem provides a lower bound for the original MINLP problem in \eqref{Reformulated_problem}.}
Moreover, the inner minimization in C5a and C5b can be obtained from the optimal solutions of the primal problem in \eqref{Primal_problem} and the feasibility-check problem in \eqref{Feasible_problem} exploiting the following lemma:
\vspace*{-2mm}\begin{lemma}
Inequality constraints C5a and C5b can be equivalently recast as the following two linear inequalities:
\begin{eqnarray}
\label{Recast_C5}
    &&\hspace*{-4mm}\overline{\mathrm{C5a}}\mathrm{:}\hspace*{1mm}\eta\geq \sum_{k\in\mathcal{K}}\left\|\mathbf{w}_k^{(t)}\right\|_2^2+f_1(\bm{\Gamma}^{(t)},\bm{\Lambda}^{(t)})+f_2(\mathbf{B},\bm{\Lambda}^{(t)}),\notag\\
    &&\hspace*{11mm}\forall t\in\mathcal{F}\cap\{1,\cdots,i\},\\
    &&\hspace*{-4mm}\overline{\mathrm{C5b}}\mathrm{:}\hspace*{1mm}0\geq f_1(\Tilde{\bm{\Gamma}}^{(t)},\Tilde{\bm{\Lambda}}^{(t)})+{f}_2(\mathbf{B},\widetilde{\bm{\Lambda}}^{(t)}),\forall t\in\mathcal{I}\cap\{1,\cdots,i\},\vspace*{-2mm}\notag
\end{eqnarray}
respectively.
\end{lemma}
\begin{proof}\vspace{-2mm}
The proof is identical to that for \cite[Proposition 1]{ng2015secure}, and thus is omitted here due to page limitation.
\end{proof}
By substituting $\bm{\Gamma}^{(t)}$ and $\widetilde{\bm{\Gamma}}^{(t)}$ in C5a and C5b, respectively, the master problem in \eqref{Master_problem} becomes a typical MILP problem, which can be optimally solved by employing standard numerical MILP solvers, e.g.,  MOSEK \cite{grant2008cvx}. The master problem provides a lower bound, $\eta^{(i)}$, to the original problem \eqref{Reformulated_problem} and its solution, $\mathbf{B}^{(i)}$, is adopted to generate the primal problem in the next iteration. 
{\color{black}\begin{proposition}
    If $\mathbf{B}^{(i)}$ is the optimal solution of the master problem in the $i$-th iteration, $\mathbf{B}^{(j)}\neq\mathbf{B}^{(i)}$ holds for the solution of the master problem in the $j$-th iteration, $\forall j > i$.
\end{proposition}
\begin{proof}
    Please refer to the proof of \cite[Theorem 2.4]{geoffrion1972generalized}.
\end{proof}
In the following section, we will exploit Proposition 1 to establish the global optimality of the proposed GBD-based algorithm.
}
\setlength{\textfloatsep}{0pt}
\begin{algorithm}[t]
\caption{Optimal Resource Allocation Algorithm}
\begin{algorithmic}[1]
\small
\STATE Set iteration index $i=0$, initialize upper bound $\mathrm{UB}^{(0)}\gg 1$, lower bound $\mathrm{LB}^{(0)}=0$, the set of the feasible iteration indices $\mathcal{F}^{(0)}=\emptyset$, the set of the infeasible iteration indices $\mathcal{I}^{(0)}=\emptyset$, and convergence tolerance $\Delta\ll 1$ and generate a feasible $\mathbf{B}^{(0)}$
\REPEAT
\STATE Set $i=i+1$
\STATE Solve \eqref{Primal_problem} for given $\mathbf{B}^{(i-1)}$
\IF{the primal problem in \eqref{Primal_problem} is feasible}
\STATE Update $\bm{\Gamma}^{(i)}$ and store the corresponding objective function value of $\sum_{k\in\mathcal{K}}\left\|\mathbf{w}_k^{(i)}\right\|_2^2$
\STATE Construct $\mathcal{L}(\bm{\Gamma},\mathbf{B},\bm{\Lambda}^{(i)})$ based on \eqref{Largrangian_conf}
\STATE Update the upper bound of \eqref{Reformulated_problem} based on \eqref{UB} 
\STATE Update $\mathcal{F}$ by $\mathcal{F}\cup\{i\}$
\ELSE
\STATE Solve \eqref{Feasible_problem} for a given $\mathbf{B}^{(i-1)}$, update $\widetilde{\bm{\Gamma}}^{(i)}$
\STATE Construct $\widetilde{\mathcal{L}}(\bm{\Gamma},\mathbf{B},\widetilde{\bm{\Lambda}}^{i})$ based on \eqref{Largrangian_feasible}
\STATE Update $\mathcal{I}$ by $\mathcal{I}\cup\{i\}$
\ENDIF
\STATE Solve the optimization problem in \eqref{Master_problem} and update $\eta^{(i)}$ and $\mathbf{B}^{(i)}$
\STATE Update the lower bound as $\mathrm{LB}^{(i)}=\eta^{(i)}$
\UNTIL $\mathrm{UB}^{(i)}-\mathrm{LB}^{(i)}\leq \Delta$
\end{algorithmic}
\end{algorithm}

\subsubsection{Overall Algorithm}
The complete GBD procedure is outlined in \textbf{Algorithm 1}. Some additional remarks are as follows:
\begin{enumerate}
    \item \textit{Initial point:} Iteration index $i$ is first set to zero and a binary matrix $\mathbf{B}$ satisfying constraints C2a and C2b is randomly generated.
    \item \textit{Optimality and convergence:} Given the binary matrix $\mathbf{B}$, problem \eqref{Primal_problem} is solved. In particular, if \eqref{Primal_problem} is feasible, we generate an optimality cut based on the intermediate solutions for $\bm{\Gamma}^{(i)}$, and the corresponding Lagrangian multiplier set $\bm{\Lambda}^{(i)}$. Additionally, the objective value $\sum_{k\in\mathcal{K}}\left\|\mathbf{w}_k^{(i)}\right\|_2^2$ obtained in the current iteration is exploited to update the performance upper bound $\mathrm{UB}^{(i)}$ as follows
    \begin{equation}\label{UB}
        \mathrm{UB}^{(i)}=\min\left\{\mathrm{UB}^{(i-1)},\hspace*{1mm}\sum_{k\in\mathcal{K}}\left\|\mathbf{w}_k^{(i)}\right\|_2^2\right\}.
    \end{equation}
     On the other hand, if \eqref{Primal_problem} is infeasible, we turn to solve the $l_1$-norm minimization problem in \eqref{Feasible_problem}. The obtained solutions and dual variables for problem \eqref{Feasible_problem} in the current iteration are adopted to generate an infeasibility cut in the master problem to reduce the feasible solution space for searching. Then, we solve the master problem in \eqref{Master_problem} optimally using a standard MILP solver. The objective function value of the master problem serves as a performance lower bound $\mathrm{LB}^{(i)}$ for the original optimization problem in \eqref{Ori_Problem1}. 
     {\color{black}We note that the value of UB is monotonically non-increasing, while the value of LB is monotonically non-decreasing. Moreover, based on Proposition 1, no $\mathbf{B}^{(i)}$ can be repeated in a subsequent iteration of \textbf{Algorithm 1} and the feasible set of $\mathbf{B}$ is a finite set. Thus, according to \cite[Theorem 2.4]{geoffrion1972generalized}, the proposed GBD-based algorithm is guaranteed to converge to a globally optimal solution of the formulated MINLP problem in a finite number of iterations for a given convergence tolerance $\Delta\geq 0$.}
    \item \textit{Complexity:} {\color{black}In each iteration of \textbf{Algorithm 1}, we have to solve a convex primal problem and an MILP master problem. Since the primal problem in \eqref{Primal_problem} involves $K$ LMI constraints with dimension $\mathcal{O}(L+K+N+1)$, the computational complexity required for solving the primal problem in each iteration is given by $\mathcal{O}(\log\frac{1}{\rho}(K(L+K+N+1)^3+K^2(L+K+N+1)^2+K^3)$, where $\mathcal{O}(\cdot)$ is the big-O notation, and $\rho>0$ is the convergence tolerance of the interior point method \cite[Theorem 3.12]{bomze2010interior}.\footnote{\color{black}According to \cite{bomze2010interior}, the computational complexity entailed by LMI constraints is much higher than that entailed by SOC constraints. Thus, the complexity of the $K$ SOC constraints in \eqref{Primal_problem} is ignored in the analysis.} On the other hand, the master problem in \eqref{Master_problem} is solved by the branch-and-bound (BnB) method adopted in current numerical MILP solvers \cite{mosek}. For solving the master problem in the $i$-th iteration of \textbf{Algorithm 1}, in each BnB iteration, a linear programming problem that involves $i$ linear constraints and $(L\times N+1)$ optimization variables has to be solved. Thus, the computational complexity of solving the master problem in the $i$-th iteration is given by $\mathcal{O}(I_{\mathrm{BnB}}(L\times N+1+i)\times (L\times N+1))$, where $I_{\mathrm{BnB}}$ denotes the number of BnB iterations\footnote{\color{black}The actual number of BnB iterations $I_{\mathrm{BnB}}$ can be significantly reduced by employing an intelligent branching strategy. Nevertheless, the worst-case number of BnB iterations $I_{\mathrm{BnB}}$ scales exponentially with the number of IRS elements \cite{mosek}.}.}  Based on our simulation results in Section \rom{5}, we observe that the proposed GBD-based \textbf{Algorithm 1} converges in significantly fewer iterations than an exhaustive search, although in the worst-case, the number of iterations required by GBD algorithms also scales exponentially with the number of IRS elements \cite{geoffrion1972generalized}.
\end{enumerate}
\vspace*{-4mm}
\subsection{Proposed SCA-Based Algorithm}
The proposed GBD-based algorithm provides a performance upper bound for the considered IRS-assisted system. However, the worst-case computational complexity of the proposed GBD-based algorithm is non-polynomial. In this section, we introduce a suboptimal resource allocation algorithm that has only polynomial time computational complexity. We note that the major difficulty in solving \eqref{Reformulated_problem} is binary constraint C2b. To circumvent this difficulty, we recast C2b equivalently in form of the following two inequality constraints
\begin{equation}
\begin{aligned}
       \overline{\mbox{C2b}}:&\hspace*{1mm} \sum_{n=1}^N\sum_{l=1}^L (b_n[l]-b_n^2[l])\leq 0,\\
       \overline{\mbox{C2c}}:&\hspace*{1mm} 0\leq b_n[l]\leq 1, \hspace*{1mm}\forall l,\forall n,
\end{aligned}
\end{equation}
where constraint $\overline{\mbox{C2b}}$ is a DC constraint. Then, problem \eqref{Reformulated_problem} can be rewritten as follows
\begin{eqnarray}
\label{DC_problem}
    &&\hspace*{-6mm}\underset{{\bm{\Gamma},\,\mathbf{B}}}{\mino}\hspace*{2mm}\sum_{k\in\mathcal{K}}\left\|\mathbf{w}_k\right\|_2^2\notag\\
    &&\hspace*{0mm}\mbox{s.t.}\hspace*{8mm}\mbox{C1a},\mbox{C1b},\mbox{C2a},\overline{\mbox{C2b}},\overline{\mbox{C2c}},\mbox{C3a},\overline{\mbox{C3b}}.
\end{eqnarray}
To facilitate a low-complexity algorithm design, we resort to the penalty method \cite{ng2015secure} to tackle the DC constraint $\overline{\mbox{C2b}}$ and rewrite the problem in \eqref{DC_problem} as follows 
\begin{eqnarray}
\label{Penalty_problem}
    &&\hspace*{-6mm}\underset{{\bm{\Gamma},\,\mathbf{B}}}{\mino}\hspace*{2mm}\sum_{k\in\mathcal{K}}\left\|\mathbf{w}_k\right\|_2^2+\frac{1}{\mu}\sum_{n=1}^N\sum_{l=1}^L (b_n[l]-b_n^2[l])\notag\\
    &&\hspace*{0mm}\mbox{s.t.}\hspace*{8mm}\mbox{C1a},\mbox{C1b},\mbox{C2a},\overline{\mbox{C2c}},\mbox{C3a},\overline{\mbox{C3b}},
\end{eqnarray}
where $\mu>0$ is a penalty factor for penalizing the violation of the binary constraint. When $\mu$ is sufficiently small, i.e., $\mu\rightarrow 0$, problem \eqref{Reformulated_problem} and problem \eqref{Penalty_problem} are equivalent \cite{le2012exact}. Now, the objective function in \eqref{Penalty_problem} is in the form of a canonical DC problem. Hence, a suboptimal solution of \eqref{Penalty_problem} can be obtained by the SCA method \cite{le2012exact}. In particular, in the $(i+1)$-th iteration of the SCA algorithm, we construct a global underestimator for the term $\sum_{n=1}^N\sum_{l=1}^Lb_n^2[l]$ by leveraging its first-order Taylor approximation as follows:
\begin{equation}
    \sum_{n=1}^N\sum_{l=1}^Lb_n^2[l]\geq \sum_{n=1}^N\sum_{l=1}^L(2b_n^{(i)}[l]b_n[l]-(b_n^{(i)}[l])^2), 
\end{equation}
where $b_n^{(i)}[l]$ is the solution for $b_n[l]$ in the $i$-th iteration. As a result, the optimization problem in the $i$-th iteration of the SCA algorithm is given by
{\color{black}
    \begin{eqnarray}
\label{SCA_problem}
    &&\hspace*{-6mm}\underset{{\bm{\Gamma},\,\mathbf{B}}}{\mino}\hspace*{2mm}\sum_{k\in\mathcal{K}}\left\|\mathbf{w}_k\right\|_2^2\notag\\
    &&\hspace*{10mm}+\frac{1}{\mu}\sum_{n=1}^N\sum_{l=1}^L (b_n[l]-2b_n^{(i)}[l]b_n[l]+(b_n^{(i)}[l])^2)\notag\\
    &&\hspace*{0mm}\mbox{s.t.}\hspace*{8mm}\mbox{C1a},\mbox{C1b},\mbox{C2a},\overline{\mbox{C2c}},\mbox{C3a},\overline{\mbox{C3b}}.
\end{eqnarray}}
Problem \eqref{SCA_problem} is convex and can be optimally solved by standard convex program solvers such as CVX. The proposed iterative SCA algorithm is summarized in \textbf{Algorithm 2}. Some further remarks are given as follows:
\begin{enumerate}
    \item \textit{Initial point:} A matrix $\mathbf{B}$ satisfying constraints $\mbox{C2a}$ and $\overline{\mbox{C2c}}$ is randomly generated.
    \item {\color{black}\textit{Optimality and convergence:} 
    Solving problem \eqref{SCA_problem} provides an upper bound for the original resource allocation problem in \eqref{Penalty_problem}. By solving \eqref{SCA_problem} iteratively, we can gradually tighten this upper bound. We note that the proposed suboptimal algorithm converges to a locally optimal solution of \eqref{Penalty_problem} \cite{razaviyayn2013unified}. Moreover, the convergence of \textbf{Algorithm 2} depends on the selection of the penalty factor $1/\mu$. In particular, \textbf{Algorithm 2} is guaranteed to converge to a binary matrix $\mathbf{B}$ for a sufficiently small $\mu$ \cite[Remark 4]{pham2010efficient}. However, how to determine the optimal value of $\mu$ is still an open problem. In this work, we adopt the method proposed in \cite[Remark 7]{pham2010efficient} to determine the value of $\mu$. Specifically, $\mu=10^{-3}$ is first chosen to formulate the optimization problem in \eqref{SCA_problem}. If the proposed algorithm does not converge to a binary matrix $\mathbf{B}$, the value of $\mu$ is further reduced until a binary matrix $\mathbf{B}$ is obtained. }
    \item \textit{Complexity:} In each iteration, we only need to solve a convex programming problem. Thus, \textbf{Algorithm 2} exhibits polynomial computational complexity \cite{razaviyayn2013unified}. In particular, since the problem in \eqref{SCA_problem} involves $K$ LMI constraints with dimension $\mathcal{O}(L+K+N+1)$, the computational complexity of each iteration of the proposed SCA-based algorithm is given by $\mathcal{O}(\log\frac{1}{\rho}(K(L+K+N+1)^3+K^2(L+K+N+1)^2+K^3)$, where $\mathcal{O}(\cdot)$ is the big-O notation, and $\rho>0$ is the convergence tolerance of the interior point method\cite[Theorem 3.12]{bomze2010interior}
    
\end{enumerate}
\begin{algorithm}[t]
\caption{Successive Convex Approximation (SCA) Algorithm}
\begin{algorithmic}[1]
\small
\STATE Set iteration index $i=0$ and generate a feasible $\mathbf{B}^{(0)}$. Set convergence tolerance $0<\Delta_{\mathrm{SCA}}\ll 1$ and penalty factor $0<\mu\ll 1$
\REPEAT
\STATE Set $i=i+1$
\STATE Solve \eqref{SCA_problem} for a given $\mathbf{B}^{(i-1)}$ and update $\mathbf{B}^{(i)}$ as the optimal solution of \eqref{SCA_problem}
\UNTIL $\frac{\|\mathbf{B}^{(i-1)}-\mathbf{B}^{(i)}\|_F}{\|\mathbf{B}^{(i-1)}\|_F}\leq \Delta_{\mathrm{SCA}}$
\end{algorithmic}
\end{algorithm}
\vspace*{-5mm}
\section{Algorithm Design for Imperfect CSI}
\begin{figure*}[!htbp]
    \centering
    \includegraphics[width=6.8in]{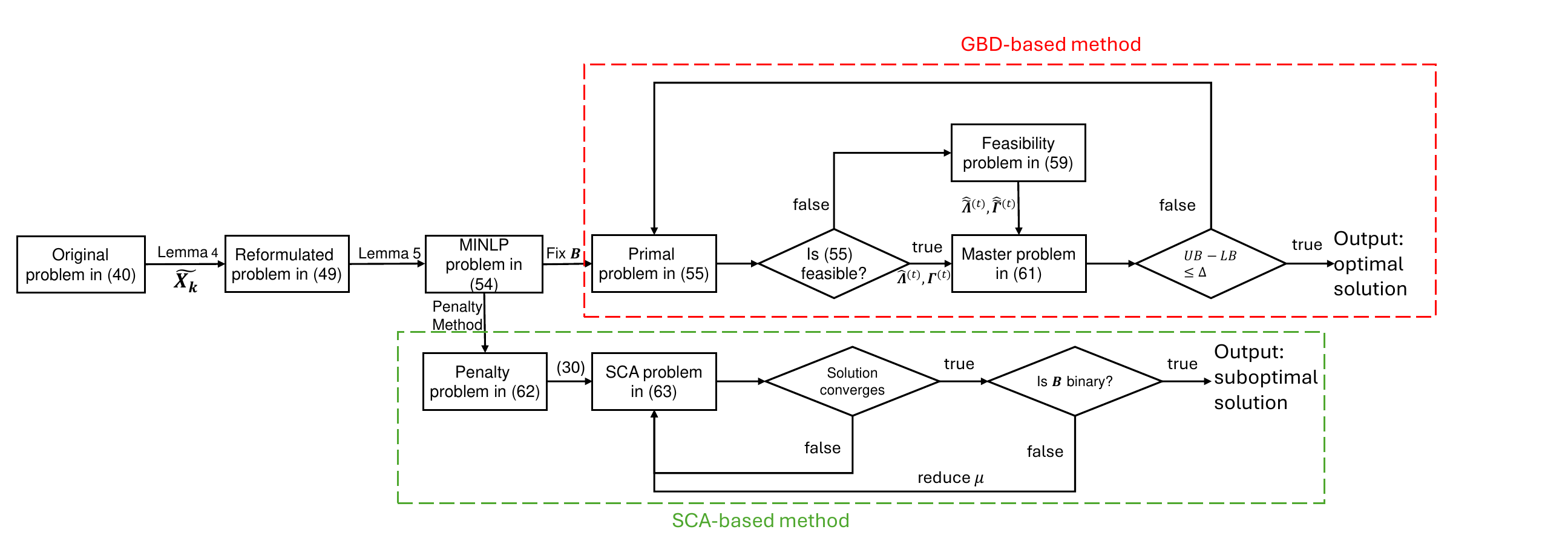}
    \caption{\color{black}Flowchart of the algorithms proposed for the case of imperfect CSI.}
    \label{fig:FC_imperfect_CSI}
\end{figure*}
{\color{black}In this section, we first introduce the adopted CSI uncertainty model and reformulate the resource allocation problem in \eqref{Ori_Problem} to account for imperfect CSI. In particular, we aim to minimize the BS transmit power while guaranteeing a worst-case received SINR for each user. Then, we suitably modify the GBD-based and SCA-based methods proposed in Section \rom{3}, respectively, to accommodate the new robust resource allocation problem. An overview of the individual steps of the derivation and flow of the proposed algorithms is provided in Fig. \ref{fig:FC_imperfect_CSI}.}

\vspace*{-5mm}
\subsection{Imperfect CSI Model}
As it is challenging to estimate the individual IRS-assisted reflected channels $\mathbf{F}$ and $\mathbf{h}_k$, we assume the cascaded reflected channel from the BS to the users, $\mathbf{E}_k=\mathrm{diag}(\mathbf{h}_k^H)\mathbf{F}, \forall k\in\mathcal{K}$, and the direct channel between the BS and the users, $\mathbf{d}_k,\forall k\in\mathcal{K}$, are estimated, see \cite{hu2021robust,you2020channel,lin2022channel} for details. In this work, we adopt the norm-bounded CSI error model for both the cascaded reflected channels and the direct channels \cite{ng2014robust,xu2022robust}. Specifically, the cascaded reflected channel $\mathbf{E}_k$ and the direct channel $\mathbf{d}_k$ for user $k$ are modeled as follows:
\begin{equation}\label{CSI_model}
    \begin{aligned}
        &\mathbf{E}_k=\Bar{\mathbf{E}}_k+\Delta\mathbf{E}_k,\, \mathbf{d}_k=\Bar{\mathbf{d}}_k+\Delta\mathbf{d}_k,\\
        &\Omega_{\mathbf{E}_k}=\{\Delta\mathbf{E}_k\in\mathbb{C}^{N\times M}:\|\Delta\mathbf{E}_k\|_F\leq \epsilon_{\mathbf{E}_k}\}, \forall k\in\mathcal{K},\\
        &\Omega_{\mathbf{d}_k}=\{\Delta\mathbf{d}_k\in\mathbb{C}^{M}:\|\Delta\mathbf{d}_k\|_2\leq \epsilon_{\mathbf{d}_k}\} , \forall k\in\mathcal{K},
    \end{aligned}
\end{equation}
where $\Bar{\mathbf{E}}_k$ and $\Bar{\mathbf{d}}_k$ denote the estimates of the cascaded reflected channel and the direct channel of user $k$, respectively. Furthermore, $\Delta\mathbf{E}_k$ and $\Delta\mathbf{d}_k$ are the CSI estimation errors of the cascaded reflected channel and the direct channel of user $k$, whose norms are bounded by constants $\epsilon_{\mathbf{E}_k}$ and $\epsilon_{\mathbf{d}_k}$, respectively. The sets $\Omega_{\mathbf{E}_k}$ and $\Omega_{\mathbf{d}_k}$ contain all possible CSI errors satisfying the bounded norm condition.
\vspace{-5mm}\subsection{Problem Formulation}
To facilitate the robust resource allocation design, we first rewrite the received signal of user $k$, $y_k$, as follows
\begin{equation}
\begin{aligned}
    y_k=\mathbf{v}^T\widehat{\mathbf{H}}_k\mathbf{W}\mathbf{s}+n_k,
\end{aligned}
\end{equation}
where $\hat{\mathbf{H}}_k=[\mathbf{E}_k^H, \mathbf{d}_k]^H$ denotes the effective channel between the BS and user $k$. Recalling \eqref{CSI_model}, the channel matrix $\widehat{\mathbf{H}}_k$ can be modelled as follows
\begin{equation}
\begin{aligned}
\widehat{\mathbf{H}}_k=\widehat{\bar{\mathbf{H}}}_k+\Delta\widehat{\mathbf{H}}_k,
\end{aligned}
\end{equation}
where $\widehat{\bar{\mathbf{H}}}_k=[\Bar{\mathbf{E}}_k^H, \Bar{\mathbf{d}}_k]^H$ and $\Delta\widehat{\mathbf{H}}_k=[\Delta\mathbf{E}_k^H, \Delta\mathbf{d}_k]^H$. Then, we can obtain an upper bound for the channel uncertainty $\Delta \widehat{\mathbf{H}}_k$ \cite{yu2021robust}
\begin{equation}
    \|\Delta\widehat{\mathbf{H}}_k\|_F\leq \sqrt{\epsilon_{\mathbf{E}_k}^2+\epsilon_{\mathbf{d}_k}^2}=\epsilon_k.
\end{equation}
{\color{black}By using $\mathrm{vec}(\mathbf{AXB})^T=\mathrm{vec}^T(\mathbf{X})(\mathbf{B}\otimes\mathbf{A}^T)$, the received signal $y_k$ can be rewritten as
\begin{equation}\label{robust_receiver}
    \begin{aligned}
        y_k&=\mathrm{vec}^T(\widehat{\mathbf{H}}_k)(\mathbf{I}_M\otimes\mathbf{v})\mathbf{W}\mathbf{s}+n_k\\
        &=\mathbf{g}_k^T\mathbf{V}\mathbf{W}\mathbf{s}+n_k,
    \end{aligned}
\end{equation}
where $\mathbf{V}=\mathbf{I}_M\otimes\mathbf{v}\in\mathbb{C}^{(N+1)M\times M}$ and $\mathbf{g}_k=\mathrm{vec}(\widehat{\mathbf{H}}_k)$. Next, $\mathbf{V}$ is rewritten as follows 
\begin{equation}
    \mathbf{V}=\mathbf{I}_M\otimes({\mathbf{B}}^T\bm{\theta})=(\mathbf{I}_M\otimes{\mathbf{B}}^T)(\mathbf{I}_M\otimes\bm{\theta})=\Bar{\mathbf{B}}\bm{\Theta},
\end{equation}
where $\Bar{\mathbf{B}}$ and $\bm{\Theta}$ are defined as $\Bar{\mathbf{B}}=\mathbf{I}_M\otimes{\mathbf{B}}^T$ and $\bm{\Theta}=\mathbf{I}_M\otimes\bm{\theta}$, respectively.}
We further rewrite effective channel vector $\mathbf{g}_k$ as
\begin{equation}
    \mathbf{g}_k=\bar{\mathbf{g}}_k+\Delta\mathbf{g}_k,
\end{equation}
where $\bar{\mathbf{g}}_k=\mathrm{vec}(\widehat{\bar{\mathbf{H}}}_k)$, $\Delta\mathbf{g}_k=\mathrm{vec}(\Delta\widehat{\mathbf{H}}_k)$, and $\|\Delta\mathbf{g}_k\|_2=\|\Delta\widehat{\mathbf{H}}_k\|_F\leq \epsilon_k.$
Then, the received signal $y_k$ can be recast as follows
\begin{equation}
    y_k=\mathbf{g}_k^T\Bar{\mathbf{B}}\bm{\Theta}\mathbf{W}\mathbf{s}+n_k.
\end{equation}
With $y_k$ in hand, we can formulate the resource allocation problem taking into account the channel uncertainty.
Considering the CSI uncertainty model in \eqref{CSI_model}, in order to obtain a robust joint BS beamforming and IRS phase shift policy, the resource allocation problem in \eqref{Ori_Problem} is extended as follows
\begin{eqnarray}
\label{Ori_Problem_robust}
    &&\hspace*{-6mm}\underset{\mathbf{W},\,{\mathbf{B}}}{\mino}\hspace*{2mm}\sum_{k\in\mathcal{K}}\left\|\mathbf{w}_k\right\|_2^2\notag\\
    &&\hspace*{-4mm}\mbox{s.t.}\hspace*{2mm} \widehat{\mbox{C1}}:\hspace*{-1mm} \min_{\|\Delta\mathbf{g}_k\|_2\leq \epsilon_k}\frac{|\mathbf{g}_k^H\Bar{\mathbf{B}}\bm{\Theta}\mathbf{w}_k|^2}{\sum_{k'\in\mathcal{K}\setminus\{k\}}|\mathbf{g}_k^H\Bar{\mathbf{B}}\bm{\Theta}\mathbf{w}_{k'}|^2+\sigma_k^2}\geq \gamma_k,\forall k,\notag\\
    &&\hspace*{2mm}\mbox{C2a},{\mbox{C2b}}.
\end{eqnarray}
{\color{black}Since the worst-case QoS constraint $\widehat{\mbox{C1}}$ comprises infinitely many non-convex inequality constraints introduced by the continuity of the CSI uncertainty set, the worst-case QoS constraint cannot be recast as a series of convex constraints as was done for the perfect CSI cases in Section \rom{3}-A. Therefore, in the following, a series of transformations and auxiliary variables are introduced to reformulate the robust resource allocation problem as a more tractable MINLP problem, as illustrated in Fig. \ref{fig:FC_imperfect_CSI}.} In particular, the numerator of SINR constraint $\widehat{\mbox{C1}}$ can be rewritten as follows
\begin{equation}
    |\mathbf{g}_k^H\Bar{\mathbf{B}}\bm{\Theta}\mathbf{w}_k|^2=\mathbf{g}_k^H\Bar{\mathbf{B}}\bm{\Theta}\mathbf{W}_k\bm{\Theta}^H\Bar{\mathbf{B}}^H\mathbf{g}_k,
\end{equation}
where $\mathbf{W}_k=\mathbf{w}_k\mathbf{w}_k^H$.
The denominator can be rewritten in a similar manner and we can recast constraint $\widehat{\mbox{C1}}$ as follows
\begin{equation}\label{Reform_hat_C1}
\begin{aligned}
    \widehat{\mbox{C1}}\Leftrightarrow \min_{\|\Delta\mathbf{g}_k\|_2\leq\epsilon_k}\sigma_k^2\gamma_k+\mathbf{g}_k^H\Bar{\mathbf{B}}\bm{\Theta}\widetilde{\mathbf{W}}_k\bm{\Theta}^H\Bar{\mathbf{B}}^H\mathbf{g}_k\leq 0, 
\end{aligned}
\end{equation}
where $\widetilde{\mathbf{W}}_k=\gamma_k\sum_{k'\in\mathcal{K}\setminus\{k\}}{\mathbf{W}}_{k'}-{\mathbf{W}}_k$.
Next, by substituting $\mathbf{g}_k=\bar{\mathbf{g}}_k+\Delta\mathbf{g}_k$ into \eqref{Reform_hat_C1},  we can rewrite constraint $\widehat{\mbox{C1}}$ as follows
\begin{eqnarray}
\label{Reform_hat_C1_2}
\widehat{\mbox{C1}}\Leftrightarrow&&\hspace*{-6mm}\Bar{\mathbf{g}}_k^H\Bar{\mathbf{B}}\bm{\Theta}\widetilde{\mathbf{W}}_k\bm{\Theta}^H\Bar{\mathbf{B}}^H\Bar{\mathbf{g}}_k+2\mathrm{Re}\{\Bar{\mathbf{g}}_k^H\Bar{\mathbf{B}}\bm{\Theta}\widetilde{\mathbf{W}}_k\bm{\Theta}^H\Bar{\mathbf{B}}^H\Delta{\mathbf{g}}_k\}\notag\\
+&&\hspace*{-6mm}\Delta{\mathbf{g}}_k^H\Bar{\mathbf{B}}\bm{\Theta}\widetilde{\mathbf{W}}_k\bm{\Theta}^H\Bar{\mathbf{B}}^H\Delta{\mathbf{g}}_k\leq -\sigma_k^2\gamma_k,\|\Delta\mathbf{g}_k\|_2\leq \epsilon_k.  
\end{eqnarray}
Now, we transform the constraint in \eqref{Reform_hat_C1_2} into an LMI constraint by exploiting the following lemma.
{\color{black}\begin{lemma}\vspace{-5mm}\label{S_procedure}
    (S-Procedure \cite{boyd2004convex}): Let $g_1(\mathbf{x})$ and $g_2(\mathbf{x})$ be real-valued functions of vector $\mathbf{x}\in\mathbb{C}^{J\times 1}$ and be defined as follows:
    \begin{equation}
        \begin{aligned}
            g_1(\mathbf{x})&=\mathbf{x}^H\mathbf{A}_1\mathbf{x}+2\operatorname{Re}\{\mathbf{a}_1^H\mathbf{x}\}+a_1,\\
            g_2(\mathbf{x})&=\mathbf{x}^H\mathbf{A}_2\mathbf{x}+2\operatorname{Re}\{\mathbf{a}_2^H\mathbf{x}\}+a_2,
        \end{aligned}
    \end{equation}
    where $\mathbf{A}_1$, $\mathbf{A}_2\in\mathbb{H}^J$, $\mathbf{a}_1$, $\mathbf{a}_2\in\mathbb{C}^{J\times 1}$, and $a_1,a_2\in\mathbb{R}$. Assume there exists a point $\hat{\mathbf{x}}$ with $\hat{\mathbf{x}}\mathbf{A}_1\mathbf{x}+2\mathrm{Re}\{\mathbf{a}_1^H\mathbf{x}\}+a_1<0$. Then, the implication $g_1(\mathbf{x})\leq 0\Rightarrow g_2(\mathbf{x})\leq 0$ holds if and only if there exists some real number $\lambda\geq 0$ such that
    \begin{equation}
        \lambda
    \begin{bmatrix}
        \mathbf{A}_1 & \mathbf{a}_1 \\
        \mathbf{a}_1^H & a_1
    \end{bmatrix}-
     \begin{bmatrix}
        \mathbf{A}_2 & \mathbf{a}_2 \\
        \mathbf{a}_2^H & a_2
    \end{bmatrix}\succeq \mathbf{0}.
    \end{equation}\vspace{-5mm}
\end{lemma}}
{\color{black}
By leveraging Lemma \ref{S_procedure}, constraint $\widehat{\mbox{C1}}$ is rewritten as an LMI constraint as follows
\begin{equation}
\label{LMI_SINR}
\begin{aligned}
    &q_k\begin{bmatrix}
       \mathbf{I}_{M(N+1)} & \mathbf{0} \\
        \mathbf{0} & -\epsilon_k^2
    \end{bmatrix}-\\
    &\begin{bmatrix}
       \Bar{\mathbf{B}}\bm{\Theta}\widetilde{\mathbf{W}}_k\bm{\Theta}^H\Bar{\mathbf{B}}^H & \Bar{\mathbf{B}}\bm{\Theta}\widetilde{\mathbf{W}}_k\bm{\Theta}^H\Bar{\mathbf{B}}^H\Bar{\mathbf{g}}_k \\
        \Bar{\mathbf{g}}_k^H\Bar{\mathbf{B}}\bm{\Theta}\widetilde{\mathbf{W}}_k\bm{\Theta}^H\Bar{\mathbf{B}}^H & \Bar{\mathbf{g}}_k^H\Bar{\mathbf{B}}\bm{\Theta}\widetilde{\mathbf{W}}_k\bm{\Theta}^H\Bar{\mathbf{B}}^H\Bar{\mathbf{g}}_k+\sigma_k^2\gamma_k
    \end{bmatrix}\succeq \mathbf{0},\\
    &\Leftrightarrow\mathbf{P}_k-
    \mathbf{G}_k^H\Bar{\mathbf{B}}\bm{\Theta}\widetilde{\mathbf{W}}_k\bm{\Theta}^H\Bar{\mathbf{B}}^H\mathbf{G}_k\succeq \mathbf{0}, \forall k,
\end{aligned}
\end{equation}
where 
\begin{equation}
    \begin{aligned}
        \mathbf{P}_k&=
    \begin{bmatrix}
        q_k\mathbf{I}_{M(N+1)} & \mathbf{0} \\
        \mathbf{0} & -q_k\epsilon_k^2-\sigma_k^2\gamma_k
    \end{bmatrix},\\
    \mathbf{G}_k&=[\mathbf{I}_{M(N+1)}\quad \Bar{\mathbf{g}}_k],
    \end{aligned}
\end{equation}
and $q_k\geq 0$ is an auxiliary variable.}
{\color{black}Next, we define a new optimization variable $\hat{\mathbf{X}}_k=\Bar{\mathbf{B}}\bm{\Theta}{\mathbf{W}}_k\bm{\Theta}^H\Bar{\mathbf{B}}^H$ to rewrite the constraint in \eqref{LMI_SINR} as follows
\begin{equation}
\begin{aligned}
    \widehat{\overline{\mbox{C1}}}:\hspace*{1mm} 
    &\mathbf{P}_k-
    \mathbf{G}_k^H\left(\gamma_k\sum_{k'\in\mathcal{K}\setminus\{k\}}\hat{\mathbf{X}}_{k'}-\hat{\mathbf{X}}_k\right)\mathbf{G}_k\succeq \mathbf{0}, \forall k,\\
    \Leftrightarrow\hspace*{1mm} &\mathbf{P}_k-
    \mathbf{G}_k^H\widetilde{\mathbf{X}}_k\mathbf{G}_k\succeq \mathbf{0}, \forall k,
\end{aligned}
\end{equation}
where $\widetilde{\mathbf{X}}_k=\gamma_k\sum_{k'\in\mathcal{K}\setminus\{k\}}\hat{\mathbf{X}}_{k'}-\hat{\mathbf{X}}_k$. 
Note that constraint $\widehat{\overline{\mbox{C1}}}$ is convex w.r.t. variables $\hat{\mathbf{X}}_{k}, \forall k$. Then,
\eqref{Ori_Problem_robust} can be recast into the following equivalent optimization problem:
\begin{eqnarray}
\label{SDR_Problem_robust}
    &&\hspace*{-12mm}\underset{\hat{\mathbf{X}},\mathbf{W},{\mathbf{B}}}{\mino}\hspace*{2mm}\sum_{k\in\mathcal{K}}\left\|\mathbf{w}_k\right\|_2^2\notag\\
    &&\hspace*{-6mm}\mbox{s.t.}\hspace*{4mm} \widehat{\overline{\mbox{C1}}},\mbox{C2a},\mbox{C2b},\hspace*{1mm}\widehat{\overline{\mbox{C3}}}\mbox{:}\hspace*{1mm}\hat{\mathbf{X}}_k=\Bar{\mathbf{B}}\bm{\Theta}{\mathbf{W}}_k\bm{\Theta}^H\Bar{\mathbf{B}}^H,\hspace*{1mm}\forall k,
\end{eqnarray}
where $\hat{\mathbf{X}}=[\hat{\mathbf{X}}_{1},\cdots, \hat{\mathbf{X}}_{K}]$ denotes the collection of all $\hat{\mathbf{X}}_k$, $\forall k$.
Note that $\widehat{\overline{\mbox{C3}}}$ is a non-convex constraint due to the coupled variables ${\mathbf{W}}$ and $\bar{\mathbf{B}}$. In particular, since equality constraint $\widehat{\overline{\mbox{C3}}}$ is quadratic w.r.t. $\bar{\mathbf{B}}$, $\widehat{\overline{\mbox{C3}}}$ is neither a bilinear nor a biconvex constraint. To the best of our knowledge, there is no computationally efficient general method to handle such non-convex constraints. However, thanks to the binary nature of $\bar{\mathbf{B}}$, we are able to convexify constraint $\widehat{\overline{\mbox{C3}}}$ by exploiting the following lemma.
\setcounter{lemma}{4}
\begin{lemma}
    Equality constraint $\widehat{\overline{\mbox{C3}}}$ is equivalent to the following LMI constraints:
    \begin{eqnarray}\label{sdp_robust}
\widehat{\overline{\mathrm{C3a}}}&\mbox{:}&\hspace*{1mm}
   \begin{bmatrix}
        \hat{\mathbf{S}}_k & \hat{\mathbf{X}}_k & \Bar{\mathbf{B}}\bm{\Theta}\\
        \hat{\mathbf{X}}_k^H & \hat{\mathbf{T}}_k & \mathbf{Y}_k\\
        \bm{\Theta}^H\Bar{\mathbf{B}}^H & \mathbf{Y}_k^H & \mathbf{I}_M
    \end{bmatrix}\succeq \mathbf{0},\notag\\
    \widehat{\overline{\mathrm{C3b}}}&\mbox{:}&\hspace{1mm}\mathrm{Tr}\big(\hat{\mathbf{S}}_k\big)-(N+1)M\leq 0,\notag\\
\widehat{\overline{\mathrm{C3c}}}&\mbox{:}&\hspace{1mm}  \begin{bmatrix}
        \mathbf{U}_k & {\mathbf{Y}}_k & \Bar{\mathbf{B}}\bm{\Theta}\\
        {\mathbf{Y}}_k^H & \mathbf{V}_k & {\mathbf{W}}_k\\
        \bm{\Theta}^H\Bar{\mathbf{B}}^H & {\mathbf{W}}_k^H & \mathbf{I}_M
    \end{bmatrix}\succeq \mathbf{0},\notag\\
    \widehat{\overline{\mathrm{C3d}}}&\mbox{:}&\hspace{1mm}\mathrm{Tr}\big({\mathbf{U}}_k\big)-(N+1)M\leq 0,
\end{eqnarray}
where $\hat{\mathbf{S}}_k\in\mathbb{C}^{(N+1)M\times (N+1)M}$, $\hat{\mathbf{T}}_k\in\mathbb{C}^{(N+1)M\times (N+1)M}$, $\mathbf{Y}_k\in\mathbb{C}^{(N+1)M\times M}$, $\mathbf{U}_k\in\mathbb{C}^{(N+1)M\times (N+1)M}$, and $\mathbf{V}_k\in\mathbb{C}^{M\times M}$ are auxiliary variables.
\end{lemma}}
\begin{proof}
    We first define an auxiliary variable $\mathbf{Y}_k$ as $\mathbf{Y}_k=\Bar{\mathbf{B}}\bm{\Theta}{\mathbf{W}}_k^H$. Then, constraint $\widehat{\overline{\mbox{C3}}}$ can be recast into the following bilinear constraints
    \begin{equation}
        \widehat{\overline{\mbox{C3}}}\Leftrightarrow\hat{\mathbf{X}}_k=\Bar{\mathbf{B}}\bm{\Theta}{\mathbf{Y}}_k^H,\quad\mathbf{Y}_k=\Bar{\mathbf{B}}\bm{\Theta}{\mathbf{W}}_k^H.
    \end{equation}
    Based on Lemma 1, $\hat{\mathbf{X}}_k=\Bar{\mathbf{B}}\bm{\Theta}{\mathbf{Y}}_k$ can be rewritten as constraint $\widehat{\overline{\mbox{C3a}}}$ and inequality constraint
    \begin{equation}\label{C2b_original}
        \mathrm{Tr}\big(\hat{\mathbf{S}}_k-\Bar{\mathbf{B}}\bm{\Theta}\bm{\Theta}^H\Bar{\mathbf{B}}^H\big)\leq 0.
    \end{equation}
    According to the property of the Kronecker product and the binary nature of $\Bar{\mathbf{B}}$, we can rewrite \eqref{C2b_original} as the following linear inequality constraint
\begin{equation}\label{C2b}
\begin{aligned}
    &\hspace*{1mm}\mathrm{Tr}\big(\widehat{\mathbf{S}}_k-\Bar{\mathbf{B}}\bm{\Theta}\bm{\Theta}^H\Bar{\mathbf{B}}^H\big)=\mathrm{Tr}\big(\widehat{\mathbf{S}}_k-\mathbf{V}\mathbf{V}^H\big)\leq 0\\
    \Leftrightarrow&\hspace*{1mm}\widehat{\overline{\mbox{C3b}}}:\mathrm{Tr}\big(\widehat{\mathbf{S}}_k\big)-(N+1)M\leq 0.
\end{aligned}
\end{equation}
Similar to the above approach, we can recast constraint $\mathbf{Y}_k=\Bar{\mathbf{B}}\bm{\Theta}{\mathbf{W}}_k^H$ in terms of constraints $\widehat{\overline{\mbox{C3c}}}$ and $\widehat{\overline{\mbox{C3d}}}$,  which completes the proof.
\end{proof}
To simplify our notation, we define $\hat{\mathbf{W}}=[\mathbf{W}_1,\cdots,\mathbf{W}_K]$, $\widehat{\mathbf{S}}=[\hat{\mathbf{S}}_{1},\cdots, \hat{\mathbf{S}}_{K}]$, $\widehat{\mathbf{T}}=[\hat{\mathbf{T}}_{1},\cdots, \hat{\mathbf{T}}_{K}]$, $\mathbf{U}=[\mathbf{U}_1,\cdots,\mathbf{U}_K]$, $\mathbf{V}=[\mathbf{V}_1,\cdots,\mathbf{V}_K]$, and $\mathbf{Y}=[\mathbf{Y}_1,\cdots,\mathbf{Y}_K]$ to collect all ${\mathbf{W}}_{k}$, $\hat{\mathbf{S}}_{k}$, $\hat{\mathbf{T}}_{k}$, $\mathbf{U}_k$, $\mathbf{V}_k$, and $\mathbf{Y}_k$, $\forall k \in\mathcal{K}$, respectively. Hence, the robust resource allocation problem in \eqref{SDR_Problem_robust} can now be recast as
\begin{eqnarray}
\label{Ref_Problem_robust}
    &&\hspace*{-4mm}\underset{\substack{\widetilde{\mathbf{X}},\Bar{\mathbf{B}},\hat{\mathbf{S}},\hat{\mathbf{T}},\hat{\mathbf{W}}\\\mathbf{U},\mathbf{V},\mathbf{Y},\mathbf{q}}}{\mino}\hspace*{2mm}\sum_{k\in\mathcal{K}}\mathrm{Tr}(\mathbf{W}_k)\notag\\
    &&\hspace*{2mm}\mbox{s.t.}\hspace*{7mm} \widehat{\overline{\mbox{C1}}},\mbox{C2a},{\mbox{C2b}},\widehat{\overline{\mbox{C3a}}},\widehat{\overline{\mbox{C3b}}},\widehat{\overline{\mbox{C3c}}},\widehat{\overline{\mbox{C3d}}},\notag\\
    &&\hspace*{14mm} \mbox{C7}: \mathrm{rank}(\mathbf{W}_k)\leq 1,\forall k\in\mathcal{K},\\
    &&\hspace*{14mm}\mbox{C8:} \mathbf{W}_k\succeq\mathbf{0},\forall k\in\mathcal{K}\notag,
\end{eqnarray}
with vector $\mathbf{q}=[q_1,\cdots,q_K]^T$. Note that \eqref{Ref_Problem_robust} is an MINLP problem, where the non-convexity originates from the binary constraint C2b and the rank-one constraint C7.  
To facilitate the derivation of the globally optimal algorithm based on the GBD method, we first tackle the rank-one constraint C7 by semidefinite relaxation (SDR). The tightness of the relaxation is revealed in the following theorem.
\vspace*{-2mm}
\begin{theorem}
    For all feasible binary selection matrices $\bar{\mathbf{B}}$, there always exists an optimal beamforming matrix $\mathbf{W}_k^{\mathrm{opt}}$ with $\mathrm{rank}(\mathbf{W}_k^{\mathrm{opt}})\leq 1, \forall k$. 
\end{theorem}
\begin{proof}
Please refer to the appendix.
\end{proof}
\begin{remark}\vspace*{-3mm}
    By omitting the rank-one constraint C7, problem \eqref{Ref_Problem_robust} becomes convex w.r.t. continuous variables $\hat{\mathbf{X}},\hat{\mathbf{S}},\hat{\mathbf{T}},\hat{\mathbf{W}}$,$\mathbf{U},\mathbf{V},\mathbf{Y},$ and $\mathbf{q}$, whereas it is linear w.r.t. discrete matrix $\mathbf{B}$. Thus, the MINLP problem in \eqref{Ref_Problem_robust} can be optimally solved by applying a similar procedure as in \textbf{Algorithm 1}.
\end{remark}
 
\vspace*{-4mm}\subsection{Proposed GBD-based Algorithm}
We first decompose the rank-one relaxed version of \eqref{Ref_Problem_robust} into a primal problem and a master problem. For a given binary matrix $\Bar{\mathbf{B}}^{(i-1)}$, the primal problem is obtained as follows,
\begin{eqnarray}
\label{Master_Problem_robust}
    &&\hspace*{-4mm}\underset{\substack{\widetilde{\mathbf{X}},\hat{\mathbf{S}},\hat{\mathbf{T}},\hat{\mathbf{W}}\\\mathbf{U},\mathbf{V},\mathbf{Y},\mathbf{q}}}{\mino}\hspace*{2mm}\sum_{k\in\mathcal{K}}\mathrm{Tr}(\mathbf{W}_k)\notag\\
    &&\hspace*{2mm}\mbox{s.t.}\hspace*{7mm} \widehat{\overline{\mbox{C1}}},\widehat{\overline{\mbox{C3a}}},\widehat{\overline{\mbox{C3b}}}, \widehat{\overline{\mbox{C3c}}}, \widehat{\overline{\mbox{C3d}}},\mbox{C8}.
\end{eqnarray}
Note that \eqref{Master_Problem_robust} is convex w.r.t. $\hat{\mathbf{W}}, \hat{\mathbf{X}},\hat{\mathbf{S}},\hat{\mathbf{T}},\mathbf{U},\mathbf{V}$, $\mathbf{Y}$, and $\mathbf{q}$, and can be optimally solved by standard convex solvers such as CVX. For the sake of notational simplicity, we define $\widehat{\bm{\Gamma}}=\{\hat{\mathbf{X}},\hat{\mathbf{S}},\hat{\mathbf{T}},\hat{\mathbf{W}},\mathbf{U},\mathbf{V},\mathbf{Y},\mathbf{q}\}$ to collect all optimization variables.  Similar to the notations in Section \ref{Opt_perfect_CSI}, we let \\$\widehat{\bm{\Gamma}}^{(i)}=\left\{\hat{\mathbf{X}}^{(i)},\hat{\mathbf{S}}^{(i)},\hat{\mathbf{T}}^{(i)},\hat{\mathbf{W}}^{(i)},\mathbf{U}^{(i)},\mathbf{V}^{(i)},\mathbf{Y}^{(i)},\mathbf{q}^{(i)}\right\}$ denote the optimal solution of \eqref{Master_Problem_robust} in the $i$-th iteration. The Lagrangian function of \eqref{Master_Problem_robust} is given by
\begin{equation}
    \bar{\mathcal{L}}(\widehat{\bm{\Gamma}},\bar{\mathbf{B}},\widehat{\bm{\Lambda}})=\sum_{k\in\mathcal{K}}\mathrm{Tr}(\mathbf{W}_k)+\hat{f}_1(\widehat{\bm{\Gamma}},\widehat{\bm{\Lambda}})+\sum_{k\in\mathcal{K}}\hat{f}_2(\bar{\mathbf{B}},\widehat{\bm{\Lambda}}),
\end{equation}
where
$\hat{f}_1(\bm{\Gamma},\widehat{\bm{\Lambda}})$ collects the terms that are independent of $\bar{\mathbf{B}}$ and $\hat{f}_2(\bar{\mathbf{B}},\widehat{\bm{\Lambda}})$ is defined as
\begin{equation}
    \hat{f}_2(\bar{\mathbf{B}},\widehat{\bm{\Lambda}})\hspace*{-0.5mm}=\hspace*{-0.5mm}2\left[\mathrm{Re}\left\{\mathrm{Tr}(\bar{\mathbf{B}}\bm{\Theta}\hat{\mathbf{Q}}_{k,31})\right\}\hspace*{-0.5mm}+\hspace*{-0.5mm}\mathrm{Re}\left\{\mathrm{Tr}(\bar{\mathbf{B}}\bm{\Theta}\hat{\bm{\Xi}}_{k,31})\right\}\right].
\end{equation}
Here, $\widehat{\bm{\Lambda}}=\{\hat{\mathbf{Q}}_k,\hat{\bm{\Xi}}_k\}$ denotes the collection of dual variables, where $\hat{\mathbf{Q}}_k\in\mathbb{C}^{(2N+3)M\times (2N+3)M}$ and $\hat{\bm{\Xi}}_k\in\mathbb{C}^{(N+3)M\times (N+3)M}$ represent the dual matrices for constraints $\widehat{\overline{\mbox{C3a}}}$ and $\widehat{\overline{\mbox{C3c}}}$, respectively. Similar to the notations in \eqref{Qi}, 
the dual variable matrices $\hat{\mathbf{Q}}_k$ and $\hat{\bm{\Xi}}_k$ are decomposed as follows
\begin{equation}\label{Qi_Ei}
\setlength\arraycolsep{1.5pt}
  \hat{\mathbf{Q}}_k\hspace*{-0.5mm}=\hspace*{-0.5mm}\left[ \begin{array}{ccc}
        \hat{\mathbf{Q}}_{k,11} & \hat{\mathbf{Q}}_{k,21}^H & \hat{\mathbf{Q}}_{k,31}^H\\
        \hat{\mathbf{Q}}_{k,21} & \hat{\mathbf{Q}}_{k,22} & \hat{\mathbf{Q}}_{k,32}^H\\
        \hat{\mathbf{Q}}_{k,31} & \hat{\mathbf{Q}}_{k,32} & \hat{\mathbf{Q}}_{k,33}
    \end{array}\right],\, \hat{\bm{\Xi}}_k\hspace*{-0.5mm}=\hspace*{-0.5mm}\left[ \begin{array}{ccc}
        \hat{\bm{\Xi}}_{k,11} & \hat{\bm{\Xi}}_{k,21}^H & \hat{\bm{\Xi}}_{k,31}^H\\
       \hat{\bm{\Xi}}_{k,21} & \hat{\bm{\Xi}}_{k,22} & \hat{\bm{\Xi}}_{k,32}^H\\
       \hat{\bm{\Xi}}_{k,31} & \hat{\bm{\Xi}}_{k,32} & \hat{\bm{\Xi}}_{k,33}
    \end{array}\right],
\end{equation}
where $\hat{\mathbf{Q}}_{k,11}\in\mathbb{C}^{(N+1)M\times (N+1)M}$, $\hat{\mathbf{Q}}_{k,21}\in\mathbb{C}^{(N+1)M\times (N+1)M}$, $\hat{\mathbf{Q}}_{k,22}\in\mathbb{C}^{(N+1)M\times (N+1)M}$, $\hat{\mathbf{Q}}_{k,31}\in\mathbb{C}^{M\times (N+1)M}$, $\hat{\mathbf{Q}}_{k,32}\in\mathbb{C}^{M\times (N+1)M}$, $\hat{\mathbf{Q}}_{k,33}\in\mathbb{C}^{M\times M}$ and $\hat{\bm{\Xi}}_{k,11}\in\mathbb{C}^{(N+1)M\times (N+1)M}$, $\hat{\bm{\Xi}}_{k,21}\in\mathbb{C}^{M\times (N+1)M}$, $\hat{\bm{\Xi}}_{k,22}\in\mathbb{C}^{M\times M}$, $\hat{\bm{\Xi}}_{k,31}\in\mathbb{C}^{M\times (N+1)M}$, $\hat{\bm{\Xi}}_{k,32}\in\mathbb{C}^{M\times M}$, $\hat{\bm{\Xi}}_{k,33}\in\mathbb{C}^{M\times M}$ denote sub-matrices of $\hat{\mathbf{Q}}_k$ and $\hat{\bm{\Xi}}_k$, respectively.

On the other hand, if the primal problem in \eqref{Master_Problem_robust} is not feasible in the $i$-th iteration for the given $\bar{\mathbf{B}}^{(i-1)}$, we turn to solve the following feasibility-check problem:
\begin{eqnarray}
\label{Feasible_problem_robust}
    &&\hspace*{-6mm}\underset{\widehat{\bm{\Gamma}},\hat{\bm{\lambda}}}{\mino}\hspace*{2mm}\sum_{k\in\mathcal{K}}\hat{\lambda}_k\notag\\
    &&\hspace*{0mm}\mbox{s.t.}\hspace*{6mm}\widehat{\overline{\mbox{C3a}}},\widehat{\overline{\mbox{C3b}}},\widehat{\overline{\mbox{C3c}}},\widehat{\overline{\mbox{C3d}}},\mbox{C8}, \notag\\
    &&\hspace*{10mm}\widehat{\widetilde{\mbox{C1}}}\mbox{:}\hspace*{1mm} \mathbf{P}_k-
    \mathbf{G}_k^H\widetilde{\mathbf{X}}_k\mathbf{G}_k\succeq -\hat{\lambda}_k\mathbf{I}_{M(N+1)+1},\hspace*{1mm}\forall k\in \mathcal{K},\notag\\
    &&\hspace*{10mm}\overline{\mbox{C4}}\mbox{:}\hspace*{1mm}\hat{\lambda}_k\geq 0, \forall k\in\mathcal{K},\vspace*{-2mm}
\end{eqnarray}
where $\hat{\bm{\lambda}}=[\hat{\lambda}_1,\cdots,\hat{\lambda}_K]$ is an auxiliary optimization variable. Problem \eqref{Feasible_problem_robust} is convex, always feasible, and can be solved with CVX \cite{grant2008cvx}. Similar to the notation in \eqref{Primal_problem}, we define $\widehat{\widetilde{\bm{\Lambda}}}=\{\widehat{\widetilde{\mathbf{Q}}}_k,\widehat{\widetilde{\bm{\Xi}}}_k\}$, where $\widehat{\widetilde{\mathbf{Q}}}_k$ and $\widehat{\widetilde{\bm{\Xi}}}_k$ are the dual matrices for constraints $\widehat{\overline{\mbox{C3a}}}$ and $\widehat{\overline{\mbox{C3d}}}$ in \eqref{Feasible_problem_robust}, respectively, and the optimal solution of \eqref{Feasible_problem_robust} is denoted by $\widehat{\widetilde{\bm{\Gamma}}}^{(i)}$. The Lagrangian of \eqref{Feasible_problem_robust} is given by
\vspace*{-2mm}
\begin{equation}\label{Largrangian_feasible_robust}
    \widehat{\widetilde{\mathcal{L}}}(\widehat{\bm{\Gamma}},\bar{\mathbf{B}},\widehat{\widetilde{\bm{\Lambda}}})=\hat{f}_{1}(\widehat{\bm{\Gamma}},\widehat{\widetilde{\bm{\Lambda}}})+\hat{f}_{2}(\bar{\mathbf{B}},\widehat{\widetilde{\bm{\Lambda}}}).
\end{equation}
Based on the solutions of \eqref{Master_Problem_robust} and \eqref{Feasible_problem_robust}, the master problem in the $i$-th iteration is formulated as follows
\begin{eqnarray}
\label{Master_problem_robust_Corr}
    &&\hspace*{-8mm}\underset{\bar{\mathbf{B}},\hat{\eta}}{\mino}\hspace*{2mm}\hat{\eta}\notag\\
    &&\hspace*{-4mm}\mbox{s.t.}\hspace*{2mm}\mbox{C2a}, \mbox{C2b},\notag\\
    &&\hspace*{2mm}\widehat{\mbox{{C5a}}}\mbox{:}\hspace*{1mm}\hat{\eta}\geq\bar{\mathcal{L}}(\widehat{\bm{\Gamma}}^{(t)},\bar{\mathbf{B}},\widehat{\bm{\Lambda}}^{(t)}),\hspace*{0mm}\forall t\in\hat{\mathcal{F}}\cap\{1,\cdots, i\},\notag\\
    &&\hspace*{2mm}\widehat{\mbox{{C5b}}}\mbox{:}\hspace*{1mm}0\geq\widehat{\widetilde{\mathcal{L}}}(\widehat{\widetilde{\bm{\Gamma}}}^{(i)},\bar{\mathbf{B}},\widehat{\widetilde{\bm{\Lambda}}}^{(i)}),\hspace*{0mm}\forall t\in\hat{\mathcal{I}}\cap\{1,\cdots, i\},\vspace*{-2mm}
\end{eqnarray}
where $\hat{\mathcal{F}}$ and $\hat{\mathcal{I}}$ denote the sets collecting the iteration indices for which \eqref{Master_Problem_robust} is feasible and infeasible, respectively. The master problem in \eqref{Master_problem_robust_Corr} is an MILP problem, which can be optimally solved by exploiting standard numerical solvers for MILPs, e.g., MOSEK. By optimally solving \eqref{Master_problem_robust_Corr}, a lower bound $\hat{\eta}^{(i)}$ for the original problem in \eqref{Ori_Problem_robust} is obtained, and the obtained optimal solution $\bar{\mathbf{B}}^{(i)}$ is exploited to generate the primal problem in the next iteration.

 The resulting GBD-based algorithm can be summarized in a similar manner as \textbf{Algorithm 1}. Specifically, the primal problem in \eqref{Master_Problem_robust} is an SDP problem with $K$ LMI constraints in $\widehat{\widebar{\mbox{C1}}}$ of size $((N+1)M+1)\times ((N+1)M+1)$, $K$ LMI constraints in $\widehat{\widebar{\mbox{C3a}}}$ of size $(2N+3)M\times (2N+3)M$, $K$ LMI constraints in $\widehat{\widebar{\mbox{C3c}}}$ of size $(N+3)M\times (N+3)M$, and $K$ LMI constraints in C8 of size $M\times M$. Therefore, the computational complexity in each iteration for solving the primal problem is given by $\mathcal{O}(\log\frac{1}{\rho_{\mathrm{SDP}}}(K(((N+1)M+1)^3+((2N+3)^3+(N+3)^3+1)M^3)+K^2((N+1)M+1)^2+((2N+3)^2+(N+3)^2+1)M^2)+K^3))$, where $\rho_{\mathrm{SDP}}$ is the convergence tolerance of the interior point method used to solve the SDP problem\cite[Theorem 3.12]{bomze2010interior}. Moreover, the initial point, the convergence, and the optimality analysis are identical to those for \textbf{Algorithm 1} in Section \rom{3}-B, and are not described here in detail due to page limitation. 
 \vspace*{-2mm}\begin{remark}
     We note that the computational complexity for imperfect CSI is significantly higher compared to that for perfect CSI since more LMI constraints and auxiliary variables have to be introduced to suitably reformulate the robust QoS constraint. Thus, it is not computationally efficient to employ the GBD-based algorithm designed for imperfect CSI to solve the resource allocation problem for perfect CSI by setting the CSI estimation error to zero, i.e., $\epsilon_k=0$.
 \end{remark}

\vspace*{-4mm}\subsection{Proposed Penalty-based SCA Algorithm}
Following a similar procedure as in Section \rom{3}-C, we can recast the robust resource allocation problem in \eqref{Ref_Problem_robust} as 
\begin{eqnarray}
\label{Penalty_Problem_robust}
    &&\hspace*{-6mm}\underset{\widehat{\bm{\Gamma}},\bar{\mathbf{B}}}{\mino}\hspace*{2mm}\sum_{k\in\mathcal{K}}\mathrm{Tr}(\mathbf{W}_k)+\frac{1}{\mu}\sum_{n=1}^N\sum_{l=1}^L (b_n[l]-b_n^2[l])\notag\\
    &&\hspace*{0mm}\mbox{s.t.}\hspace*{6mm} \widehat{\overline{\mbox{C1}}},\mbox{C2a},\overline{\mbox{C2c}},\widehat{\overline{\mbox{C3a}}},\widehat{\overline{\mbox{C3b}}},\widehat{\overline{\mbox{C3c}}},\widehat{\overline{\mbox{C3d}}},\mbox{C7},\mbox{C8},
\end{eqnarray}
By using SDP relaxation as for the optimal resource allocation algorithm, we can relax rank-one constraint C7 in problem \eqref{Penalty_Problem_robust}. Then, we apply SCA to obtain the following optimization problem in the $(i+1)$-th iteration
\begin{eqnarray}
\label{SCA_Problem_robust}
    &&\hspace*{-4mm}\underset{\widehat{\bm{\Gamma}},\bar{\mathbf{B}}}{\mino}\hspace*{2mm}\sum_{k\in\mathcal{K}}\mathrm{Tr}(\mathbf{W}_k)+\frac{1}{\mu}\sum_{n=1}^N\sum_{l=1}^L\notag\\
    &&\hspace*{13mm} (b_n[l]-2b_n^{(i)}[l]b_n[l]-(b_n^{(i)}[l])^2)\notag\\
    &&\hspace*{2mm}\mbox{s.t.}\hspace*{7mm} \widehat{\overline{\mbox{C1}}},\mbox{C2a},\overline{\mbox{C2c}},\widehat{\overline{\mbox{C3a}}},\widehat{\overline{\mbox{C3b}}},\widehat{\overline{\mbox{C3c}}},\widehat{\overline{\mbox{C3d}}},\mbox{C8},
\end{eqnarray}
which is convex w.r.t. $\widehat{\bm{\Gamma}}$ and $\bar{\mathbf{B}}$ and thus can be optimally solved by standard convex program solvers such as CVX. The resulting SCA algorithm can be implemented similarly to Algorithm 2 with slight changes in the notations. Specifically, problem \eqref{SCA_problem} is substituted by problem \eqref{SCA_Problem_robust}. The initial point as well as the optimality and convergence analysis are identical to those for \textbf{Algorithm 2} in Section \rom{3}-C, and thus are not provided here due to page limitation. The algorithm is guaranteed to converge to a stationary point of problem \eqref{SCA_Problem_robust} in polynomial time \cite{razaviyayn2013unified}. In particular, the relaxed problem in \eqref{SCA_Problem_robust} has a similar structure as the primal problem in \eqref{Master_Problem_robust}. Therefore, the computational complexity of each iteration of the proposed SCA-based algorithm is given by $\mathcal{O}(\log\frac{1}{\rho_{\mathrm{SDP}}}(K(((N+1)M+1)^3+((2N+3)^3+(N+3)^3+1)M^3)+K^2((N+1)M+1)^2+((2N+3)^2+(N+3)^2+1)M^2)+K^3))$, where $\rho_{\mathrm{SDP}}$ is the convergence tolerance of the interior point method used for solving the SDP problem\cite[Theorem 3.12]{bomze2010interior}. 
\vspace*{-2mm}\begin{remark}
    The proposed SCA-based algorithm for imperfect CSI is not efficient in solving the resource allocation problem for perfect CSI due to the significantly higher computational complexity compared to \textbf{Algorithm 2}.
\end{remark}
\vspace*{-5mm}
\section{Simulation Results}
\vspace*{-1mm}\subsection{Simulation Setting}
\begin{figure}\vspace*{-3mm}
    \centering
    \includegraphics[width=3.2in]{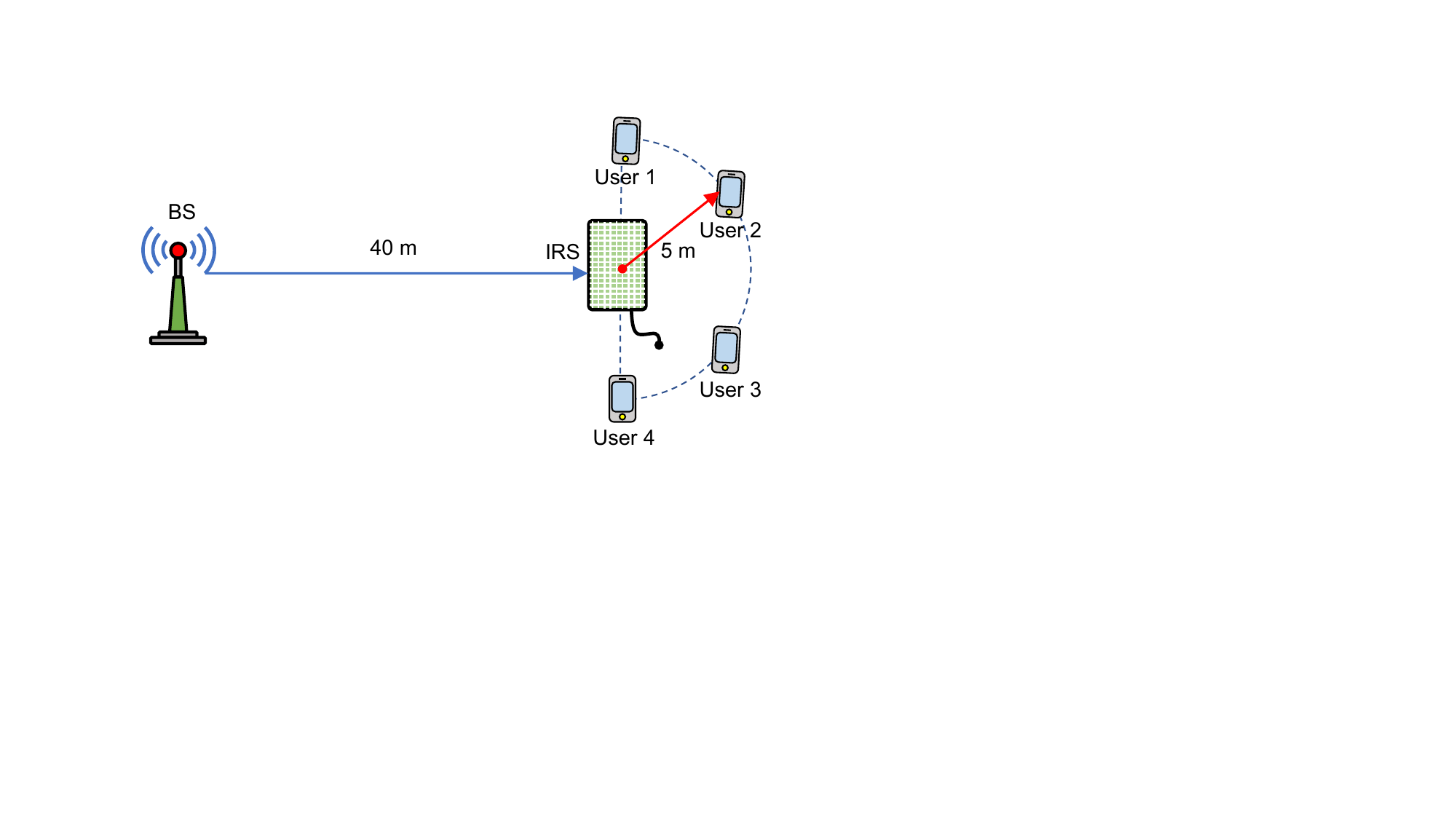}
    \caption{IRS-assisted wireless network topology adopted for simulations.}
    \label{fig::Sim_setting}
\end{figure}
As shown in Fig. \ref{fig::Sim_setting}, we consider an IRS-assisted multiuser MISO system, where the BS is equipped with $M=6$ transmit antennas and serves $K=4$ single-antenna users. 
{\color{black} The IRS is deployed $D=40$ m from the BS, while the users are located on a circle having its origin at the IRS and a radius of $r=5$ m as shown in Fig. \ref{fig::Sim_setting}.}
The noise variances of all users are set to $\sigma_k^2=-90$ dBm, $\forall k \in\mathcal{K}$. The channel matrix $\mathbf{F}$ between the BS and the IRS is modeled as Rician fading
\begin{equation}
    \mathbf{F}=\sqrt{L_0D^{-\alpha_{\mathrm{BI}}}}\left(\sqrt{\frac{\beta_{\mathrm{BI}}}{1+\beta_{\mathrm{BI}}}}\mathbf{F}_L+\sqrt{\frac{1}{1+\beta_{\mathrm{BI}}}}\mathbf{F}_N\right),
\end{equation}
where $L_0$ is the path loss at the reference distance $d_0=1$ m. $\alpha_{\mathrm{BI}}=2.2$ and $\beta_{\mathrm{BI}}=1$ denote the path loss exponent and the Rician factor, respectively \cite{yu2021robust}. Matrices $\mathbf{F}_L$ and $\mathbf{F}_N$ denote the line-of-sight (LoS) and non-LoS (NLoS) components, respectively. Here, the LoS component is the product of the receive and transmit array response vectors, while the entries of NLoS matrix $\mathbf{F}_N$ are modeled as Rayleigh-distributed random variables. We generate the channel vector between the IRS and user $k$, i.e., $\mathbf{h}_k$, in a similar manner as $\mathbf{F}$. The corresponding path loss exponent and the Rician fading factor of $\mathbf{h}_k$, $\forall k \in\mathcal{K}$, i.e., $\alpha_{\mathrm{IU}}^k$ and $\beta_{\mathrm{IU}}^k$, are set to $2.8$ and $1$, respectively \cite{yu2021robust}. 
{\color{black} Furthermore, we assume that the users are not located in the direct path between BS and IRS and that the LoS path between the BS and users is fully blocked. Thus, the direct channel vector $\mathbf{d}_k$ between the BS and user $k$ has only an NLoS component.}
Hence, $\mathbf{d}_k$ is modeled as Rayleigh distributed, where the path loss exponent for the direct channels is given by $\alpha_d=4$. On the other hand, the channel estimates $\bar{\mathbf{E}}_k$ and $\bar{\mathbf{d}}_k$ are generated as $\bar{\mathbf{E}}_k=\mathrm{diag}(\mathbf{h}_k^H)\mathbf{F}-\Delta\mathbf{E}_k$ and $\bar{\mathbf{d}}_k=\mathbf{d}_k-\Delta\mathbf{d}_k$, respectively, where the channel estimation errors $\Delta\mathbf{E}_k$ and $\Delta\mathbf{d}_k$ are randomly generated and meet the norm-bound constraints. For ease of exposition, we assume all users impose the same SINR requirement, i.e., $\gamma_k=\gamma, \forall k$, and define the maximum normalized estimation error of the channels as $\kappa_k=\epsilon_k/\|\bar{\mathbf{H}}_k\|_F=\kappa$, $\forall k$. The number of channel realizations considered for the following simulations is $200$.

For comparison, we consider four baseline schemes. For \textbf{Baseline Scheme 1}, we evaluate the performance of a conventional multiuser MISO system without the deployment of an IRS and solve problem \eqref{Primal_problem} and \eqref{Master_Problem_robust} optimally with $\mathbf{B}=\mathbf{0}$ for the cases of perfect and imperfect CSI, respectively. For \textbf{Baseline Scheme 2}, the binary IRS selection matrix $\mathbf{B}_{\mathrm{rdn}}$ is randomly generated. The beamforming matrix is designed by solving the problems in \eqref{Primal_problem} and \eqref{Master_Problem_robust} with $\mathbf{B}_{\mathrm{rdn}}$ for perfect and imperfect CSI, respectively. Furthermore, we adopt the AO-based algorithms in \cite{wu2019intelligent} and \cite{yu2021robust} as \textbf{Baseline Scheme 3} for the cases of perfect and imperfect CSI, respectively. Here, the BS beamforming matrix and continuous IRS phase shifts are alternatingly optimized in an iterative manner by adopting semidefinite relaxation until convergence. After convergence, the obtained continuous phase shift matrix is quantized to the feasible discrete phase shift matrix $\bm{\Phi}_{\mathrm{AO}}$, and the beamforming matrix is designed by solving problems in \eqref{Primal_problem} and \eqref{Master_Problem_robust} with $\bm{\Phi}_{\mathrm{AO}}$. Moreover, for perfect CSI, we adopt the algorithm based on the inner approximation (IA) method reported in \cite{yu2021robust} as \textbf{Baseline Scheme 4}. {\color{black} In this case, the coupling between the BS beamforming matrix and the continuous IRS phase shifts is handled by the IA technique, while the non-convexity introduced by the continuous IRS phase shifts is reformulated using SCA. After convergence, the feasible discrete phase shift matrix $\bm{\Phi}_{\mathrm{IA}}$ and beamforming matrix are generated based on the obtained continuous IRS phase shift matrix using a similar procedure as for {Baseline Scheme 3}. Since the IA algorithm achieves a close-to-optimal performance for perfect CSI with continuous IRS phase shifts, Baseline Scheme 4 can reveal the performance loss caused by quantization of continuous phase shifts to discrete ones.}
\vspace*{-4mm}\subsection{Numerical Results for Perfect CSI}
\subsubsection{Convergence of the proposed algorithms}
Fig. \ref{fig::conver_perfect} illustrates the convergence of the proposed GBD-based and SCA-based algorithms for different numbers of IRS elements $N$ and perfect CSI at the BS. 
{\color{black}In the upper half of Fig. \ref{fig::conver_perfect}, the objective function values of the upper bound (UB) and lower bound (LB) in the $i$-th iteration are obtained from \eqref{UB} and the optimization problem in \eqref{Master_problem}, respectively. Furthermore, we adopt an exhaustive search (ES) approach as a benchmark to verify the global optimality of the proposed GBD-based method. Specifically, the ES method solves the problem in \eqref{Primal_problem} considering all $L^N$ possible IRS phase shift configurations to attain the global optimum. As can be seen from the upper half of Fig. \ref{fig::conver_perfect}, the proposed GBD-based algorithm converges to the optimal solution obtained with the ES method, which confirms the global optimality of the proposed GBD-based method. Moreover, the upper bound value obtained by solving the primal problem coincides with the lower bound value obtained from \eqref{Master_problem} after less than $150$ iterations for $N=16$, which is much faster than the brute force ES method that requires evaluating all $L^N=2^{16}$ possible IRS phase shift configurations \cite{sahinidis1991convergence,ng2014robust}.}
On the other hand, as the lower half of Fig. \ref{fig::conver_perfect} shows, the proposed suboptimal algorithm converges to a locally optimal value in fewer than $10$ iterations, approaching the optimal bound obtained via the GBD-based method, revealing a trade-off between complexity and performance. 
\begin{figure}[t]\vspace*{-5mm}
	\centering
	\includegraphics[width=3.2in]{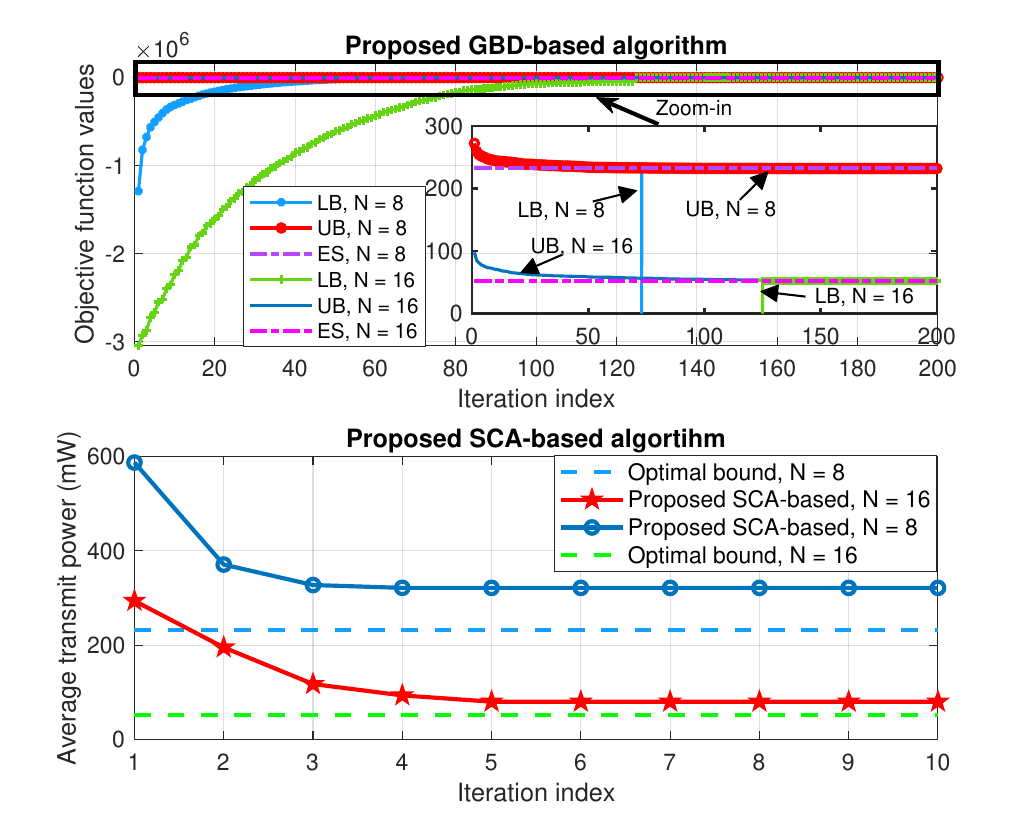}\vspace*{-2mm}
	\caption{Convergence behavior of proposed algorithms for different values of $N$ and perfect CSI. The system parameters are set as $\gamma=5$ dB, $B=1$, and $L=2$.}
	\label{fig::conver_perfect}
\end{figure}
\begin{table}[!htbp]
\color{black}
\caption{Runtime comparison for perfect CSI.}
\centering
\vspace{3mm}
\begin{tabular}{|l|llll|}
\hline
\multirow{2}{*}{Employed method} & \multicolumn{4}{l|}{Runtime (seconds)}                              \\ \cline{2-5} 
                                 & \multicolumn{1}{l|}{$N=4$} & \multicolumn{1}{l|}{$N=8$} & \multicolumn{1}{l|}{$N=16$} & \multicolumn{1}{l|}{$N=32$} \\ \hline
GBD-based method                 & \multicolumn{1}{l|}{1.36}    & \multicolumn{1}{l|}{2.68} & \multicolumn{1}{l|}{6.64} & \multicolumn{1}{l|}{24.75} \\ \hline
SCA-based method                 & \multicolumn{1}{l|}{0.11}    & \multicolumn{1}{l|}{0.25} & \multicolumn{1}{l|}{0.56} & \multicolumn{1}{l|}{1.82} \\ \hline
Baseline Scheme 3                 & \multicolumn{1}{l|}{0.20}    & \multicolumn{1}{l|}{0.42} & \multicolumn{1}{l|}{0.73} & \multicolumn{1}{l|}{1.20} \\ \hline
Baseline Scheme 4                 & \multicolumn{1}{l|}{0.32}    & \multicolumn{1}{l|}{0.44} & \multicolumn{1}{l|}{0.94} & \multicolumn{1}{l|}{2.27} \\ \hline
\end{tabular}
\color{black}
\end{table}

{\color{black} Next, the runtimes of the proposed algorithms, and Baseline Schemes 3 and 4 are investigated for different numbers of IRS elements $N$ for the case of perfect CSI. In particular, the convergence tolerances of the proposed GBD-based and SCA-based algorithms are set to $\Delta=10^{-3}$ and $\Delta_{\mathrm{SCA}}=10^{-3}$, respectively. Moreover, the convergence criterion of Baseline Schemes 3 and 4 are set as $\frac{\|\bm{\Phi}_{\mathrm{AO}}^{(i)}-\bm{\Phi}_{\mathrm{AO}}^{(i-1)}\|}{\|\bm{\Phi}_{\mathrm{AO}}^{(i-1)}\|_F}\leq 10^{-3}$ and $\frac{\|\bm{\Phi}_{\mathrm{IA}}^{(i)}-\bm{\Phi}_{\mathrm{IA}}^{(i-1)}\|}{\|\bm{\Phi}_{\mathrm{IA}}^{(i-1)}\|_F}\leq 10^{-3}$, respectively. The algorithms are implemented in MATLAB R2018a and tested on a PC with a Core i9-13900K CPU. MOSEK was adopted as CVX solver for time efficiency. As can be observed, the proposed SCA-based algorithm requires significantly less runtime compared with the GBD-based algorithm, confirming the computational time efficiency of the proposed SCA-based algorithm. Moreover, the proposed SCA-based method also has a lower runtime than the Baseline Schemes 3 and 4 for small-scale IRSs, e.g., $N\leq 16$, since it needs only a few iterations to converge. On the other hand, for moderate-to-large scale IRSs, due to the significant dimension increase caused by the introduced auxiliary variables, the proposed SCA-based algorithm requires a longer runtime to converge compared to the AO-based Baseline Scheme 3.\footnote{ \color{black}Note that the actual runtime of the proposed algorithm can be significantly reduced by employing customized convex program solvers, and processing hardware\cite{grant2008cvx,mosek}. Yet, these are beyond the scope of this paper.} }
\subsubsection{Average transmit power versus minimum required SINR}
Fig. \ref{fig::SINR_perfect} depicts the average BS transmit power for different minimum required SINRs, when perfect CSI is available at the BS. It is observed that as the minimum required SINR increases, the BS has to consume more power to satisfy the more stringent QoS requirements of the users. Furthermore, the proposed schemes outperform Baseline Schemes 1 and 2 for all considered values of $\gamma$. In particular, since no IRS is deployed for Baseline Scheme 1, while for Baseline Scheme 2, the phase shift pattern is randomly generated, Baseline Schemes 1 and 2 cannot effectively shape the channels for the users. As for Baseline Scheme 3, the adopted AO algorithm optimizes the beamforming vectors and the phase shift matrix, leading to a $5$ dB gain compared to Baseline Scheme 2. However, the AO-based algorithm, which searches only a subset of the feasible region close to the initial point, is prone to getting trapped in stationary points, resulting in a $4.5$ dB performance gap compared to the proposed GBD-based optimal solution in the considered setting. On the other hand, since the IA-based algorithm suffers from the quantization loss after convergence, the GBD-based scheme achieves a considerable gain ($3$ dB) compared to the IA-based scheme. Furthermore, the proposed SCA-based algorithm outperforms the IA-based and AO-based schemes and achieves a performance close to the optimal GBD-based scheme. Moreover, as can be seen from Fig.  \ref{fig::SINR_perfect}, employing slightly more expensive 2-bit phase shifters, i.e., with $L=2^2=4$ discrete phase levels, leads to a $2$ dB gain compared to the case with 1-bit phase shifters, indicating a trade-off between the required transmit power and the phase shifter cost in IRS-assisted communication systems.
\begin{figure}[t]\vspace*{-5mm}
	\centering
	\includegraphics[width=2.8in]{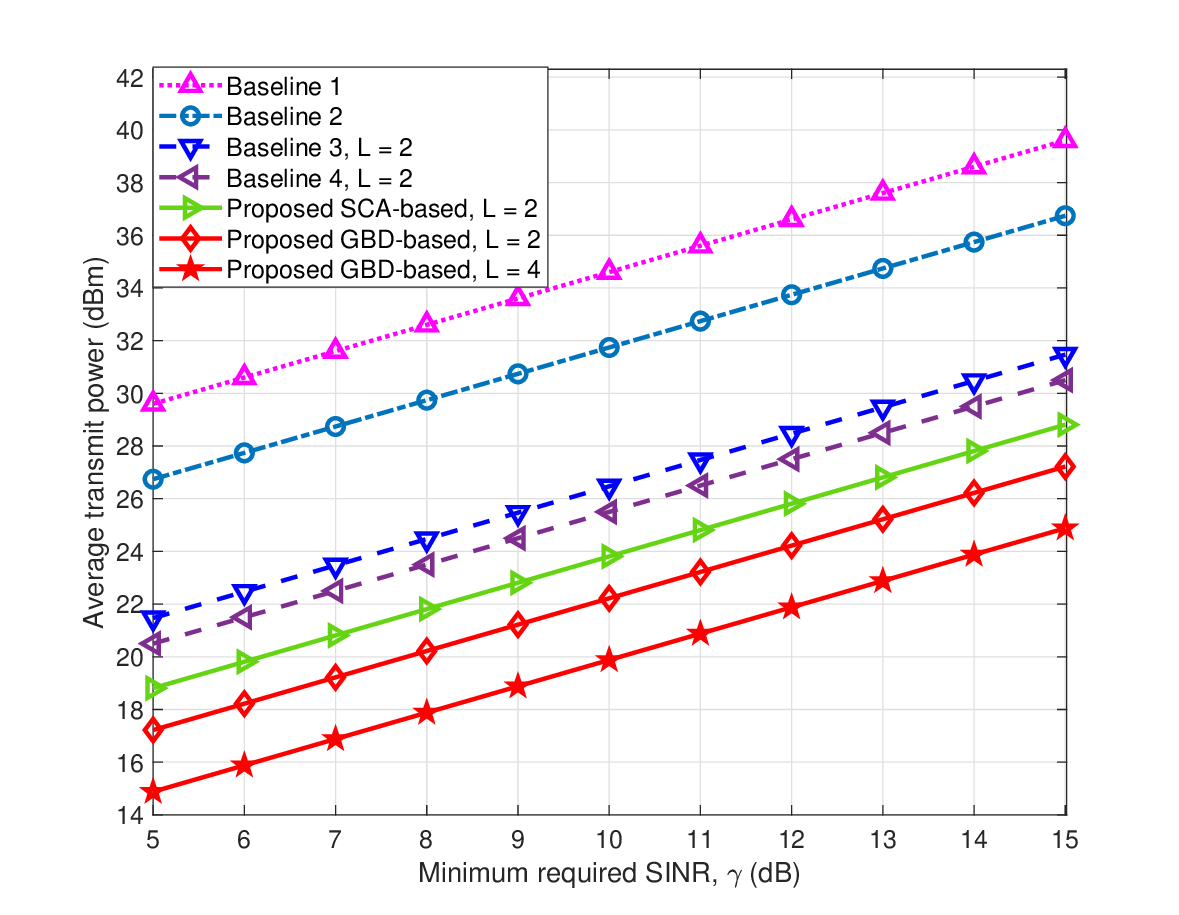}
	\caption{Average BS transmit power versus the minimum required SINR with perfect CSI. The number of IRS elements is set as $N=16$.}
	\label{fig::SINR_perfect}
\end{figure}

\subsubsection{Average transmit power versus number of IRS elements}

\begin{figure}[t]\vspace*{-5mm}
	\centering
	\includegraphics[width=3.2in]{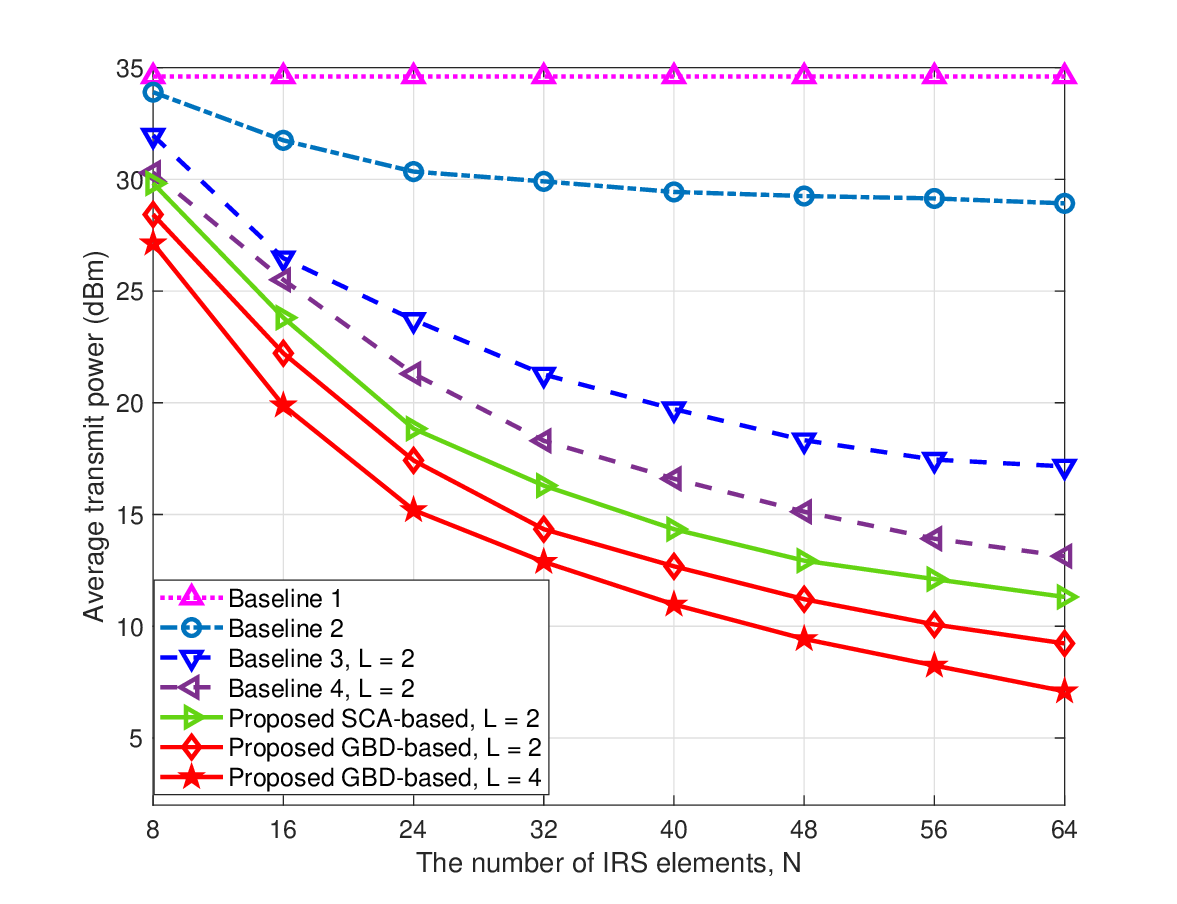}
	\caption{Average BS transmit power versus number of IRS elements with perfect CSI. The minimum required SINR is set to $\gamma=10$ dB.}
	\label{fig::N_perfect}
\end{figure}
Fig. \ref{fig::N_perfect} shows the BS transmit power versus the number of IRS elements for $\gamma=10$ dB. As can be observed, the transmit powers of the proposed schemes and Baseline Schemes 2, 3, and 4 decrease as the number of IRS elements increases. This is attributed to the fact that a larger number of IRS elements can reflect more power towards the users, leading to a power gain. Moreover, the additional phase shifters provide extra degrees of freedom (DoFs) for resource allocation design, which can be exploited to tailor a more favorable wireless channel for information transmission. Furthermore, we observe that the gap between the proposed schemes and Baseline Scheme 2 is enlarged as the number of IRS elements increases. In particular, random phase shifts cannot fully exploit the additional DoFs introduced by deploying more IRS elements, leading to limited passive IRS beamforming gains. On the other hand, the performance gap between the Baseline Scheme 3 and the proposed GBD-based scheme is enlarged from $N=8$ to $N=40$. In fact, the rapid expansion of the feasible solution set as $N$ increases makes Baseline Scheme 3 more susceptible to getting trapped in an unsatisfactory stationary point. For instance, for $N=64$, an additional $8$ dB transmit power is required for Baseline Scheme 3 to approach the performance of the proposed GBD-based scheme. Moreover, a $5$ dB gap between the proposed GBD-based and Baseline Scheme 4 is observed for $N=64$ due to the phase shift quantization loss. 
{\color{black}On the other hand, the average transmit power required for the proposed SCA-based algorithm is $5$ dB lower than that required with Baseline Scheme 3, which clearly shows the effectiveness of the proposed SCA-based algorithm, although only 10 iterations are employed.}
Moreover, increasing the resolution of the discrete phase shifts, $L$, leads to a $2.5$ dB gain at the expense of higher hardware costs and additional control bits $B$. Given a fixed total number of $64$ IRS control bits, i.e., $N\times B=64$, $N=64$ IRS elements with 1-bit phase shifters outperform $N=32$ IRS elements with 2-bit phase shifters by $3$ dB for the considered system. This observation highlights the appeal of deploying a large number of low-resolution IRS elements.
\vspace*{-5mm}\subsection{Numerical Results for Imperfect CSI}

\begin{figure}[t]\vspace*{-5mm}
	\centering
	\includegraphics[width=3.2in]{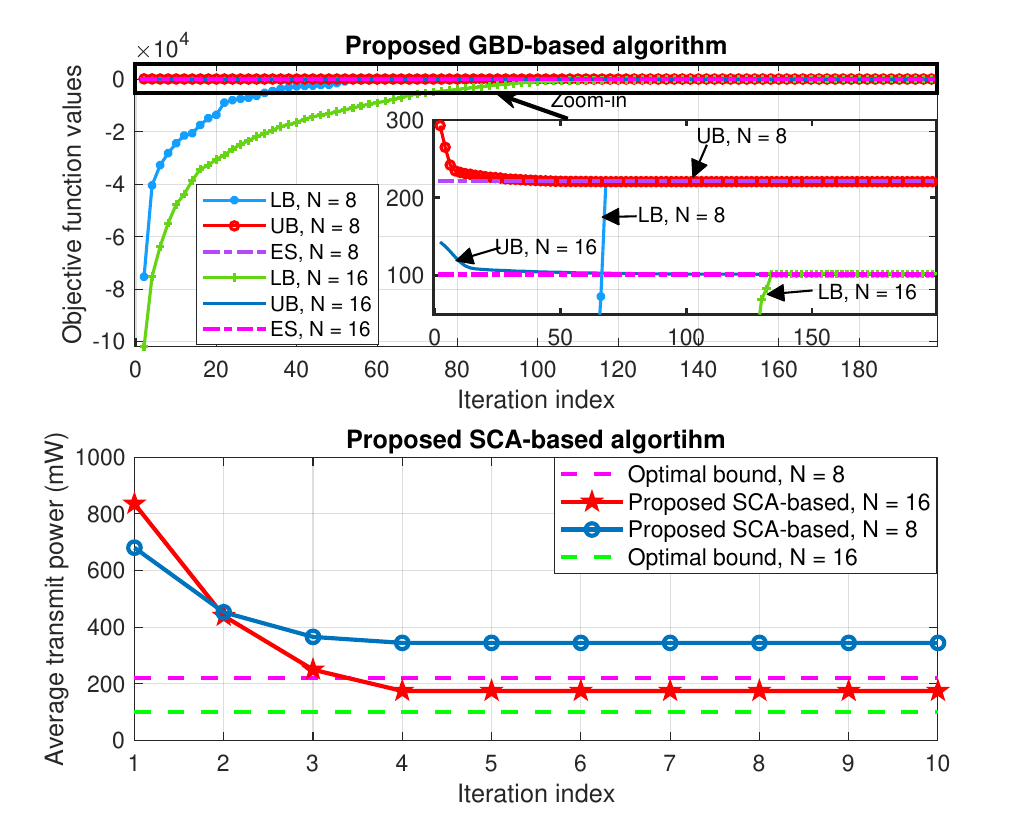}
	\vspace*{-3mm}\caption{Convergence behavior of the proposed algorithms for different values of $N$ and imperfect CSI. The system parameters are set as $\gamma=5$ dB, $B=1$, $L=2$, and $\kappa=0.1$.}
	\label{fig::conver_imperfect}
\end{figure}
\subsubsection{Convergence of the proposed algorithms}
{\color{black}We investigate the convergence of the proposed GBD-based and SCA-based algorithms for imperfect CSI in Fig. \ref{fig::conver_imperfect}. Similar to Fig. \ref{fig::conver_perfect}, we use the ES method that solves the optimization problem in \eqref{Master_Problem_robust} considering all $L^N$ possible IRS configurations as a benchmark. As can be observed, the proposed GBD-based algorithm converges to the global optimal solution provided by ES but requires a higher transmit power compared to the case of perfect CSI, see Fig. \ref{fig::conver_perfect}.}
In fact, the higher transmit power is required to combat the negative effects of CSI uncertainty. Moreover, the GBD method converges within 60 and 130 iterations on average for $N=8$ and $N=16$, respectively, which is much faster than the brute force ES method.
{\color{black} In addition, comparing with Fig. \ref{fig::conver_perfect}, we observe that for the proposed GBD-based methods, the performance difference between perfect CSI and imperfect CSI is larger for $N=16$ IRS elements than for $N=8$ IRS elements. In particular, increasing the number of IRS elements increases the dimension of the cascade channel matrix ${{\mathbf{E}}}_k$, leading to a larger norm-bounded region for the CSI estimation error. On the other hand, increasing the CSI error region also enlarges the feasible region of the CSI uncertainties. Therefore, more transmit power is required to ensure the robustness of the system with $N=16$. }
Furthermore, the proposed SCA-based algorithm achieves a locally optimal performance in less than $10$ iterations also for imperfect CSI, revealing the computational efficiency of the proposed algorithm for robust resource allocation design.
\begin{table}[!htbp]\color{black}
\caption{Runtime comparison for imperfect CSI.}
\centering
\vspace{3mm}
\begin{tabular}{|l|llll|}
\hline
\multirow{2}{*}{Employed method} & \multicolumn{4}{l|}{Runtime (seconds)}                              \\ \cline{2-5} 
                                 & \multicolumn{1}{l|}{$N=4$} & \multicolumn{1}{l|}{$N=8$} & \multicolumn{1}{l|}{$N=12$} & \multicolumn{1}{l|}{$N=16$} \\ \hline
GBD-based method                 & \multicolumn{1}{l|}{4.12}    & \multicolumn{1}{l|}{9.77} & \multicolumn{1}{l|}{21.35} & \multicolumn{1}{l|}{42.20} \\ \hline
SCA-based method                 & \multicolumn{1}{l|}{0.70}    & \multicolumn{1}{l|}{1.13} & \multicolumn{1}{l|}{2.89} & \multicolumn{1}{l|}{6.40} \\ \hline
Baseline Scheme 3                 & \multicolumn{1}{l|}{0.84}    & \multicolumn{1}{l|}{1.55} & \multicolumn{1}{l|}{2.51} & \multicolumn{1}{l|}{5.08} \\ \hline
\end{tabular}\color{black}
\end{table}

{\color{black}Table \rom{2} illustrates the runtime of the proposed algorithms and Baseline Scheme 3 w.r.t. the number of IRS elements $N$. In particular, we adopt the same convergence criteria for the proposed algorithms and Baseline Scheme 3 as for the case of perfect CSI in Table \rom{1}. Compared with the results in Table \rom{1}, both the SCA-based and GBD-based algorithms for imperfect CSI require longer runtimes to converge. In fact, the proposed robust designs require additional LMI constraints and auxiliary variables for the reformulation of the robust QoS constraint, leading to a relatively higher computational complexity compared to the algorithm for perfect CSI. Nevertheless, the proposed SCA-based algorithm requires a shorter runtime than the AO-based Baseline Scheme 3 for IRSs of small size, e.g., $N\leq 8$, since the SCA-based method requires fewer iterations to converge. However, as the IRS size grows, so do the dimensions of the auxiliary variables required for the SCA-based method. Thus, for larger IRSs, e.g., $N=16$, the SCA-based method demands a longer runtime compared to Baseline Scheme 3.}
\subsubsection{Average transmit power versus minimum required SINR}
Fig. \ref{fig::SINR_imperfect} illustrates the BS transmit power versus the minimum required SINR for imperfect CSI. Due to the negative impact of the CSI uncertainty, the baseline schemes cannot find a feasible solution for a large minimum required SINR $\gamma$, i.e., $\gamma\geq 5$ dB. Therefore, we investigate the BS transmit power for $\gamma$ from $0$ to $5$ dB. The proposed algorithms reduce the BS transmit power significantly compared to the three baseline schemes. Specifically, we observe that Baseline Scheme 1 without IRS outperforms Baseline Scheme 2 with randomly tuned IRS elements since the reflected links in Baseline Scheme 2 introduce additional CSI uncertainty. 
This underscores the importance of jointly optimizing the beamforming vectors and the IRS phase shift matrix in case of imperfect CSI.
{\color{black}On the other hand, the proposed SCA-based scheme achieves a considerable gain ($2$ dB) compared to the AO-based Baseline Scheme 3 in the considered setting, which confirms the superiority of the SCA-based algorithm also for imperfect CSI. Note that the results obtained by the AO-based method have to be quantized to feasible discrete IRS matrices after convergence. The coarse quantization from continuous phase shift values to discrete ones leads to a mismatch for beamforming and interference cancellation \cite{wu2019beamforming}. Therefore, a significant performance degradation is expected after quantization.}
Moreover, the performance gap between the SCA-based algorithm and the GBD-based algorithm is enlarged in Fig. \ref{fig::SINR_imperfect} compared to the case of perfect CSI considered in Fig. \ref{fig::SINR_perfect}. In particular, the SCA-based solution exhibits a higher sensitivity to CSI errors than the optimal GBD solution, as the latter does not get trapped in ineffective local optima, even if the size of the feasible solution set shrinks due to a low CSI quality.
\vspace*{-1mm}\subsubsection{Average transmit power versus number of IRS elements}
\begin{figure}[t]\vspace*{-3mm}
	\centering
	\includegraphics[width=3.2in]{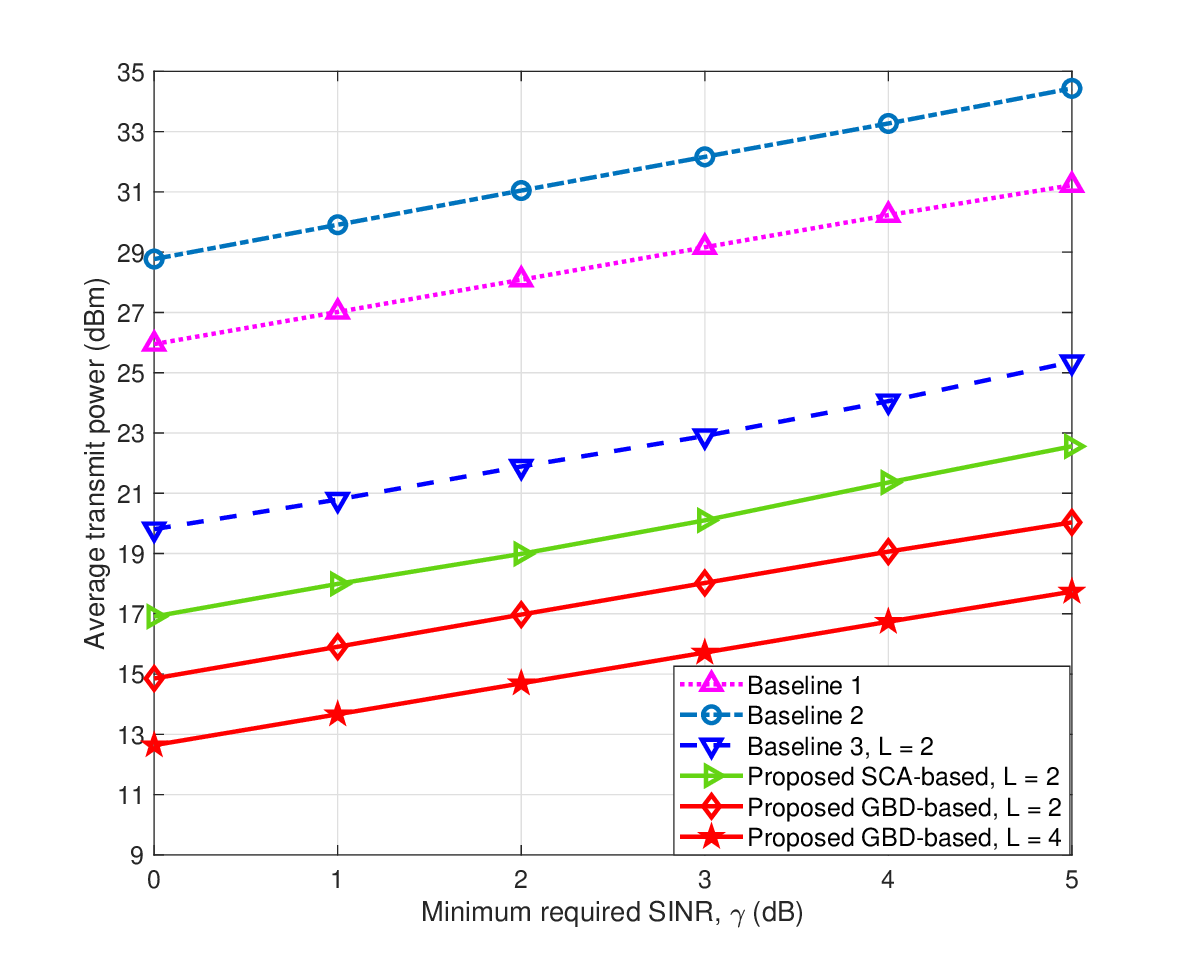}
	\caption{Average BS transmit power versus minimum required SINR $\gamma$ for imperfect CSI. The system parameters are set as $N=16$, and $\kappa=0.1$.}
	\label{fig::SINR_imperfect}
\end{figure}
We investigate the BS transmit power as a function of the number of IRS elements in Figs. \ref{fig::N_imperfect1} and \ref{fig::N_imperfect2}. Due to the high complexity of the GBD-based algorithm, we investigate its performance only in Fig. \ref{fig::N_imperfect1} for a relatively small system size. 
{\color{black} Here, we also show the optimal performance of the IRS-assisted systems with perfect CSI as a performance upper bound to investigate the impact of CSI estimation errors. As can be observed, for Baseline Scheme 2, the BS consumes more power to combat the additional CSI estimation errors introduced by the IRS-assisted channel. On the other hand, similar to Fig. 7, the transmit power for the proposed schemes and Baseline Scheme 3 decreases as $N$ grows, which is attributed to the additional DoFs introduced by the IRS elements. Furthermore, a significant performance gap, i.e., $6$ dB, between the GBD-based method and Baseline Scheme 3 is observed for $N=20$, which highlights the significance of optimally exploiting all DoFs provided by the IRS for the considered imperfect CSI scenario. Moreover, the gap between the optimal performance for perfect CSI and imperfect CSI is enlarged as the number of IRS elements increases. In particular, increasing the IRS size also increases the size of the estimated cascaded channel matrix, leading to an enlarged feasible region for the CSI uncertainty. Thus, the BS has to invest more transmit power to compensate for the additional CSI estimation error as $N$ increases.}
On the other hand, according to Figs. \ref{fig::N_imperfect1} and \ref{fig::N_imperfect2}, the performance gap between the proposed SCA-based algorithm and Baseline Scheme 3 is also enlarged as $N$ increases. In particular, the additional CSI estimation error introduced by a larger number of IRS elements reduces the size of the feasible solution set. Thus, the AO-based Baseline Scheme 3 tends to search over a region close to the given initial point for large $N$, resulting in an unsatisfactory performance.
\begin{figure*}[htbp]\vspace*{-5mm}
	\centering
 \subfloat[$N\in(4,20).$]{
         \centering
         \includegraphics[width=3.2in]{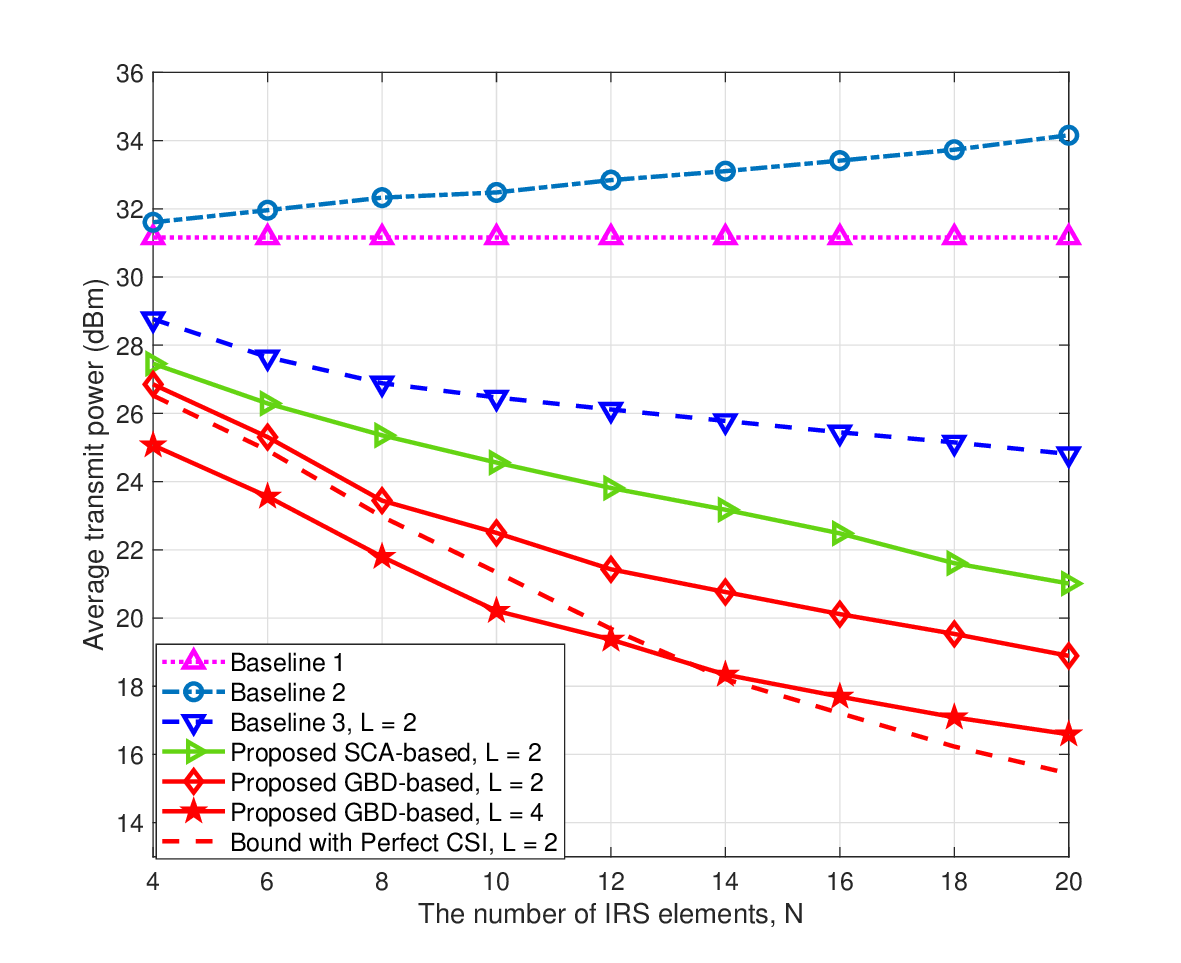}     \label{fig::N_imperfect1}
}
 \subfloat[$N\in(8,64).$]{
         \centering
         \includegraphics[width=3.3in]{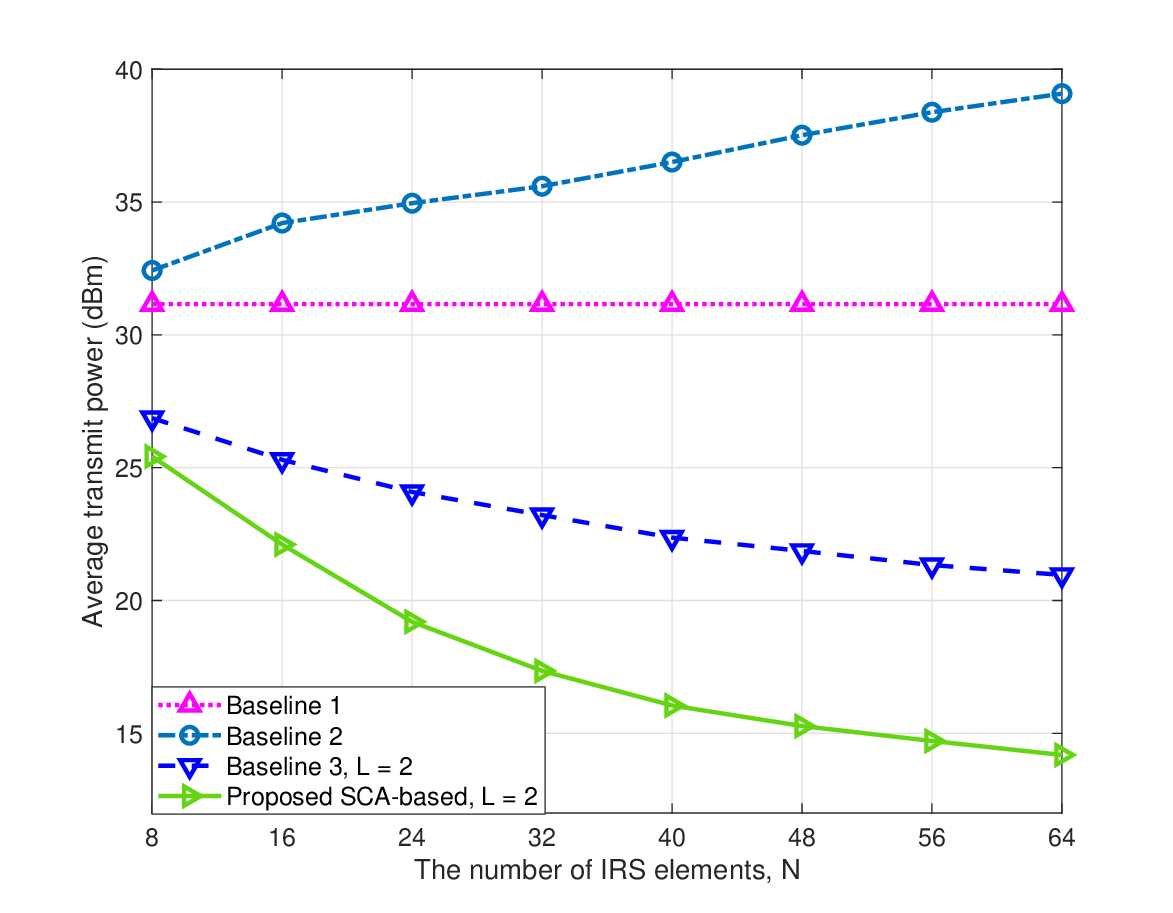}
         \label{fig::N_imperfect2}
}
         
	\vspace*{3mm}\caption{Average BS transmit power versus number of IRS elements with imperfect CSI. The system parameters as set as $\gamma=5$ dB, and $\kappa=0.1$.}
    \vspace{-1.5em}

\end{figure*}

\vspace*{-1mm}\subsubsection{CSI uncertainty}
Fig. \ref{fig::kappa} shows the average transmit power versus the maximum normalized channel estimation error $\kappa$. It is observed that the average transmit power increases as the CSI estimation quality decreases. In particular, as the quality of the estimated CSI deteriorates, it becomes more challenging for the BS to perform accurate beamforming and to satisfy the QoS requirements, such that a higher transmit power is required to ensure the robustness of the communication system. In addition, both the proposed SCA-based and GBD-based algorithms outperform the baseline schemes, which reveals the capability of the proposed schemes to effectively exploit the spatial DoFs even in the presence of severe CSI uncertainty. Furthermore, for large CSI estimation errors, i.e., $\kappa=0.2$, Baseline Scheme 2 requires more transmit power compared to Baseline Scheme 1. Specifically, Baseline Scheme 2 suffers from the additional CSI uncertainty of the reflected link, which is not present for Baseline Scheme 1. 
Moreover, the performance gain achieved by the proposed schemes over the AO-based scheme increases with $\kappa$. This confirms that the proposed schemes are more efficient in minimizing the transmit power in the presence of high CSI uncertainty.
\begin{figure}[t]\vspace*{-3mm}
	\centering
	\includegraphics[width=3.2in]{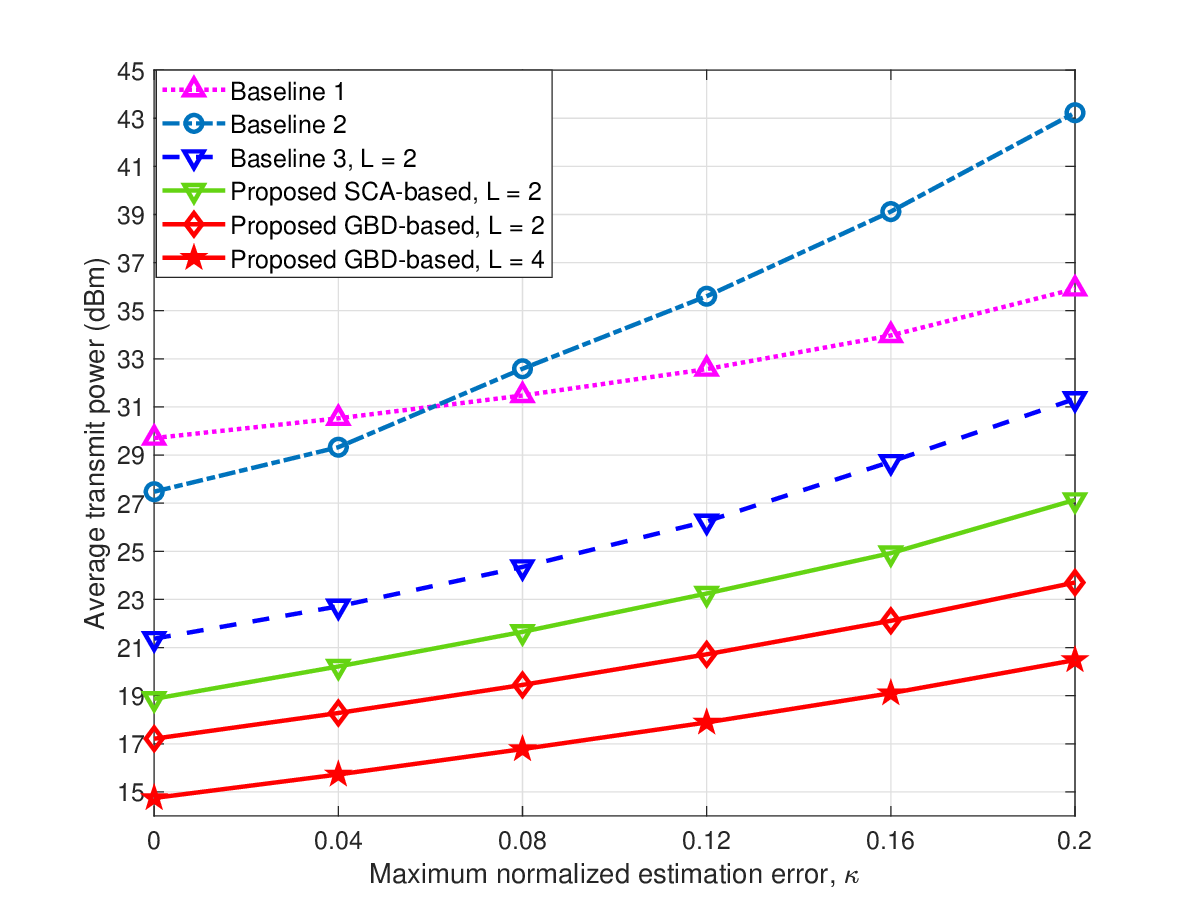}
	\caption{Average BS transmit power versus maximum normalized channel estimation error $\kappa$.}
	\label{fig::kappa}
\end{figure}
\vspace*{-4mm}\section{Conclusion}
In this paper, we investigated the optimal resource allocation design for IRS-assisted multiuser MISO wireless communication for the cases of perfect and imperfect CSI, respectively. Our goal was to minimize the total transmit power while ensuring the QoS requirements of the users. {\color{black}Since the intricate coupling between the discrete IRS phase shift matrix and the BS beamforming vectors presents an obstacle to efficient resource allocation algorithm design, we leveraged and recast a series of transformations introduced in \cite{6698281,zhang2008joint} and developed a novel optimization framework that converted the considered problems into more tractable MINLP problems for perfect and imperfect CSI, respectively. The proposed optimization framework facilitated optimal and suboptimal designs for perfect and imperfect CSI. In particular, to provide a performance upper bound for practical IRS-assisted communication systems, two globally optimal algorithms based on the GBD method were proposed for jointly designing the beamforming matrix at the BS and the discrete phase shifts at the IRS for perfect and imperfect CSI, respectively. To overcome the high complexity of the GBD-based schemes, exploiting the proposed optimization framework, we also proposed two corresponding SCA-based algorithms, which attained locally optimal solutions in polynomial time.}
Our simulation results verified the global optimality of the proposed GBD-based algorithms and demonstrated the effectiveness of all proposed algorithms. Moreover, we highlighted the robustness of the proposed GBD-based and SCA-based algorithms for imperfect CSI, ensuring reliable performance even under uncertain channel conditions. We note that the proposed GBD-based methods can serve as performance upper bounds when studying any suboptimal resource allocation algorithm for IRS-assisted wireless systems with discrete phase shifters. Furthermore, the proposed SCA-based algorithms outperform the state-of-the-art AO-based and IA-based methods and achieve an excellent performance for both perfect and imperfect CSI, particularly for IRSs of moderate-to-large sizes. 

\appendix[Proof of Theorem 1]\vspace*{-3mm}
    To start with, we note that proving Theorem 1 is equivalent to proving that for each given feasible $\bar{\mathbf{B}}=\bar{\mathbf{B}}_f$, we can obtain an optimal beamforming matrix $\mathbf{W}_k^{\mathrm{opt}}$ with $\mathrm{rank}(\mathbf{W}_k^{\mathrm{opt}})\leq 1$ by solving the following optimization problem
    \begin{eqnarray}
\label{Master_Problem_proof}
    &&\hspace*{-4mm}\underset{\substack{\widetilde{\mathbf{X}},\hat{\mathbf{S}},\hat{\mathbf{T}},\hat{\mathbf{W}}\\\mathbf{U},\mathbf{V},\mathbf{Y},\mathbf{q}}}{\mino}\hspace*{2mm}\sum_{k\in\mathcal{K}}\mathrm{Tr}(\mathbf{W}_k)\notag\\
    &&\hspace*{2mm}\mbox{s.t.}\hspace*{7mm} \widehat{\overline{\mbox{C1}}},\widehat{\overline{\mbox{C3a}}},\widehat{\overline{\mbox{C3b}}}, \widehat{\overline{\mbox{C3c}}}, \widehat{\overline{\mbox{C3d}}},\mbox{C8},
\end{eqnarray}
where binary matrix $\bar{\mathbf{B}}$ in $\widehat{\overline{\mbox{C3a}}},\widehat{\overline{\mbox{C3b}}}, \widehat{\overline{\mbox{C3c}}}, \widehat{\overline{\mbox{C3d}}}$ is set to $\bar{\mathbf{B}}_f$. 
The problem in \eqref{Master_Problem_proof} is convex w.r.t. \\$\hat{\mathbf{X}},\Bar{\mathbf{B}},\hat{\mathbf{S}},\hat{\mathbf{T}},\mathbf{U},\mathbf{V}$,$\mathbf{Y}$, $\mathbf{q}$ and Slater's condition is satisfied \footnote{Here, if the optimization problem in (65) is feasible, there exists a solution satisfying the strict feasibility of constraints $\widehat{\overline{\mbox{C1}}}$ and $\mbox{C8}$. Moreover, constraints $\widehat{\overline{\mbox{C3b}}}$ and $\widehat{\overline{\mbox{C3d}}}$ are linear inequality constraints and constraints $\widehat{\overline{\mbox{C3a}}}$ and $\widehat{\overline{\mbox{C3c}}}$ can be made strictly feasible by tuning the auxiliary matrices $\hat{\mathbf{S}},\hat{\mathbf{T}},\mathbf{U},$ and $\mathbf{V}$.}\cite{boyd2004convex,yu2021robust}. Thus, strong duality holds and the corresponding Lagrangian function is given by
    \begin{equation}
        \begin{aligned}
            \Bar{\mathcal{L}}=&\sum_{k\in\mathcal{K}}\left[\mathrm{Tr}(\mathbf{W}_k)-\mathrm{Tr}(\mathbf{Z}_k\mathbf{W}_k)+\mathrm{Tr}(\mathbf{A}_k\mathbf{G}_k^H\widetilde{\mathbf{X}}_k\mathbf{G}_k)-\mathrm{Tr}(\hat{\mathbf{Q}}_{k,22}\hat{\mathbf{T}}_k)-\mathrm{Tr}(\hat{\bm{\Xi}}_{k,22}\mathbf{V}_k)\right]\\
            &-2\mathrm{Re}\left\{\sum_{k\in\mathcal{K}}\left[\mathrm{Tr}(\hat{\bm{\Xi}}_{k,32}\mathbf{W}_k)+\mathrm{Tr}((\hat{\mathbf{Q}}_{k,32}+\hat{\bm{\Xi}}_{k,21}){\mathbf{Y}}_k)+\mathrm{Tr}(\hat{\mathbf{Q}}_{k,21}\hat{\mathbf{X}}_k)\right]\right\}+\nu,
        \end{aligned}
    \end{equation}
    where $\mathbf{Z}_k$ and $\mathbf{A}_k$ are the Lagrangian multiplier matrices associated with constraint $\mathbf{W}_k\succeq \mathbf{0}$ and $\widehat{\overline{\mbox{C1}}}$, respectively. Here, $\nu$ represents the collection of terms that do not involve optimization variables $\mathbf{W}_k$, $\mathbf{Y}_k$, $\hat{\mathbf{X}}_k$, $\hat{\mathbf{T}}_k$, and $\mathbf{V}_k$. 
    Moreover, since the problem in \eqref{Master_Problem_proof} corresponds to the primal problem in \eqref{Master_Problem_robust}, we use the same notations for $\hat{\mathbf{Q}}$ and $\hat{\bm{\Xi}}$ and their respective submatrices as in \eqref{Qi_Ei}. For ease of representation, we denote the optimal variables for \eqref{Ref_Problem_robust} with superscript "opt". The Karush–Kuhn–Tucker (KKT) conditions for \eqref{Ref_Problem_robust} without rank-one constraint C7 and without binary constraint C3b are given by
    \begin{eqnarray}
        &&\mathrm{K1:} \hspace*{2mm}\mathbf{Z}_k^{\mathrm{opt}}\succeq \mathbf{0},\mathbf{A}_k^{\mathrm{opt}}\succeq \mathbf{0},\notag\\
     &&\mathrm{K2.1:}\hspace*{2mm}\mathbf{Z}_k^{\mathrm{opt}}\mathbf{W}_k^{\mathrm{opt}}=\mathbf{0},\notag\\
        &&\mathrm{K2.2:}\hspace*{2mm}{\mathbf{A}}_{k}^{\mathrm{opt}}\left(\mathbf{P}_k-\mathbf{G}_k^H\widetilde{\mathbf{X}_k}\mathbf{G}_k\right)=\mathbf{0},\notag\\
       &&\mathrm{K2.3:}\hspace*{2mm}\hat{\mathbf{Q}}_{k,21}^{\mathrm{opt}}\bar{\mathbf{B}}_f\bm{\Theta}+\hat{\mathbf{Q}}_{k,22}^{\mathrm{opt}}{\mathbf{Y}}_k^{\mathrm{opt}}+(\hat{\mathbf{Q}}_{k,32}^{\mathrm{opt}})^H=\mathbf{0},\notag\\
       &&\mathrm{K2.4:}\hspace*{2mm}{\bm{\Xi}}_{k,21}^{\mathrm{opt}}\bar{\mathbf{B}}_f\bm{\Theta}+{\bm{\Xi}}_{k,22}^{\mathrm{opt}}{\mathbf{W}}_k^{\mathrm{opt}}+({\bm{\Xi}}_{k,32}^{\mathrm{opt}})^H=\mathbf{0},\notag\\           
        &&\mathrm{K3.1:}\hspace*{2mm}\triangledown_{\mathbf{W}_k} \Bar{\mathcal{L}}=\mathbf{I}_M-\mathbf{Z}_k^{\mathrm{opt}}-(\hat{\bm{\Xi}}_{k,32}^{\mathrm{opt}}+(\hat{\bm{\Xi}}_{k,32}^{\mathrm{opt}})^H)=\mathbf{0},\\
        &&\mathrm{K3.2:}\hspace*{2mm}\triangledown_{\mathbf{Y}_k} \Bar{\mathcal{L}}=\hat{\mathbf{Q}}_{k,32}^{\mathrm{opt}}+\hat{\bm{\Xi}}_{k,21}^{\mathrm{opt}}=\mathbf{0},\notag\\
        &&\mathrm{K3.3:}\hspace*{2mm}\triangledown_{\hat{\mathbf{T}}_k} \Bar{\mathcal{L}}=\hat{\mathbf{Q}}_{k,22}^{\mathrm{opt}}=\mathbf{0},\triangledown_{\mathbf{V}_k} \Bar{\mathcal{L}}=\hat{\bm{\Xi}}_{k,22}^{\mathrm{opt}}=\mathbf{0},\notag\\
        &&\mathrm{K3.4:}\hspace*{2mm}\triangledown_{\hat{\mathbf{X}}_k} \Bar{\mathcal{L}}=\sum_{k'\in\mathcal{K}\setminus k}\gamma_{k'}\mathbf{G}_{k'}\mathbf{A}_{k'}^{\mathrm{opt}}\mathbf{G}_{k'}-\mathbf{G}_k\mathbf{A}_k^{\mathrm{opt}}\mathbf{G}_k^H-(\mathbf{Q}_{k,21}^{\mathrm{opt}}+(\mathbf{Q}_{k,21}^{\mathrm{opt}})^H)=\mathbf{0},\quad\forall k\in\mathcal{K}.\notag
    \end{eqnarray}
    From K3.3, we obtain $\hat{\mathbf{Q}}_{k,22}^{\mathrm{opt}}=\mathbf{0}$ and $\hat{\bm{\Xi}}_{k,22}^{\mathrm{opt}}=\mathbf{0}$. By substituting $\hat{\mathbf{Q}}_{k,22}^{\mathrm{opt}}=\mathbf{0}$ and $\hat{\bm{\Xi}}_{k,22}^{\mathrm{opt}}=\mathbf{0}$ into K2.3 and K2.4, respectively, we can rewrite K2.3 and K2.4 as follows
    \begin{eqnarray}\label{Rank1_proof1}
        &&\mathrm{K2.3:} \hspace*{2mm}(\hat{\mathbf{Q}}_{k,32}^{\mathrm{opt}})^H=-\hat{\mathbf{Q}}_{k,21}^{\mathrm{opt}}\bar{\mathbf{B}}_f\bm{\Theta},\notag\\
       &&\mathrm{K2.4:}\hspace*{2mm}({\bm{\Xi}}_{k,32}^{\mathrm{opt}})^H=-{\bm{\Xi}}_{k,21}^{\mathrm{opt}}\bar{\mathbf{B}}_f\bm{\Theta}.
    \end{eqnarray}
    Then, we substitute \eqref{Rank1_proof1} into $\hat{\mathbf{Q}}_{k,32}^{\mathrm{opt}}+\hat{\bm{\Xi}}_{k,21}^{\mathrm{opt}}=\mathbf{0}$ and obtain
    \begin{equation}
    \hat{\bm{\Xi}}_{k,32}^{\mathrm{opt}}+\bar{\mathbf{B}}_f\bm{\Theta}\hat{\mathbf{Q}}_{k,21}^{\mathrm{opt}}(\bar{\mathbf{B}}_f\bm{\Theta})^H=\mathbf{0}.
    \end{equation}
    Using the above result, we can rewrite $\triangledown_{\hat{\mathbf{X}}_k} \Bar{\mathcal{L}}=\mathbf{0}$ as follows
    \begin{equation}
       \hat{\bm{\Xi}}_{k,32}^{\mathrm{opt}}+(\hat{\bm{\Xi}}_{k,32}^{\mathrm{opt}})^H=2\bar{\mathbf{B}}_f\bm{\Theta}\left(\mathbf{G}_k\mathbf{A}_k^{\mathrm{opt}}\mathbf{G}_k^H-\sum_{k'\in\mathcal{K}\setminus k}\gamma_{k'}\mathbf{G}_{k'}\mathbf{A}_{k'}^{\mathrm{opt}}\mathbf{G}_{k'}\right)(\bar{\mathbf{B}}_f\bm{\Theta})^H.
    \end{equation}
    Hence, $\triangledown_{\mathbf{W}_k} \Bar{\mathcal{L}}=\mathbf{0}$ can be recast as follows
    \begin{equation}\label{Zk_opt}
        \mathbf{Z}_k^{\mathrm{opt}}=\mathbf{I}_M-2\bar{\mathbf{B}}_f\bm{\Theta}\left(\mathbf{G}_k\mathbf{A}_k^{\mathrm{opt}}\mathbf{G}_k^H-\sum_{k'\in\mathcal{K}\setminus k}\gamma_{k'}\mathbf{G}_{k'}\mathbf{A}_{k'}^{\mathrm{opt}}\mathbf{G}_{k'}\right)(\bar{\mathbf{B}}_f\bm{\Theta})^H.
    \end{equation}
    Then, we post-multiply \eqref{Zk_opt} by $\mathbf{W}_k^{\mathrm{opt}}$, i.e., 
    \begin{equation}
        \begin{aligned}
            2\bar{\mathbf{B}}_f\bm{\Theta}\mathbf{G}_k\mathbf{A}_k^{\mathrm{opt}}\mathbf{G}_k^H(\bar{\mathbf{B}}_f\bm{\Theta})^H\mathbf{W}_k^{\mathrm{opt}}=\left(\mathbf{I}_M+2\bar{\mathbf{B}}_f\bm{\Theta}\sum_{k'\in\mathcal{K}\setminus k}\gamma_{k'}\mathbf{G}_{k'}\mathbf{A}_{k'}^{\mathrm{opt}}\mathbf{G}_{k'}(\bar{\mathbf{B}}_f\bm{\Theta})^H\right)\mathbf{W}_k^{\mathrm{opt}}
        \end{aligned}
    \end{equation}
    Since $\mathbf{I}_M+2\bar{\mathbf{B}}_f\bm{\Theta}\sum_{k'\in\mathcal{K}\setminus k}\gamma_{k'}\mathbf{G}_{k'}\mathbf{A}_{k'}^{\mathrm{opt}}\mathbf{G}_{k'}(\bar{\mathbf{B}}_f\bm{\Theta})^H$ is full-rank, we obtain the following rank inequality
    \begin{equation}
    \begin{aligned}
              \mathrm{rank}(\mathbf{W}_k^{\mathrm{opt}})&=\mathrm{rank}\left(\left(\mathbf{I}_M+2\bar{\mathbf{B}}_f\bm{\Theta}\sum_{k'\in\mathcal{K}\setminus k}\gamma_{k'}\mathbf{G}_{k'}\mathbf{A}_{k'}^{\mathrm{opt}}\mathbf{G}_{k'}(\bar{\mathbf{B}}_f\bm{\Theta})^H\right)\mathbf{W}_k^{\mathrm{opt}}\right)\\
              &=\mathrm{rank}\left(\bar{\mathbf{B}}_f\bm{\Theta}\mathbf{G}_k\mathbf{A}_k^{\mathrm{opt}}\mathbf{G}_k^H(\bar{\mathbf{B}}_f\bm{\Theta})^H\mathbf{W}_k^{\mathrm{opt}}\right)\\
              &\leq\min\left\{\mathrm{rank}(\mathbf{W}_k^{\mathrm{opt}}),\mathrm{rank}\left(\mathbf{G}_k\mathbf{A}_k^{\mathrm{opt}}\mathbf{G}_k^H\right)\right\}.
    \end{aligned}
    \end{equation}
    Thus, we can prove $\mathrm{rank}(\mathbf{W}_k^{\mathrm{opt}})\leq 1$ by showing that $\mathrm{rank}\left(\mathbf{G}_k\mathbf{A}_k^{\mathrm{opt}}\mathbf{G}_k^H\right)\leq1$. First, we obtain the following equalities \cite{alavi2017robust}
    \begin{equation}
    \begin{aligned}
        \mathbf{G}_k[\mathbf{I}_{M(N+1)},\mathbf{0}]^T&=[\mathbf{I}_{M(N+1)}\quad \Bar{\mathbf{g}}_k][\mathbf{I}_{M(N+1)},\mathbf{0}]^T=\mathbf{I}_{M(N+1)},\\
        \mathbf{P}_k[\mathbf{I}_{M(N+1)},\mathbf{0}]^T&=\begin{bmatrix}
        q_k\mathbf{I}_{M(N+1)} & \mathbf{0} \\
        \mathbf{0} & -q_k\epsilon_k^2-\sigma_k^2\gamma_k
    \end{bmatrix}\begin{bmatrix}
        \mathbf{I}_{M(N+1)}  \\
        \mathbf{0} 
    \end{bmatrix}=\begin{bmatrix}
        q_k\mathbf{I}_{M(N+1)}  \\
        \mathbf{0} 
    \end{bmatrix}=q_k\mathbf{G}_k^H-[\mathbf{0},q_k\bar{\mathbf{g}}_k]^T.
    \end{aligned}
    \end{equation}
    Then, we pre-multiply $\mathbf{G}_k$ and post-multiply $[\mathbf{I},\mathbf{0}]^T$ by K2.2, respectively, and rewrite K2.2 as follows
    \begin{equation}
        \begin{aligned}
            \mathbf{G}_k{\mathbf{A}}_{k}^{\mathrm{opt}}\left(q_k\mathbf{G}_k^H-[\mathbf{0},q_k\bar{\mathbf{g}}_k]^T-\mathbf{G}_k^H\widetilde{\mathbf{X}}_k\right)&=\mathbf{0},\\
            \mathbf{G}_k{\mathbf{A}}_{k}^{\mathrm{opt}}\mathbf{G}_k^H(q_k\mathbf{I}_{M(N+1)}-\widetilde{\mathbf{X}}_k)&=\mathbf{G}_k{\mathbf{A}}_{k}^{\mathrm{opt}}[\mathbf{0},q_k\bar{\mathbf{g}}_k]^T.
        \end{aligned}
    \end{equation}
    According to constraint $\widehat{\widebar{\mbox{C1}}}$, we have
    \begin{equation}
        \begin{bmatrix}
        q_k\mathbf{I}_{M(N+1)}-\widetilde{\mathbf{X}}_k & -\widetilde{\mathbf{X}}_k \bar{\mathbf{g}}_k\\
        -\bar{\mathbf{g}}_k^H\widetilde{\mathbf{X}}_k & -q_k\epsilon_k^2-\sigma_k^2\gamma_k-\bar{\mathbf{g}}_k^H\widetilde{\mathbf{X}}_k\bar{\mathbf{g}}_k
    \end{bmatrix}\succeq\mathbf{0}.
    \end{equation}
    Hence, we can claim $q_k\mathbf{I}_{M(N+1)}-\widetilde{\mathbf{X}}_k\succeq\mathbf{0}$ is non-singular \cite{alavi2017robust}. Since multiplying by a non-singular matrix will not alter the rank of a matrix, the following rank property holds
\begin{equation}
\begin{aligned}
    \mathrm{rank}(\mathbf{G}_k{\mathbf{A}}_{k}^{\mathrm{opt}}\mathbf{G}_k^H)= \mathrm{rank}(\mathbf{G}_k{\mathbf{A}}_{k}^{\mathrm{opt}}[\mathbf{0},\bar{\mathbf{g}}_k]^T)\leq 1
\end{aligned}
    \end{equation}
    This completes the proof of Theorem 1.

\bibliographystyle{IEEEtran}
\vspace{-4mm}\bibliography{Reference_List}

\begin{thebibliography}{10}
\providecommand{\url}[1]{#1}
\csname url@samestyle\endcsname
\providecommand{\newblock}{\relax}
\providecommand{\bibinfo}[2]{#2}
\providecommand{\BIBentrySTDinterwordspacing}{\spaceskip=0pt\relax}
\providecommand{\BIBentryALTinterwordstretchfactor}{4}
\providecommand{\BIBentryALTinterwordspacing}{\spaceskip=\fontdimen2\font plus
\BIBentryALTinterwordstretchfactor\fontdimen3\font minus
  \fontdimen4\font\relax}
\providecommand{\BIBforeignlanguage}[2]{{%
\expandafter\ifx\csname l@#1\endcsname\relax
\typeout{** WARNING: IEEEtran.bst: No hyphenation pattern has been}%
\typeout{** loaded for the language `#1'. Using the pattern for}%
\typeout{** the default language instead.}%
\else
\language=\csname l@#1\endcsname
\fi
#2}}
\providecommand{\BIBdecl}{\relax}
\BIBdecl

\bibitem{wu2023globally}
Y.~Wu, D.~Xu, D.~W.~K. Ng, R.~Schober, and W.~Gerstacker, ``Globally optimal
  resource allocation design for {IRS}-assisted multiuser networks with
  discrete phase shifts,'' in \emph{Proc. IEEE Int. Conf. Commun. (ICC)}, Rome,
  Italy, June. 2023, pp. 1--6.

\bibitem{wu2019intelligent}
Q.~Wu and R.~Zhang, ``Intelligent reflecting surface enhanced wireless network
  via joint active and passive beamforming,'' \emph{IEEE Trans. Wireless
  Commun.}, vol.~18, no.~11, pp. 5394--5409, 2019.

\bibitem{huang2019reconfigurable}
C.~{Huang} \emph{et~al.}, ``Reconfigurable intelligent surfaces for energy
  efficiency in wireless communication,'' \emph{IEEE Trans. Wireless Commun.},
  vol.~18, no.~8, pp. 4157--4170, 2019.

\bibitem{wu2019beamforming}
Q.~{Wu} and R.~{Zhang}, ``Beamforming optimization for wireless network aided
  by intelligent reflecting surface with discrete phase shifts,'' \emph{IEEE
  Trans. Commun.}, vol.~68, no.~3, pp. 1838--1851, 2019.

\bibitem{xu2020resource}
D.~Xu, X.~Yu, Y.~Sun, D.~W.~K. Ng, and R.~Schober, ``Resource allocation for
  {IRS}-assisted full-duplex cognitive radio systems,'' \emph{IEEE Trans.
  Commun.}, vol.~68, no.~12, pp. 7376--7394, 2020.

\bibitem{wu2021intelligent}
Q.~Wu, S.~Zhang, B.~Zheng, C.~You, and R.~Zhang, ``Intelligent reflecting
  surface-aided wireless communications: A tutorial,'' \emph{IEEE Trans.
  Commun.}, vol.~69, no.~5, pp. 3313--3351, 2021.

\bibitem{yu2021robust}
X.~Yu, D.~Xu, D.~W.~K. Ng, and R.~Schober, ``{IRS}-assisted green communication
  systems: Provable convergence and robust optimization,'' \emph{IEEE Trans.
  Commun.}, vol.~69, no.~9, pp. 6313--6329, 2021.

\bibitem{xu2021resource}
D.~Xu, X.~Yu, D.~W.~K. Ng, and R.~Schober, ``Resource allocation for active
  {IRS}-assisted multiuser communication systems,'' in \emph{Proc. Asilomar
  Conference on Signals, Systems, and Computers}, 2021, pp. 113--119.

\bibitem{yu2020robust}
X.~Yu, D.~Xu, Y.~Sun, D.~W.~K. Ng, and R.~Schober, ``Robust and secure wireless
  communications via intelligent reflecting surfaces,'' \emph{IEEE J. Sel.
  Areas Commun.}, vol.~38, no.~11, pp. 2637--2652, 2020.

\bibitem{xu2022optimal}
D.~Xu, V.~Jamali, X.~Yu, D.~W.~K. Ng, and R.~Schober, ``Optimal resource
  allocation design for large {IRS}-assisted {SWIPT} systems: A scalable
  optimization framework,'' \emph{IEEE Trans. Commun.}, vol.~70, no.~2, pp.
  1423--1441, 2022.

\bibitem{9014322}
X.~{Yu}, D.~{Xu}, and R.~{Schober}, ``Enabling secure wireless communications
  via intelligent reflecting surfaces,'' in \emph{Proc. IEEE Global Commun.
  Conf. (GLOBECOM)}, Waikoloa, HI, USA, Dec. 2019, pp. 1--6.

\bibitem{shi2022multiuser}
M.~Shi, X.~Li, T.~Fan, J.~Liu, and S.~Lv, ``Multiuser beamforming optimization
  for {IRS}-aided systems with discrete phase shifts,'' \emph{IET Commun.}, pp.
  1523--1530, Aug. 2022.

\bibitem{bezdek2002some}
J.~C. Bezdek and R.~J. Hathaway, ``Some notes on alternating optimization,'' in
  \emph{Proc. AFSS Int. Conf. on Fuzzy Systems}.\hskip 1em plus 0.5em minus
  0.4em\relax Springer, 2002, pp. 288--300.

\bibitem{zhou2020framework}
G.~Zhou \emph{et~al.}, ``A framework of robust transmission design for
  {IRS}-aided {MISO} communications with imperfect cascaded channels,''
  \emph{IEEE Trans. Signal Process.}, vol.~68, pp. 5092--5106, 2020.

\bibitem{6698281}
U.~Rashid \emph{et~al.}, ``Joint optimization of source precoding and relay
  beamforming in wireless {MIMO} relay networks,'' \emph{IEEE Trans. Commun.},
  vol.~62, no.~2, pp. 488--499, 2014.

\bibitem{zhang2008joint}
R.~{Zhang}, C.~{Chai}, and Y.-C. {Liang}, ``Joint beamforming and power control
  for multiantenna relay broadcast channel with {QoS} constraints,'' \emph{IEEE
  Trans. Signal Process.}, vol.~57, no.~2, pp. 726--737, 2008\color{black}.

\bibitem{abeywickrama2020intelligent}
S.~\color{black} Abeywickrama, R.~Zhang, Q.~Wu, and C.~Yuen, ``Intelligent
  reflecting surface: Practical phase shift model and beamforming
  optimization,'' \emph{IEEE Trans. Commun.}, vol.~68, no.~9, pp. 5849--5863,
  2020\color{black}.

\bibitem{hu2021robust}
S.~Hu, Z.~Wei, Y.~Cai, C.~Liu, D.~W.~K. Ng, and J.~Yuan, ``Robust and secure
  sum-rate maximization for multiuser {MISO} downlink systems with
  self-sustainable {IRS},'' \emph{IEEE Trans. Commun.}, vol.~69, no.~10, pp.
  7032--7049, 2021.

\bibitem{geoffrion1972generalized}
A.~M. Geoffrion, ``Generalized {Benders} decomposition,'' \emph{Journal of
  Optimization Theory and Applications}, vol.~10, no.~4, pp. 237--260, 1972.

\bibitem{sahinidis1991convergence}
N.~Sahinidis and I.~E. Grossmann, ``Convergence properties of generalized
  {Benders} decomposition,'' \emph{Computers \& Chemical Engineering}, vol.~15,
  no.~7, pp. 481--491, 1991.

\bibitem{grant2008cvx}
M.~{Grant} and S.~{Boyd}, ``{CVX}: Matlab software for disciplined convex
  programming, version 2.1,'' \emph{http://cvxr.com/cvx}, Jan. 2020.

\bibitem{ng2015secure}
D.~W.~K. Ng and R.~Schober, ``Secure and green {SWIPT} in distributed antenna
  networks with limited backhaul capacity,'' \emph{IEEE Trans. Wireless
  Commun.}, vol.~14, no.~9, pp. 5082--5097, 2015.

\bibitem{bomze2010interior}
I.~M. Bomze \emph{et~al.}, ``Interior point methods for nonlinear
  optimization,'' \emph{Nonlinear Optimization: Lectures given at the CIME
  Summer School held in Cetraro, Italy, July 1-7, 2007}, pp. 215--276,
  2010\color{black}.

\bibitem{mosek}
{\color{black}MOSEK ApS}, \emph{{\color{black}The MOSEK {O}ptimization
  {T}oolbox for MATLAB {M}anual. Version 10.0.}},
  \color{black}2022\color{black}.

\bibitem{le2012exact}
L.~Thi \emph{et~al.}, ``Exact penalty and error bounds in {DC} programming,''
  \emph{Journal of Global Optimization}, vol.~52, no.~3, pp. 509--535, 2012.

\bibitem{razaviyayn2013unified}
M.~Razaviyayn, M.~Hong, and Z.-Q. Luo, ``A unified convergence analysis of
  block successive minimization methods for nonsmooth optimization,''
  \emph{SIAM Journal on Optimization}, vol.~23, no.~2, pp. 1126--1153,
  2013\color{black}.

\bibitem{pham2010efficient}
T.~\color{black} Pham~Dinh, N.~Nguyen~Canh, and H.~A. Le~Thi, ``An efficient
  combined {DCA} and {B\&B} using {DC/SDP} relaxation for globally solving
  binary quadratic programs,'' \emph{Journal of Global Optimization}, vol.~48,
  pp. 595--632, 2010\color{black}.

\bibitem{you2020channel}
C.~You, B.~Zheng, and R.~Zhang, ``Channel estimation and passive beamforming
  for intelligent reflecting surface: Discrete phase shift and progressive
  refinement,'' \emph{IEEE J. Sel. Areas Commun.}, vol.~38, no.~11, pp.
  2604--2620, 2020.

\bibitem{lin2022channel}
T.~Lin, X.~Yu, Y.~Zhu, and R.~Schober, ``Channel estimation for {IRS}-assisted
  millimeter-wave {MIMO} systems: Sparsity-inspired approaches,'' \emph{IEEE
  Trans. Commun.}, vol.~70, no.~6, pp. 4078--4092, 2022.

\bibitem{ng2014robust}
D.~W.~K. Ng, E.~S. Lo, and R.~Schober, ``Robust beamforming for secure
  communication in systems with wireless information and power transfer,''
  \emph{IEEE Trans. Wirel. Commun.}, vol.~13, no.~8, pp. 4599--4615, 2014.

\bibitem{xu2022robust}
D.~Xu, X.~Yu, D.~W.~K. Ng, A.~Schmeink, and R.~Schober, ``Robust and secure
  resource allocation for {ISAC} systems: A novel optimization framework for
  variable-length snapshots,'' \emph{IEEE Trans. Commun.}, vol.~70, no.~12, pp.
  8196--8214, 2022.

\bibitem{boyd2004convex}
S.~P. Boyd and L.~Vandenberghe, \emph{Convex {O}ptimization}.\hskip 1em plus
  0.5em minus 0.4em\relax Cambridge {U}niversity {P}ress, 2004.

\bibitem{alavi2017robust}
F.~Alavi, K.~Cumanan, Z.~Ding, and A.~G. Burr, ``Robust beamforming techniques
  for non-orthogonal multiple access systems with bounded channel
  uncertainties,'' \emph{IEEE Commun. Lett.}, vol.~21, no.~9, pp. 2033--2036,
  2017.

\end{thebibliography}
\end{document}